%% file: main.tex
\documentclass{article}
\usepackage{graphicx} 
\usepackage[pagebackref, colorlinks = true, linkcolor = blue, urlcolor  = blue, citecolor = red, bookmarks, hypertexnames=false]{hyperref}
\usepackage{authblk}

\usepackage[a4paper, margin=1.3in]{geometry}

\usepackage{amssymb}
\usepackage{amsmath}
\usepackage{amsthm}
\usepackage{dsfont}
\usepackage{fdsymbol}
\usepackage{tikz}
\usepackage{tikz-cd}
\usepackage{mathtools}
\usepackage{physics}
\usepackage{braket}
\usepackage{xcolor}
\usepackage[capitalise]{cleveref}
\usepackage[utf8]{inputenc}
\usepackage[most]{tcolorbox}
\usepackage[skip=5pt]{caption}

\usetikzlibrary{calc}
\usetikzlibrary{fadings}
\usetikzlibrary{shapes.misc, shapes.geometric}
\usetikzlibrary{decorations.pathreplacing,decorations.pathmorphing,angles,quotes}

\newtheoremstyle{newdefinition}{}{}{\normalfont}{}{\bfseries}{}{\newline}
{\thmname{#1} \thmnumber{#2}\thmnote{ (#3)}}

\newtheoremstyle{newplain}{}{}{\itshape}{}{\bfseries}{}{1em}
{\thmname{#1} \thmnumber{#2}\thmnote{ (#3)}}

\newtheoremstyle{newremark}{}{}{\normalfont}{}{\bfseries}{}{1em}
{\thmname{#1}}

\theoremstyle{newdefinition}
\newtheorem{definition}{Definition}[section]

\theoremstyle{newplain}
\newtheorem{theorem}[definition]{Theorem}
\newtheorem{lemma}[definition]{Lemma}
\newtheorem{proposition}[definition]{Proposition}
\newtheorem{corollary}[definition]{Corollary}

\newtheorem{remark}[definition]{Remark}

\DeclareMathOperator{\Z}{\mathbb{Z}}
\DeclareMathOperator{\C}{\mathbb{C}}
\DeclareMathOperator{\E}{\mathbb{E}}

\DeclareMathOperator{\F}{\mathbb{F}}

\DeclareMathOperator{\HH}{\mathcal{H}}
\DeclareMathOperator{\KK}{\mathcal{K}}

\DeclareMathOperator{\BB}{\mathcal{B}}
\DeclareMathOperator{\AAA}{\mathcal{A}}

\DeclareMathOperator{\FF}{\mathcal{F}}
\DeclareMathOperator{\LL}{\mathcal{L}}

\DeclareMathOperator{\NN}{\mathcal{N}}
\DeclareMathOperator{\EE}{\mathcal{E}}
\DeclareMathOperator{\PP}{\mathcal{P}}
\DeclareMathOperator{\TT}{\mathcal{T}}
\DeclareMathOperator{\RR}{\mathcal{R}}
\DeclareMathOperator{\SSS}{\mathcal{S}}

\DeclareMathOperator{\diam}{diam}
\DeclareMathOperator{\dist}{dist}
\DeclareMathOperator{\Img}{Im}

\DeclareMathOperator{\gap}{gap}
\newcommand{\identity}{\ensuremath{\mathds{1}}}
\def\lsym{\Lambda} 
\DeclareMathOperator{\famsym}{{\mathcal{F}}} 
\DeclareMathOperator{\strset}{{\mathbb{S}}}
\DeclareMathOperator{\plqset}{{\mathbb{P}}}
\newcommand{\bdy}{\partial}
\newcommand{\cobdy}{d}

\def\str{\medwhitestar}
\def\plq{\square}
\def\strplq{\boxplus}
\def\wildcard{\sharp}

\numberwithin{equation}{section}

\title{Modified logarithmic Sobolev inequalities for CSS codes}

\author[1,2]{Sebastian Stengele}
\author[3,4]{\'Angela Capel}
\author[5]{Li Gao}
\author[6]{Angelo Lucia}
\author[7,8]{David Pérez-García}
\author[9]{Antonio Pérez-Hernández}
\author[10]{Cambyse Rouzé}
\author[1,2]{Simone Warzel}

\affil[1]{\small Departments of Mathematics and Physics, TU M\"{u}nchen, 85747 Garching, Germany} 
\affil[2]{\small Munich Center for Quantum Science and Technology,
80799 M\"{u}nchen, Germany}
\affil[3]{\small Department of Applied Mathematics and Theoretical Physics, University of Cambridge, Wilberforce Road, Cambridge, CB3 0WA, United Kingdom}
\affil[4]{\small Fachbereich Mathematik, Universität Tübingen, 72076 Tübingen, Germany}
\affil[5]{\small School of Mathematics and Statistics, Wuhan University, Wuhan, 430072, China}
\affil[6]{\small Dipartimento di Matematica, Politecnico di Milano, 20133 Milano, Italy}
\affil[7]{\small Departamento de An\'{a}lisis Matemático y Matemática Aplicada, Universidad Complutense de Madrid, 28040 Madrid, Spain}
\affil[8]{\small Instituto de Ciencias Matemáticas, 28049 Madrid, Spain}
\affil[9]{\small Departamento de Matem\'{a}tica Aplicada I, Escuela T\'{e}cnica Superior de Ingenieros Industriales, Universidad Nacional de Educación a Distancia, 28040 Madrid, Spain}
\affil[10]{\small Inria, Télécom Paris - LTCI, Institut Polytechnique de Paris, 91120 Palaiseau, France}

\date{October 16, 2025}

\begin{document}

\maketitle

\begin{abstract}
We consider the class of Davies quantum semigroups modelling thermalization for translation-invariant Calderbank-Shor-Steane (CSS) codes in $D$ dimensions. We prove that conditions of Dobrushin-Shlosman-type on the quantum Gibbs state imply a modified logarithmic Sobolev inequality with a constant that is uniform in the system’s size.
This is accomplished by generalizing parts of the classical results on thermalization by Stroock, Zegarlinski, Martinelli, and Olivieri to the CSS quantum setting. The results in particular imply the rapid thermalization at any positive temperature of the toric code in $2$D and the star part of the toric code in $3$D, implying a rapid loss of stored quantum information for these models. 
\end{abstract}
\bigskip
\tableofcontents

\section{Introduction and results}

The interplay between static and dynamical properties of classical Markovian evolutions on lattice spin systems is by now broadly understood. Logarithmic Sobolev inequalities (LSI) have been identified as a powerful tool for bounding the mixing times of classical Markov processes such as the Glauber dynamics on spin configurations. In seminal works, Stroock and Zegarlinski \cite{stroock_equivalencelogarithmicSobolev_1992} and Martinelli and Olivieri \cite{martinelli_ApproachEquilibriumGlauber_1994, martinelli_ApproachEquilibriumGlauber_1994a} established the equivalence of a LSI for Glauber dynamics to a Dobrushin-Shlosmann \cite{dobrushin_CompletelyAnalyticalInteractions_1987} high-temperature condition on the Gibbs equilibrium state.
It is an interesting and long-standing question whether any of these results can be generalized to quantum spin systems.

While quantum generalizations of LSI have found interest in the context of the theory of non-commutative integration and hypercontractive semigroups on $ C^* $-algebras \cite{carlen_Optimalhypercontractivityfermi_1993,olkiewicz_HypercontractivitynoncommutativeL_1999,bardet_HypercontractivityLogarithmicSobolev_2022}, more versatile from the point of view of thermalization when described by a Markovian quantum semigroup are modified logarithmic Sobolev inequalities (MLSIs) \cite{kastoryano_QuantumLogarithmicSobolev_2013}. They estimate the relative entropy between an initial quantum state and the dynamics’ fixed point in terms of the entropy production \cite{spohn_Entropyproductionquantum_1978}. For Davies semigroups \cite{davies_Generatorsdynamicalsemigroups_1979}, which constitute a well-established model for thermal noise in the weak coupling limit and are generated by a Lindbladian satisfying a quantum version of detailed balance, the unique fixed point is by construction the Gibbs state, and an MLSI then allows to deduce tight bounds on the mixing time.

Such MLSIs have have been established in limited settings only: for the heat-bath dynamics of the generalized depolarizing semigroup~\cite{capel_Quantumconditionalrelative_2018,beigi_Quantumreversehypercontractivity_2020}, and that of specific $1$D systems~\cite{bardet_ModifiedLogarithmicSobolev_2021}, for Schmidt generators of 2-local, commuting Hamiltonians~\cite{capel_ModifiedLogarithmicSobolev_2021}, for the Davies dynamics of  $k$-local, $1$D Hamiltonians at any positive temperature~\cite{bardet_RapidThermalizationSpin_2023,bardet_EntropyDecayDavies_2024}, that of any dimensional 2-local, commuting Hamiltonians at high enough temperature~\cite{kochanowski_RapidThermalizationDissipative_2025}, and for the Davies dynamics of non-interacting Hamiltonians conjugated with IQP circuits  \cite{bergamaschi_QuantumComputationalAdvantage_2024}.

In this work, we generalize the classical arguments of \cite{stroock_equivalencelogarithmicSobolev_1992,martinelli_ApproachEquilibriumGlauber_1994, martinelli_ApproachEquilibriumGlauber_1994a} to an open system's quantum dynamics, which corresponds to an arbitrary Calderbank-Shor-Steane (CSS) code~\cite{calderbank_Goodquantumerrorcorrecting_1996,steane_Simplequantumerrorcorrecting_1996} on a $ D$-dimensional lattice. This is a broad class of quantum codes, which enable some form of error correction and which cover Kitaev's toric or the surface code as its most prominent examples \cite{nielsen_QuantumComputationQuantum_2012,gottesman_introductionquantumerror_2010,kitaev_Faulttolerantquantumcomputation_2003}. We prove that a certain decay of correlations of the Gibbs state, analogous to the Dobrushin-Shlosman so-called high-temperature condition (which in some instances like the $2$D toric code holds up zero temperature), implies an MLSI for the Davies dynamics associated with the CSS code. While the two constituents of the Davies dynamics bear some resemblance to a classical Glauber dynamics (they are in fact identical on suitably chosen diagonals, i.e.\ classical states), the Davies semigroups studied here are nonetheless genuine Markovian quantum dynamics generated by local Lindbladians. Key aspects of the classical proofs~\cite{stroock_equivalencelogarithmicSobolev_1992,martinelli_ApproachEquilibriumGlauber_1994, martinelli_ApproachEquilibriumGlauber_1994a}, in particular those related to probabilistic conditioning, have no straightforward quantum analogue. A novel ingredient will be an explicit expression for the Davies conditional expectations, that is, the projectors onto the fixed points of the local dynamics. 
Similarly to the classical case, our proof is also based on a multiscale argument, reducing an MLSI from a larger to a smaller region.

\subsection{Families of CSS Hamiltonians} 
We consider families of CSS Hamiltonians on a $ D$-dimensional lattice.
Each unit cell of the lattice is assumed to carry a fixed, finite number of qubits, stars, and plaquettes, and each qubit, star, and plaquette is associated with a unique unit cell. See \cref{fig:code_examples} and \cref{fig:surface_code_boundary} for examples of CSS codes and unit cells, respectively. 
Depending on the model, the qubits might be placed, for example, on edges or corners of a hypercubic lattice. Our results, however, do not rely on any special placement and hold for general unit cells.
Our main results will be formulated for fixed, rectangular subsets $\lsym$ of the lattice, where rectangles are defined in terms of unit cells, with bounds which hold for any set in an increasing family $\famsym = \{\lsym_1, \lsym_2, \ldots\}$ of rectangles.

The distance of two qubits $v,v'\in \lsym$ on unit cells $z,z'\in \Z^D$ is the maximum distance between their unit cells, $$\dist(v,v')\coloneqq\max_{i=1,\ldots D}|z_i-z_i'| . $$
In particular, their distance is $0$ if $v$ and $v'$ lie in the same unit cell. Following standard conventions, for any two sets $U,V$, we also denote by $\operatorname{dist}(U,V)$ the minimum distance between any two elements $v\in U,v'\in V$, and $\operatorname{dist}(v,V)=\dist(\{v\},V)$. The diameter of a set $U$ of qubits is defined as $\max_{v,v'\in U}\dist(v,v')$.

For a fixed and finite $ \lsym $, the CSS Hamiltonian acts on the tensor-product Hilbert space 
 \begin{equation*}
      \HH_{\lsym} = \bigotimes_{v\in \lsym} \C^2 .
  \end{equation*}
We denote the three Pauli operators associated with the qubit at $ v \in \lsym $ by $ X_v, Y_v$ and $ Z_v $. 
CSS Hamiltonians are composed of two sets of commuting interactions. Borrowing the nomenclature of the most prominent example, Kitaev's toric code \cite{kitaev_Faulttolerantquantumcomputation_2003}, the pair of interactions is encoded in a set of stars $  \strset_{\lsym} $, and a set of plaquettes $\plqset_{\lsym} $, which are assumed to obey the translation symmetry of the lattice. For simplicity, we also assume the vertex support of any star or plaquette to be of a diameter of at most $1$. That is, any two qubits connected by a star or plaquette are at most in the next lattice cell, i.e.\ at distance at most $1$. 
These interaction rules give rise to the star and plaquette interaction operators, 
\begin{equation}\label{eq:operatordef}
 A_s\coloneqq \bigotimes_{v\in ds}X_v
 \quad \mathrm{and} \quad
 B_p \coloneqq \bigotimes_{v\in \bdy p}Z_v
\end{equation}
indexed by $ s \in  \strset_{\lsym} $, and $ p \in \plqset_{\lsym} $. The tensor products in~\eqref{eq:operatordef} are over the vertices associated with a star and plaquette, respectively.  
\cref{fig:code_examples} illustrates these interactions in five examples. We use the symbols $\cobdy$ and $\bdy$ in reminiscence of the equivalence of CSS codes to (co-)chain complexes, see e.g. \cite{breuckmann_PhDThesisHomological_2018}.
In order to relate to this homological point of view, one lifts the maps $\bdy$ and $\cobdy$ to $\F_2$-linear maps from subsets of stars $ 2^{\strset_{\lsym}} $ or plaquettes $ 2^{\plqset_{\lsym}}$ to the subsets $ 2^\lsym $ of classical Ising spin configurations,  
\begin{equation*}
  \begin{tikzcd}
    2^{\strset_{\lsym}}
    \arrow[r, "\cobdy", shift left]
    & 2^{\lsym}
    \arrow[r, "\cobdy", shift left]
    \arrow[l, "\bdy", shift left]
    & 2^{\plqset_{\lsym}} \arrow[l, "\bdy", shift left]
  \end{tikzcd} 
\end{equation*}
where, by slight abuse of notation, we also identify the transpose $\bdy^T=\cobdy$. In particular, we use $\bdy v$ and $\cobdy v$ to denote the set of stars and plaquettes connected to a single qubit $v$, respectively. 
In this language, the relations $\cobdy \circ\cobdy = 0$ and $\bdy \circ \bdy =0 $, encode a key property of the CSS code: the star  and plaquette operators commute
\begin{equation}
\label{eq:commrule}
[ A_s , B_p ] = 0 .
\end{equation}
Our subsequent analysis will not rely on any deep insight from homology on how to construct CSS codes, but rather only on the fact that the commutation rule~\eqref{eq:commrule} is satisfied for the operators in~\eqref{eq:operatordef}.

\cref{fig:code_examples} lists some examples of CSS codes:
\begin{itemize}
    \item The rotated surface code has qubits on the vertices of a square lattice. Stars and plaquettes are placed on alternating faces in a checkerboard way. On a finite lattice, we place boundary conditions by truncating the interactions.
    \item The $2$D toric code has qubits placed on the edges of a square lattice. Stars are placed on vertices and plaquettes on faces. They commute since they overlap in two or zero qubits. The unit cell is given by a vertical and a horizontal qubit forming the intersection of a star and a plaquette.
    \item A straightforward generalization of the $2$D toric code are tessellation models. Given any periodic tiling of the plane, associating qubits to edges, stars to vertices, and plaquettes to faces defines a valid CSS code \cite[Theorem 2.10]{breuckmann_PhDThesisHomological_2018}. We furthermore require that the unit cells of these models are both star- and plaquette-connected. 
    \item The $3$D toric code is defined on the cubic lattice, with stars on vertices, qubits on edges, and plaquettes on faces. The unit cell is given by three qubits connected at a star and pointing in orthogonal directions.
    \item The $3$D toric code also admits a straightforward generalisation to periodic tessellations, which define CSS codes \cite[Theorem 2.10]{breuckmann_PhDThesisHomological_2018}. We require that the unit cells of these models are star-connected. 
\end{itemize}

\begin{figure}
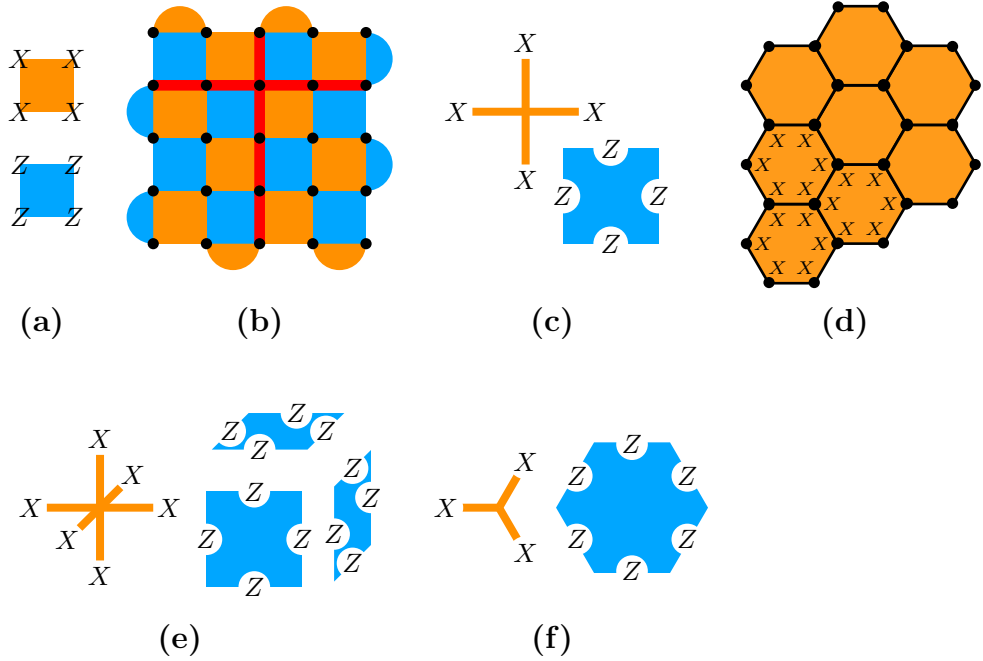

  \centering
  \include{tikz_css_code_examples}
  \caption{Examples of CSS codes. Stars are colored orange, plaquettes blue.  The commutation relation~\eqref{eq:commrule} applies to all. (a) The stars (top) and plaquettes (bottom) of the rotated surface code. (b) The rotated surface code on $5\times 5$ qubits. The red vertical line is a logical $X$-operator, and the red horizontal line is a logical $Z$-operator. (c) The stars (top) and plaquette (bottom) of the $2$D toric code, which are then combined in a $2$D regular lattice. (d) The stars of the $2$D hexagonal color code are the faces of the hexagonal lattice. Qubits are placed on the vertices. The plaquettes are, up to switching $X\leftrightarrow Z$, identical to the stars. (e) The star (left) and plaquettes (right) of the $3$D Toric code. There are three orientations of the plaquette. Each plaquette is connected to $4$ qubits, each star to $6$. (f) The star (left) and plaquette (right) of a $2$D tessellation code on the hexagonal lattice. }
  \label{fig:code_examples}
\end{figure}

The commuting interactions  give rise to the star and plaquette Hamiltonians
\begin{equation}\label{eq:Ham}
H^\str_{\lsym} \coloneqq -\sum_{s\in \strset_{\lsym}} A_s, \qquad  H^\plq_{\lsym}  \coloneqq  - \sum_{p\in \plqset_{\lsym}} B_p \ ,
\end{equation}
which act on $\HH_{\lsym} $, as well as the full CSS Hamiltonian
\begin{equation}
      H^\strplq_{\lsym} \coloneqq H^\str_{\lsym} + H^\plq_{\lsym} . 
\end{equation} 
CSS Hamiltonians are spectrally trivial in the sense that $H^\str_{\lsym} $
is diagonal in the canonical tensor eigenbasis of the Pauli $ X_v $-operators with $ v \in \lsym$, and likewise for $ H^\plq_{\lsym}  $ with the canonical $ Z_v$-basis instead. Moreover, both parts can be jointly diagonalized. 
The ground state eigenspace of $  H^\strplq_{\lsym} $ is by construction the code space of the CSS code. In this paper, we do not focus heavily on the information-theoretic aspects of CSS codes, and we refer the interested reader to~\cite{nielsen_QuantumComputationQuantum_2012, gottesman_introductionquantumerror_2010}. As we will explain below, we will, however, show as a corollary to our main result that CSS codes are strongly not self-correcting in the high-temperature phase.

As part of our proof technique, we will also analyze CSS Hamiltonians restricted to subsets $R\subseteq \lsym$. They will be constructed similarly, starting from a set of stars and plaquettes touching this subset, 
\begin{equation}\label{eq:relstarpsets}
   \strset_R \coloneqq \{s\in \strset_{\lsym} | \cobdy s \cap R\neq \emptyset \} \quad \mathrm{and} \quad \plqset_R \coloneqq \{p \in \plqset_{\lsym} | \bdy p \cap R \neq \emptyset \} \ ,
\end{equation}
where we suppress the implicit dependence on $ \lsym$. By a similar abuse of notation, we write 
\begin{equation}\label{eq:HamR}
H^\str_R \coloneqq  - \sum_{s\in \strset_R}A_s, \qquad   H^\plq_R\coloneqq  - \sum_{p\in \plqset_R}B_p \ ,
\end{equation}
and $H^\strplq_R \coloneqq H^\str_R + H^\plq_R $ for the Hamiltonians on $R\subseteq \lsym$. We stress that the Hamiltonians in~\eqref{eq:Ham} and \eqref{eq:HamR} come with their choice of boundary interactions. 
These are encoded in how the relations $\bdy\colon\plqset_{\lsym}\to 2^{\lsym}$ and $\cobdy \colon \strset_{\lsym}\to 2^{\lsym}$ act at the boundary of the rectangle $\lsym$. That is, stars and plaquettes in the outermost layer of unit cells can differ from those in the bulk, see \cref{fig:surface_code_boundary} for an example.

\begin{figure}
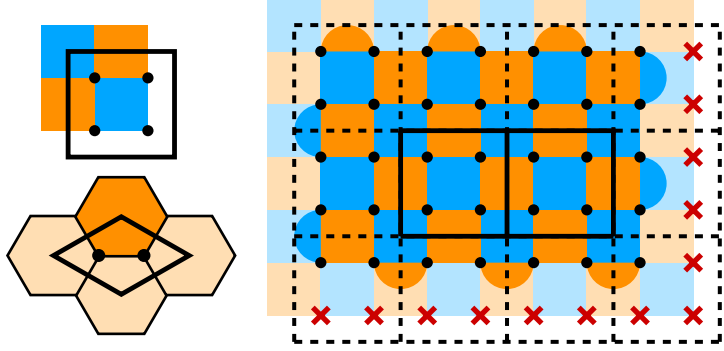

  \centering
  \include{tikz_surface_code_boundary}
  \caption{Unit cells and boundaries. Top left: a single unit cell of the rotated surface code containing four qubits, two stars (orange) and two plaquettes (blue). Bottom left: A single unit cell of the $2$D hexagonal color code containing two qubits and one star.  Right: a small rectangle with boundary conditions. The dashed unit cells lie at the boundary and are modified. The shaded interactions are truncated or fully removed. The crossed-out qubits are not interacting and can be ignored.}
  \label{fig:surface_code_boundary}
\end{figure}

\subsection{Davies Lindbladians and thermal equilibria}
The CSS code's equilibrium properties at inverse temperature $ \beta > 0 $ are described in terms of the canonical Gibbs states
\begin{equation}\label{eq:Gibbsstate}
    \rho_\lsym^\wildcard \coloneqq \frac{e^{-\beta H_\lsym^\wildcard}}{ Z_\lsym^\wildcard }, \quad \mbox{with} \quad \wildcard \in \{ \str, \plq , \strplq \}  , 
\end{equation}
and the partition functions $ Z_\lsym^\wildcard \coloneqq \Tr  e^{-\beta H_\lsym^\wildcard} $. 
(We suppress the dependence on $\beta$ in these quantities since we do not switch temperatures in our arguments.)

While it immediately follows from the commutativity of the star and plaquette Hamiltonian that the Gibbs state $\rho_\lsym^\strplq $ factorizes, as a first simple result, we show that the factors are the star and plaquette Gibbs states up to a factor of the Hilbert space's dimension.
\begin{theorem}\label{thm:equi}
For any CSS code the partition function factorizes, $Z_\lsym^\strplq = Z_\lsym^\str \ Z_\lsym^\plq \ 2^{-|\lsym|}  $, with $ |\lsym | $ the number of qubits within $ \lsym $, and hence the Gibbs state factorizes as well
$$
\rho_\lsym^\strplq = 2^{|\lsym|}\rho_\lsym^\str \ \rho_\lsym^\plq . 
$$
\end{theorem}
This will be proven, among other equilibrium properties of CSS codes, as  \cref{thm:equi:main} in \cref{sec:static}. 
The theorem reveals that the equilibrium properties of a CSS code are fully encoded in two underlying classical models, which are in one-to-one correspondence with the star and plaquette Gibbs states. More precisely,  since the star Gibbs state and its reduced states on subsets are all diagonal in the canonical $ X $-basis,  the star Gibbs state may hence be identified with the classical Gibbs measure: 
  \begin{equation}\label{eq:clGibbsm}
  \mu^\str(\textbf{x}) \coloneqq \bra{\textbf{x}} \rho^{\str}_{\lsym} \ket{\textbf{x}} ,
  \end{equation}
  on classical spin configuration $\textbf{x} \in 2^{\lsym} $. Here and in the following, we use braket notations, in which $\ket{\textbf{x}} \in \HH_\lsym $ is the corresponding $X$-basis product vector. 
  Likewise, one has $ \mu^\plq(\textbf{z}) \coloneqq \bra{\textbf{z}} \rho^{\str}_{\lsym} \ket{\textbf{z}} $, where $\textbf{z} \in 2^{\lsym} $ stands for a classical spin configuration, and $\ket{\textbf{z}}$ is the corresponding $Z$-basis vector.

Since we only consider qubits, the classical models corresponding to the restriction of $ H_\lsym^\str $ and $ H_\lsym^\str $ to the $ X $- and $ Z $-basis respectively are abelian $\Z_2$-gauge theories involving Ising spins~$\pm 1$. 
An example of such a lattice gauge theory is the famous Wilson-loop gauge theory \cite{wilson_ConfinementQuarks_1974, wegner_DualityGeneralizedIsing_1971} in $3$D. If one puts the qubits in the $3$D toric on the edges, it equals the plaquette interaction of this CSS code. 
The equilibrium properties of such classical lattice gauge theories are generally highly non-trivial. Available rigorous results concern duality results and other properties of phase transitions~\cite{wegner_DualityGeneralizedIsing_1971,aizenman_GeometricAnalysisIsing_2025,adhikari_Correlationdecayfinite_2025, forsstrom_DecayCorrelationsFinite_2022, duncan_TopologicalPhasesPlaquette_2025}. \\

The main aim of this paper is to study the mixing times for the convergence to the Gibbs equilibrium when the system is governed by a Davies Lindbladian, the standard model for thermal noise in the weak-coupling
limit, whose construction dates back to Davies ~\cite{davies_Diffusionweaklycoupled_1972}. In the Heisenberg picture, the Lindbladian dynamics acts on the space of bounded linear operators $ \BB(\HH_\lsym) $, on which $ \| O\| $ stands for the operator norm and $ O^\dagger $ for the Hilbert-space adjoint of $ O \in \BB(\HH_\lsym) $. The space of bounded operators may also be equipped with the Hilbert-Schmidt scalar product
$   \langle O_1, O_2\rangle \coloneqq \Tr(O_1^\dag O_2) $.
  For a linear (super-)operator $\mathcal{O}:\BB(\HH_\lsym)\to \BB(\HH_\lsym)$, we denote  by $\mathcal{O}^*$ its adjoint with respect to this inner product.

In the Heisenberg picture, the Davies Lindbladian of the CSS code  is given by a sum of two parts 
  \begin{equation}\label{def:LL}
    \LL^{\strplq}_\lsym = \LL^{\str}_\lsym + \LL^{\plq}_\lsym . 
  \end{equation}
In Davies' construction~\cite{davies_Generatorsdynamicalsemigroups_1979,davies_Diffusionweaklycoupled_1972}, these parts result from tracing out interactions with a bath given by the Pauli operators $ Z_v$ and $ X_v$  with $ v \in \lsym$, respectively, and taking the weak-coupling limit. 
  For any $ R \subseteq \lsym$, the parts  $\LL_R^{\wildcard}$ with $ \wildcard \in \{ \str, \plq \} $ are sums over local terms $\LL^{\wildcard}_R \coloneqq \sum_{v\in R} \LL^{\wildcard}_{v}$ whose action on $ O \in \BB(\HH_\lsym)  $ are in standard Lindbladian form
  \begin{equation}\label{def:LLv}
    \LL^{\wildcard}_{v}(O) = \sum_{\omega} h^{\wildcard}_v(\omega) \left( L_v^{\wildcard, \dag}(\omega) O L^{\wildcard}_v(\omega) -  \frac{1}{2}\{L_v^{\wildcard, \dag}(\omega)L^{\wildcard}_v(\omega), O\} \right).
  \end{equation}
  The summation is over all Bohr frequencies $\omega$, i.e.\ differences of energies of $H^{\wildcard}_{v}$, and $\{\cdot,\cdot\}$ denotes the anti-commutator.
  The jump operators $L^{\wildcard}_v(\omega)$ are given by
  \begin{equation}\label{def:jumpops}
    L_v^\str(\omega) \coloneqq Z_v P^\str_{\bdy v}(\omega)
    \quad \mathrm{and} \quad
    L_v^\plq(\omega) \coloneqq X_v P^\plq_{\cobdy v}(\omega) ,
  \end{equation}
  where $P^\str_{\bdy v}(\omega)$ and $P^\plq_{\cobdy v}(\omega)$ are the spectral projections of
  $  -2\sum_{s\in \bdy v} A_s $ and $ -2\sum_{p\in \cobdy v}B_p $,
  respectively. Note that the eigenvalues of the latter are in one-to-one correspondence with the energy differences of $H^{\str}_{v} $ and $H^{\plq}_{v}$.
  The jump rates incorporate the bath's inverse temperature $ \beta $, are assumed to be strictly positive, $h^{\wildcard}_v(\omega)> 0 $, and satisfy the detailed balance condition: 
  \begin{equation}\label{eq:jumprate}
    h^{\wildcard}_v(-\omega) = h^{\wildcard}_v(\omega) e^{{-\beta \omega}} . 
  \end{equation}
  An example would be $ h_v(\omega) = \big( 1+ e^{-\beta\omega}\big)^{-1}$. 
  To avoid issues with jump rates vanishing asymptotically at infinity, in our main result we also need the following condition.
  \begin{definition}\label{def:jumprates}
  For a family of CSS codes, 
    the $ \wildcard$-Lindbladian's jump rates  are said to be \emph{uniformly positive} if for any  inverse temperature $ \beta > 0 $ there is some  $ g^\wildcard > 0 $ such that uniformly in the energy differences $ \omega $ of all $ H_v^\wildcard$ and all vertices:
    $$
    \inf_{v, \omega}\  h_v^\wildcard(\omega) e^{-\beta\omega/2} \geq g^\wildcard . 
    $$
  \end{definition}
  In case of translation invariance of the CSS Hamiltonian, the infimum over $\omega$ is a minimum. In this case, the condition applies if, e.g., the jump rates $ h_v^\wildcard $ are independent of $ v$. \\

In the Schr\"odinger picture, the Davies Lindbladian acts as the (Hilbert-Schmidt) adjoint $ \mathcal{L}^{\strplq, *}_\lsym  $ on the set of quantum states 
  \begin{equation*}
    \SSS(\HH_\lsym)\coloneqq \{\sigma \in \BB(\HH_\lsym)| \sigma=\sigma^\dag,\  \sigma \geq 0, \ \Tr(\sigma)=1\} . 
  \end{equation*}
The Davies construction and the fact that the set of Pauli matices $ X_v , Z_v $ with $ v \in \lsym $ generate the full matrix algebra $ \BB(\HH_\lsym) $ ensures~\cite{davies_Generatorsdynamicalsemigroups_1979,frigerio_StationaryStatesQuantum_1978} that the Gibbs state $ \rho_\lsym^\strplq  $ is the unique fixed point of the completely positive, trace-preserving semigroup, $ \exp\left(t\mathcal{L}^{\strplq, *}_\lsym\right) $, $ t \geq 0 $. 
The infinite-time limit of this Lindbladian semigroup yields the full-rank Gibbs state 
\begin{equation*}
\lim_{t\to \infty} \exp\left(t\mathcal{L}^{\strplq, *}_\lsym\right)(\sigma) = \rho_\lsym^\strplq 
\end{equation*}
for any initial state $ \sigma \in \SSS(\HH_\lsym)$. 
More information and further properties of the Davies Lindbladians are collected in \cref{sec:davies}.
  
\subsection{MLSI and rapid mixing towards equilibrium}

For a general Lindbladian $ \mathcal{L}_\lsym $, the long-time limit in the Schr\"odinger picture will be denoted by 
\begin{equation}\label{def:LTlimit}
  \E^*_{\lsym} \coloneqq \lim_{t\to\infty }e^{t\LL^*_{\lsym}} ,
\end{equation}
assuming its existence, as in the case studied here.
One of our goals is to present a bound on the mixing time 
\begin{equation}\label{eq:mixtime}
  t_{\mathrm{mix}}(\varepsilon) = \inf \{t\geq 0 \ | \ \forall \sigma \in \SSS(\HH_{\lsym}):\ \|e^{t\LL^*_{\lsym}} (\sigma) - \E_{\lsym}^*(\sigma)\|_{1}\leq \varepsilon\} 
\end{equation}
corresponding to an error $\varepsilon> 0 $ in trace-distance $ \| \cdot \|_1 \coloneqq \Tr | \cdot | $. This will be accomplished by deriving a modified logarithmic Sobolev inequality (MLSI).

\begin{definition} 
A family of Lindbladians $\LL_{\lsym}$ indexed by $ \lsym \in \FF$ is said to satisfy an \emph{MLSI with constant $\alpha >0$} if for any $\Lambda\in \FF$ and any full-rank state $\sigma \in \mathcal{S}(\HH_{\lsym})$:
\begin{equation}\label{eq:defmlsi}
  2 \alpha \ D(\sigma \|\E^*_{\lsym}(\sigma))\leq EP_{\lsym}(\sigma) .
\end{equation}
Here 
$D(\sigma\|\sigma') \coloneqq \Tr(\sigma(\ln\sigma - \ln \sigma')) $ 
is the relative entropy between full-rank states $ \sigma, \sigma' \in  \SSS(\HH_\lsym) $, and 
\begin{equation*}
  EP_{\lsym} (\sigma)\coloneqq -\Tr(\LL^*_{\lsym}(\sigma)(\ln \sigma - \ln \E^*_{\lsym}(\sigma)))
\end{equation*}
is the entropy production.
\end{definition}
The name entropy production stems from Spohn's computation~\cite{spohn_Entropyproductionquantum_1978}
\begin{equation}\label{eq:EPder}
EP_{\lsym}(\sigma) = -\frac{\dd}{\dd t}D(e^{t\LL^*_{\lsym}}(\sigma)  \| \E^*_{\lsym}(\sigma)) \big|_{t=0} . 
\end{equation}
While an MLSI constant $ \alpha $, which depends on $ \lsym  $, is easily derived \cite{gao_Completeentropicinequalities_2022}, e.g.\ from the spectral gap of the Lindbladian on a finite $ \lsym$, we will aim at constants $ \alpha $, which are uniform in the system size.\\

Let us briefly recall how to obtain 
a bound on the mixing time~\eqref{eq:mixtime} from the MLSI. First, Grönwalls's inequality allows to integrate \eqref{eq:EPder} when inserting~\eqref{eq:defmlsi}:
\begin{equation}\label{eq:intmlsi}
  D(e^{t\LL^*_{\lsym}}(\sigma) \| \E^*_{\lsym}(\sigma))\leq e^{-2\alpha t} D(\sigma \| \E^*_{\lsym}(\sigma)) \leq e^{-2 \alpha t} \ln(1/(\E^*_{\lsym}(\sigma))_{\mathrm{min}}) .
\end{equation}
Here $(\E^*_{\lsym}(\sigma))_{\mathrm{min}}$ is the minimal eigenvalue of $\E^*_{\lsym}(\sigma)$.  Since the latter has full rank, see \cref{lem:condexpkern}, this is always non-zero. 
For the Davies semigroup, $\E^*_{\lsym}(\sigma)$ agrees with the Gibbs state $ \rho_\lsym^\strplq $ and the negative logarithm of its minimal eigenvalue is $\beta \| H_{\lsym}^\strplq \| + \ln{Z_{\lsym}^\strplq} $. Since the CSS Hamiltonian is a sum of bounded, local terms, its operator norm $ \| H_\lsym^\strplq \| $ is a polynomial, $ \mathrm{poly}(|\lsym|) $, in the system's size.
By the quantum Pinsker inequality this then yields:
\begin{equation}\label{eq:rapidmix}
    \|e^{t\LL^*_{\lsym}} \sigma - \E_{\lsym}^*\sigma\|_{1} \leq \sqrt{2 D(e^{t\LL^*_{\lsym}}(\sigma) \| \E^*_{\lsym}(\sigma))} \leq \mathrm{poly}(|\lsym|) e^{-\alpha t} .
\end{equation}
In other words, if $ \alpha > 0 $ is independent of $ \lsym $, the mixing time $ t_{\mathrm{mix}}(\varepsilon) $ scales polylogarithmically in the system size. 
As shown in \cref{sec:appendix:full-rank} using a continuity argument, an MLSI for full-rank states implies rapid mixing in the sense of~\eqref{eq:rapidmix} for all initial states $ \sigma \in \SSS(\HH_{\lsym})$. 

The estimate~\eqref{eq:rapidmix} derived from the MLSI should be compared with the one obtained from a spectral gap, 
$ \gap \mathcal{L}_\lsym $, of the Lindbladian \cite{temme_chi2divergenceMixing_2010,kastoryano_QuantumLogarithmicSobolev_2013}:
\begin{equation*}
  \| e^{t\mathcal{L}_\lsym^* }\sigma - \E_{\lsym}^*(\sigma) \|_1 \leq e^{-t \, \gap \mathcal{L}_\lsym}  \sqrt{(\E^*_{\lsym}(\sigma))_{\mathrm{min}}^{-1}} \leq \exp\left[ \mathrm{poly}(|\lsym|)  -t \, \gap \mathcal{L}_\lsym\right] .
\end{equation*}
In case the gap is strictly positive uniformly lower bounded in $ \lsym $, this corresponds to a bound on the mixing time, which is only polynomial in the system size. 
For commuting Hamiltonians such as CSS codes, conditions ensuring a lower bound on $ \gap \mathcal{L}_\lsym$ which is uniform in the system size $ |\lsym |$, have been established in several works (see e.g.\ \cite{alicki_thermalizationKitaevs2D_2009,temme_LowerBoundsSpectral_2013,kastoryano_QuantumGibbsSamplers_2016,lucia_SpectralGapBounds_2025}). \\

The use of MLSIs is not limited to deriving tight bounds on the mixing time. 
In particular, it implies \cite{kastoryano_QuantumLogarithmicSobolev_2013} a Poincar\'e inequality, i.e.\ a uniform lower bound on $\gap \mathcal{L}_\lsym $. 
Through quantum functional and transport cost inequalities,  MLSIs relate to concentration of measure results with applications to e.g.\  the eigenstate thermalization hypothesis \cite{rouze_ConcentrationQuantumStates_2019,depalma_QuantumWassersteinDistance_2021,depalma_QuantumConcentrationInequalities_2022}.

\subsection{MLSI from uniform  correlation decay}
Given the star or plaquette Gibbs state~\eqref{eq:Gibbsstate}, a local Gibbs
state associated with $ R \subseteq \lsym$ is 
    \begin{equation}\label{def:restorGNS}
    \hat\rho^\wildcard_R \coloneqq e^{-\beta H^\wildcard_R}\left(\Tr_R e^{-\beta H^\wildcard_R} \right)^{-1} .
  \end{equation}
  For $ \wildcard \in \{\str,\plq \}$, both factors in the right side are diagonal in the canonical $ X$- respectively $ Z$-basis. The second operator on the right side restores the normalization, and results from the partial trace $ {\Tr_R} $ over the $ \HH_R $-component of the Hilbert space, where we use the convention that the partial trace is tensored with the identity $\identity_R$ on $R$.
  For the full region, $R=\lsym$, this state coincides with the Gibbs state, $\hat{\rho}_\lsym^{\wildcard} = \rho_\lsym^\wildcard$. As will be explained in more detail in \cref{sec:static}, for general $ R \subset \lsym $ the state $ \hat\rho^\wildcard_R $ 
differs from the reduced state $\rho^\wildcard_R \coloneqq \Tr_{R^c} \rho^\wildcard_\lsym $ in which one traces out the complement $ R^c \coloneqq \lsym \backslash R $. \\

Our condition for the validity of an MLSI will involve the decay of correlations encoded in the star and plaquette states~\eqref{def:restorGNS}. 
The decay condition involves a geometry sketched in \cref{fig:approxtensor}: two rectangular subsets 
\begin{equation*}
    UV  \coloneqq  U \uplus V , \quad\mbox{and}\quad VW\coloneqq  V\uplus W ,
\end{equation*} 
which overlap in $ V $. The latter is disjoint from both $ U $ and $ W $. These sets form the disjoint union $ UVW \coloneqq U \uplus V \uplus W$. 

\begin{figure}[ht]
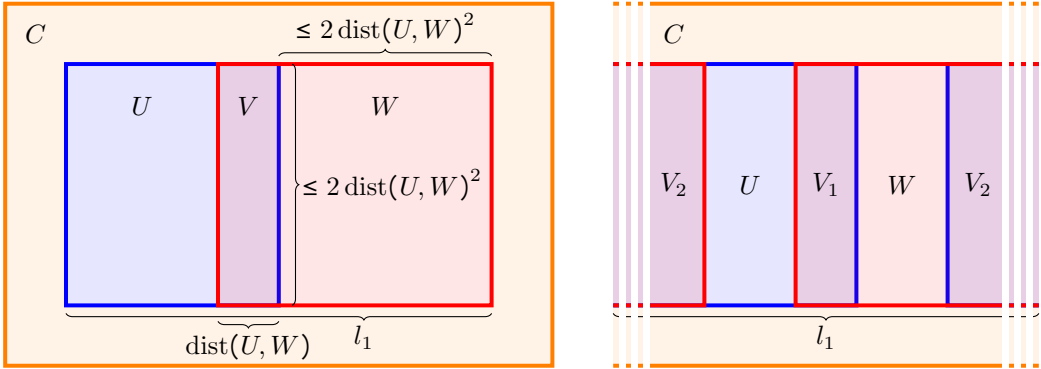

  \centering
  \include{tikz-approxtensor}
  \caption{Geometry in the DS-condition. Left: Two rectangles $U\uplus V$ and $V\uplus W$ intersecting in $V$. The example shows an alignment of $ U $ and $ W $ along the first coordinate direction.  Their union $U \uplus V\uplus W$ is again a rectangle. The height and width of $W$ are bounded by the square of the width of $V$. Right: The same setup for periodic boundaries. Now $V=V_1\uplus V_2$ is a disjoint union of two rectangles. The quantity $ l_1 $ appears in the proof of the MLSI in \cref{sec:mainresult}.}
  \label{fig:approxtensor}
\end{figure}
\begin{definition}\label{def:dscond}
    For a family of CSS Hamiltonians on a $D$-dimensional lattice and $ \wildcard \in \{\str,\plq \}$, we say that the $\wildcard$-part of the family of Gibbs states satisfies the \emph{DS-condition} at inverse temperature $ \beta > 0 $ with constants $K^\wildcard,\xi^\wildcard \in (0,\infty) $ if there is some length scale $ L_0^\wildcard > 0 $ such that
    \begin{equation}\label{eq:corrdecaywildcard}
      \left\| \frac{\Tr_{VW} \hat{\rho}^\wildcard_{UVW}}{\Tr_{V} \hat{\rho}^\wildcard_{UV}} -1 \right\| \leq K^\wildcard e^{-\xi^\wildcard \dist(U,W)}
    \end{equation}
    holds for all $\lsym\in\famsym$, and any pair of overlapping rectangles $UV\subseteq \lsym$ and $VW \subseteq \lsym$ with $\diam(W)\leq 2 \dist(U,W)^2$, and such that $UVW$ is also a rectangle,  whose diameter is bigger or equal to $ L_0^\wildcard  $. \\
    We will write $\mathrm{DS}_\beta^\wildcard(K, \xi)$ as a shorthand for the DS condition with parameters $K$ and $\xi$ above.
\end{definition}
Requiring $\diam(W)\leq 2 \dist(U,W)^2$ weakens the condition to rectangles for which the overlap is not too short. This does not impede the proof but will allow us to absorb some dependence on the size of $W$ in the proof of MLSI in \cref{sec:mainresult}.  

As will be explained in \cref{sec:decaycond} below, the DS-condition of the $ \wildcard $-part of the quantum Gibbs state is in fact implied by the Dobrushin-Shlosman characterization \cite{dobrushin_CompletelyAnalyticalInteractions_1987} of the high-temperature phase of the classical Gibbs measure $\mu^\wildcard$. In fact, since the states are diagonal in the $ X $- or $ Z $-basis, the operator norm on the left side of \eqref{eq:corrdecaywildcard} is nothing but an infinity norm of a ratio of conditional probabilities related to $ \mu^\wildcard$. E.g.\ for the star part one has
\begin{equation}\label{eq:DScondprob}
    \left\| \frac{\Tr_{VW} \hat{\rho}^{\str}_{UVW}}{\Tr_{V} \hat{\rho}^{\str}_{UV}} -1 \right\|
    = \max_{\substack{\textbf{x}_{U}\in 2^U \!\! ,\textbf{x}_{C} \in 2^C \\ \textbf{x}_{W} \in 2^W}} \left| \frac{\mu_{UC}^\str(\textbf{x}_{U}|\textbf{x}_{C})}{\mu_{UWC}^\str(\textbf{x}_{U}|\textbf{x}_{W}\textbf{x}_{C})} -1\right| \ ,
  \end{equation}
  where $\mu_{U(W)C}^\str( \cdot \ | \ (\textbf{x}_{W}) \textbf{x}_{C}) $ are conditional probabilities associated with the marginal  (defined in~\eqref{eq:defcondGibbs})  of the star Gibbs measure $ \mu^\str $ on $ U(W)C$ with $ C \coloneqq (UVW)^c$.
Based on \cref{lem:weakDS} below and known classical equivalence relations \cite{dobrushin_CompletelyAnalyticalInteractions_1987,stroock_logarithmicsobolevinequality_1992a,martinelli_ApproachEquilibriumGlauber_1994,martinelli_ApproachEquilibriumGlauber_1994a,martinelli_LecturesGlauberDynamics_1999}, the exponential decay in the distance $ \dist(U,W) $ of the conditional probabilities in~\eqref{eq:DScondprob} is also implied by one of the following spatial or equivalently temporal mixing conditions of the $ \wildcard$-part of the underlying classical $ \Z_2 $-lattice gauge theory:
\begin{itemize}
\item the exponential
decay of finite-volume covariances, uniformly in the boundary conditions (dubbed `SMT-condition' in~\cite{martinelli_LecturesGlauberDynamics_1999}).
\item the uniform spectral gap on the generator of the classical Glauber dynamics  \cite{stroock_logarithmicsobolevinequality_1992a,yoshida_RelaxedCriteriaDobrushinShlosman_1997}.
\item the uniform boundedness of the logarithmic Sobolev constant of the classical Glauber dynamics \cite{stroock_logarithmicsobolevinequality_1992a,martinelli_ApproachEquilibriumGlauber_1994,martinelli_ApproachEquilibriumGlauber_1994a}. 
\end{itemize}
As will be explained in detail in \cref{sec:trivhight}, the DS-condition applies at \emph{any} positive temperature to the star and plaquette part of any CSS code for which each qubit takes part in exactly two star and plaquette interactions. This is clearly the case for the $2$D toric code for which, e.g.\  the left side in~\eqref{eq:DScondprob} is then bounded by a temperature-dependent constant times $  \exp(- |\strset_V | |\ln\tanh(\beta)|) $, i.e.\ superexponentially in $ \dist(U, W) $. 
Notably, the same argument also applies to the star part of the $3$D toric code at any positive temperature.

Despite this suggestive close relation to a classical Glauber dynamics, we caution the reader that the nature of the $ \wildcard$-Lindbladian dynamics is much richer in that it acts on all quantum states as initial states and not just on those diagonal in a particular basis. In this sense, the following main result is also not simply implied by the classical results in~\cite{stroock_logarithmicsobolevinequality_1992a,martinelli_ApproachEquilibriumGlauber_1994,martinelli_ApproachEquilibriumGlauber_1994a}.\\

The main result is a pair of MLSIs for the star and plaquette part of the Lindbladian. Aside from the DS-condition, it only requires translation invariance and the uniform positivity of the jump rates in the sense of \cref{def:jumprates}. 
\begin{theorem}\label{thm:main}
  For a family of translation-invariant CSS Hamiltonians on a $D$-dimensional lattice for which the $ \wildcard \in \{ \str, \plq \} $ part at inverse temperature $ \beta > 0 $ satisfies $\mathrm{DS}_\beta^\wildcard(K^\wildcard, \xi^\wildcard)$ with constants $K^\wildcard <\infty$, $\xi^\wildcard>0$, and the $ \wildcard $-Lindbladian has uniformly positive jump rates, 
 there exists $\alpha^\wildcard>0$, which is independent of $\lsym$, such that for all $\lsym \in \FF $ the Lindbladian $\LL_\lsym^\wildcard$ satisfies the MLSI with constant $\alpha^\wildcard$.
\end{theorem}
The proof of this main result is found in \cref{sec:mainresult}.
As an immediate corollary,  we conclude the MLSI for the full Lindbladian if both the star and plaquette parts satisfy the conditions in~\cref{thm:main}.
\begin{corollary}\label{corMLSIfull}
    If for a family of translation-invariant CSS Hamiltonians on a $D$-dimensional lattice both parts, $ \str$ and $ \plq $, satisfy the conditions of  \cref{thm:main}, then the full Davies Lindbladian $ \LL_\lsym^\strplq $ satisfies an MLSI with constant $$\alpha^\strplq = \min\{\alpha^\str, \alpha^\plq\} . $$
\end{corollary}
The proof of this corollary is also found in \cref{sec:mainresult}. 
The mixing speed is controlled by the slower of the two parts. The reason for this is that $\LL^\str_{\lsym}$ and $\LL^\plq_{\lsym}$ mix to different subalgebras. If, say, $\LL^\str_{\lsym}$ mixes much faster than $\LL^\plq_{\lsym}$, then states will be projected quickly to the kernel of $\LL^\str_{\lsym}$ but it will take a long time for $\LL^\plq_{\lsym}$ to mix them further. \\

Another important implication of \cref{thm:main} concerns a strong lack of self-error correction of CSS codes in case one of the two parts of the Lindbladian mixes rapidly. In this context, we recall that the quantum information is encoded in the logical operators $X_L, Y_L, Z_L$ of the CSS code \cite{nielsen_QuantumComputationQuantum_2012, gottesman_introductionquantumerror_2010}. We use the convention that $X_L$ is a string of Pauli $X$-operators, $Z_L$ is a string of Pauli $Z$-operators, and $Y_L = i X_L Z_L$. 
These logical operators act non-trivially on the code space, i.e., the ground-state eigenspace of $ H^\strplq_\lsym$.  
They are not contained in the algebra spanned by the star and plaquette operators~\eqref{eq:operatordef}, which all act as the identity on the code space.  All logical operators $O_L\in \{X_L, Y_L, Z_L \} $ commute with the star and plaquette operators: $$ [O_L,A_s]=[O_L,B_p]=0 . 
$$
 We show that the quantum information is already destroyed in the kernel of either of the two parts of the Lindbladian.  
In the following, we spell this for the star part. An analogous result holds for the plaquette part. 
\begin{corollary}
\label{lem:nonselfcorr}
    Consider a family of CSS Hamiltonians on a $D$-dimensional lattice with Davies-Lindbladian $\LL^\strplq_{\lsym} = \LL^\str_{\lsym} + \LL^\plq_{\lsym}$. Assume that $\LL^\str_{\lsym}$ satisfies an MLSI with a constant $ \alpha^{\str} > 0 $ independent of $\lsym$. Then for any $ \lsym  \in \famsym $ and all states $\sigma \in \SSS(\HH_{\lsym})$ and $ t > 0 $:
    \begin{equation}
        \left|\Tr( e^{t\LL^{\strplq, *}_{\lsym}}(\sigma)\ O_L )\right|  \leq \mathrm{poly}(|\lsym|)e^{-\alpha^{\str} t}  
    \end{equation}
    for both $O_L\in \{X_L, Y_L\} $, the logical $ X $- and $ Y $-operator. 
\end{corollary}

The above result is proven at the end of \cref{sec:mainresult}. It shows that all off-diagonal terms in the protected subspace vanish logarithmically fast in the system size $|\lsym|$. In other words, only classical information encoded in the logical $Z$-basis can survive longer than logarithmically, and the system is hence not self-correcting. This, for example, applies to the toric code in $ 3$D, since, as explained above, its star part can be shown to satisfy the MLSI at any positive temperature.

\subsection{Comparison to existing literature}

\paragraph{Fast mixing via gap estimates.} 

This paper is the first instance in which an MLSI has been proven for the Davies dynamics associated with any CSS code, and in particular, for the celebrated 2D toric code at any positive temperature. This solves a long standing open problem. Prior works have derived the exponentially weaker spectral gap for Davies dynamics of broader classes of commuting Hamiltonians.

In \cite{kastoryano_QuantumGibbsSamplers_2016}, a sufficient condition for the positive spectral gap of a Davies Lindbladian associated with an arbitrary local, commuting Hamiltonian was established. It involves a strong form of clustering on the Gibbs state, which can be interpreted as a conditional version of decay of correlations in the $2$-norm, stronger in general than the standard covariance decay. For specific commuting models, it was shown in \cite{alicki_thermalizationKitaevs2D_2009} that the Davies generator associated with the $2$D toric code has a positive spectral gap at any positive temperature. This was revisited in \cite{ding_PolynomialTimePreparationLowTemperature_2025}, where another lower bound was provided for the gap, scaling polynomially on the system size, but only linearly on the inverse temperature. The work of  \cite{alicki_thermalizationKitaevs2D_2009} was subsequently extended to Abelian quantum double models in \cite{komar_Necessityenergybarrier_2016}, and to non-Abelian ones in \cite{lucia_ThermalizationKitaevsquantum_2023}. In the recent preprint \cite{lucia_SpectralGapBounds_2025}, the authors have further shown that a lower bound for the spectral gap of a class of GNS-symmetric generators associated to local, commuting Hamiltonians is equivalent to a mixing condition on the Gibbs state, which is known to be fulfilled at least for any finite-range 1D model, as well as by Kitaev’s quantum double models \cite{kitaev_Faulttolerantquantumcomputation_2003}.

\paragraph{Preparation of Gibbs states.}

The preparation of complex quantum systems is expected to be one of the main applications of quantum computers. Quantum Gibbs sampling lies at the heart of multiple fundamental problems in statistical physics, machine learning, and probabilistic inference. Some classical algorithms, such as the  Markov Chain Monte Carlo (MCMC) \cite{levin_MarkovChainsMixing_2008}, appear ubiquitously in the literature of Gibbs sampling, Typically, they are provably efficient at high enough temperature \cite{martinelli_ApproachEquilibriumGlauber_1994a}, but are generally believed to be efficient in practice in more generality \cite{brooks_HandbookMarkovChain_2011}. 

The classical MCMC algorithm inspired many quantum algorithms for the preparation of quantum Gibbs states \cite{gilyen_QuantumGeneralizationsGlauber_2024,jiang_QuantumMetropolisSampling_2024}, starting with the seminal quantum Metropolis algorithm \cite{temme_QuantumMetropolisSampling_2011}. In a slightly different direction, one of the primary avenues for preparing Gibbs states is based on dissipation. A collection of recent papers \cite{cleve_EfficientQuantumAlgorithms_2017,li_SimulatingMarkovianOpen_2023} have shown that a Davies Lindbladian can be implemented in $\mathcal{O}(n \log n)$ runtime on $n$ qubits. If the Davies Lindbladian mixes rapidly as is the case in which an MLSI is known~\cite{capel_Quantumconditionalrelative_2018,beigi_HypercontractivitylogarithmicSobolev_2016,bardet_ModifiedLogarithmicSobolev_2021,capel_ModifiedLogarithmicSobolev_2021,kochanowski_RapidThermalizationDissipative_2025,bergamaschi_QuantumComputationalAdvantage_2024}, then the circuit depth to prepare the Gibbs state is $\mathcal{O}(n \operatorname{polylog} n)$ (see e.g. \cite[Thm.~6]{paez-velasco_EfficientSimpleGibbs_2025}). In comparison, in case of fast mixing  the circuit depth is $\mathcal{O}( \operatorname{poly} n)$, and in the case of a spectral gap \cite{kastoryano_QuantumGibbsSamplers_2016,alicki_thermalizationKitaevs2D_2009,komar_Necessityenergybarrier_2016,lucia_ThermalizationKitaevsquantum_2023,lucia_SpectralGapBounds_2025}  it is $\mathcal{O}(n^{2} \operatorname{polylog} n)$. 
In the particular case of the toric code in 2D, this implies that the previous work \cite{alicki_thermalizationKitaevs2D_2009} yielded a circuit depth of  $\mathcal{O}(n^{2} \operatorname{polylog} n)$, and that of \cite{ding_PolynomialTimePreparationLowTemperature_2025}, a depth of $\mathcal{O}(n^{4} \operatorname{polylog} n)$, the latter with a linear dependence on $\beta$ and the former with an exponential dependence. Additionally, alternative approaches based on duality of the Hamiltonians with classical models, provide a circuit complexity of  $\mathcal{O}(n^{3/2})$ \cite{paez-velasco_EfficientSimpleGibbs_2025} for the $2$D toric code, and of $\mathcal{O}(n^{2})$ \cite{hwang_GibbsStatePreparation_2025}
 for the defective toric code, both independently of $\beta$. The current manuscript, with a depth of $\mathcal{O}(n \operatorname{polylog} n)$, implies therefore the most efficient algorithm preserving geometric locality to prepare the Gibbs state of the 2D toric code, as far as we know, for fixed $\beta$. 

 There have been some recent attempts to extend the above results to non-commuting Hamiltonians. In \cite{chen_QuantumThermalState_2023,chen_EfficientExactNoncommutative_2023}, the authors introduced a Lindbladian that is efficiently implementable in $\mathcal{O}(n \operatorname{polylog} n)$ depth for non-commuting, local Hamiltonians. This Lindbladian has been shown to mix rapidly at high temperatures \cite{rouze_OptimalQuantumAlgorithm_2024,rouze_EfficientThermalizationUniversal_2024}, hence preparing the corresponding Gibbs state efficiently. Another family of efficiently implementable Lindbladians for non-commuting Hamiltonians was presented in \cite{ding_EfficientQuantumGibbs_2025}. 

 There are other methods for the preparation of Gibbs states, such as Grover approaches \cite{poulin_SamplingThermalQuantum_2009,chowdhury_QuantumAlgorithmsGibbs_2017} or imaginary time evolution \cite{motta_DeterminingEigenstatesThermal_2020}.

\subsection{Comments on the proof}
The proof of the main result (\cref{thm:main}) uses a multiscale argument, bounding the MLSI constant on larger and larger length scales. This requires relating both the relative entropy and the entropy production from a large to a smaller scale. The entropy production is monotonic in the system size (\cref{pro:entropyprodprop}). 
For the left-hand side of the MLSI~\eqref{eq:defmlsi}, we show that the relative entropy $D(\sigma\| \E^*_{R\cup R'}(\sigma))$  of a union of two subsets $ R , R' \subseteq \lsym $ is upper bounded by the sum $D(\sigma\| \E^*_{R}(\sigma)) + D(\sigma\| \E^*_{R'}(\sigma))$, up to some factor decaying exponentially in the size of the intersection, a condition known as approximate tensorization. It follows from a positivity ordering of the conditional expectations \cite{gao_Completepositivityorder_2025}. 
We show that the DS-condition implies this positivity ordering. To do so, 
we use the following novel, explicit expression (here in the star case):
\begin{equation}
    \E^\str_{R}(O) = \TT^\str_R \circ \PP^\str_R(\hat\rho^\str_RO)
\end{equation}
for any $ O \in \BB(\HH_\lsym) $, where $\TT^\str_R$ is a $Z$-basis pinching on $R$, the operator $\PP^\str_R$ is a pinching built from spectral projections of star operators, and  $\hat\rho^\str_R$ the state~\eqref{def:restorGNS} for $ \wildcard = \str $. This explicit expression is valid for CSS codes on more general graphs. The positivity ordering on the conditional expectations will follow from a positivity ordering on the marginals of these states, which is exactly the DS-condition.

As an initial condition, we use the positivity of the 
MLSI-constant on some finite length scale (depending on the constants in the DS-condition). This will follow from known results on strict positivity of the spectral gap of the Lindbladian when restricted to a finite region in case the jump rates are uniformly positive (cf.~\cite{kastoryano_QuantumGibbsSamplers_2016,gao_Completeentropicinequalities_2022}). \\

Most statements in the subsequent sections will be the same for the star ($\str$) and plaquette ($\plq$) quantities with identical proofs. Thus, without loss of generality, we will only give the proofs for the star quantities.\\

\noindent
\textbf{Structure of the paper:} In \cref{sec:static} we collect some useful properties of the Gibbs state.  We show that the DS-condition is implied by one of the Dobrushin-Shlosman high-temperature conditions, and establish its validity at any positive temperature for all codes with a trivial high-temperature expansion such as the $2$D toric code. In \cref{sec:davies} we recall the construction and some basic properties of the Davies Lindbladians. In \cref{sec:condexp} we derive an explicit expression for the conditional expectations which we then use in \cref{sec:approxtensor} to prove approximate tensorization. The proofs of the main results are given in \cref{sec:mainresult}. In \cref{sec:appendix:full-rank}, we recall some more properties of conditional expectations and justify the use of full-rank states.

\section{Equilibrium}\label{sec:static}
We start our analysis by collecting properties of the Gibbs state of CSS Hamiltonians as well as their star and plaquette parts. This section also contains a more detailed discussion and tools for the verification of the DS-condition (\cref{def:dscond}). 

\subsection{Properties of the Gibbs state}
We recall from~\eqref{eq:Gibbsstate} that $    \rho^\wildcard_{\lsym} = e^{-\beta H^\wildcard_{\lsym}}/Z^\wildcard_{\lsym} $ with $ \wildcard \in \{\strplq, \str, \plq \} $ stands for the Gibbs state of a CSS Hamiltonian and its star or plaquette part on a finite set $ \lsym $ at inverse temperature $\beta \geq 0 $. 
The corresponding reduced states on the Hilbert space over $ R \subset \lsym $ 
result from tracing out the complement of $ R $ within $ \lsym$:
\begin{equation}
\rho^\wildcard_{R} \coloneqq \Tr_{R^c} \rho^\wildcard_{\lsym} ,
\end{equation}
where, by a consistent slight abuse of notation, we suppress the dependence on $ \lsym $ in the left side. The above reduced states should not be confused with the states $ \hat\rho^\wildcard_{R} $ introduced in~\eqref{def:restorGNS} and on which more will be said in the next subsection. \\

As announced in \cref{thm:equi}, the partition function and Gibbs states factorize. As a consequence, equilibrium properties of the full system are captured by the properties of the star- and plaquette systems. As the subsequent proof reveals, the following theorem is valid for CSS Hamiltonians on general finite sets $ \lsym$. Its proof only relies on the algebraic structure of the elementary star and plaquette operators~\eqref{eq:operatordef}. As is evident from their definition, the star Gibbs state is diagonal in the canonical $ X$-basis and the plaquette Gibbs state is diagonal in the canonical $ Z$-basis. The same applies to their reduced states.
\begin{theorem}\label{thm:equi:main}
  Consider a CSS  Hamiltonian on a finite set $\lsym $, and let any $0 \leq \beta < \infty$. Then $ Z^\strplq_{\lsym}            = Z^\str_{\lsym}     \cdot  Z^\plq_{\lsym}     \cdot 2^{-|\lsym|}  $, and for any $R\subseteq \lsym$:
  \begin{equation}\label{eq:reducedstatedecomp}
  \rho_R^\strplq      =\rho^\str_R    \cdot  \rho^\plq_R  \cdot 2^{|R|} 
  \end{equation}
  with factors that commute. Moreover, the von Neumann entropy $S(\sigma)= -\Tr(\sigma \ln \sigma)$ satisfies
  \begin{equation}\label{eq:entropy}
  S(\rho_R^\strplq)  = S(\rho^\str_R ) +   S(\rho^\plq_R)-   |R|\ln 2  .
  \end{equation}
  For any operator $O_X\in \BB(\HH_\lsym)$  diagonal in the $X$-basis, and any operator $O_Z\in \BB(\HH_\lsym)$ diagonal in the $Z$-basis:
  \begin{align*}
&\Tr(\rho^\strplq_{\lsym} O_X)= \Tr(\rho^\str_{\lsym} O_X) ,\\
    &\Tr(\rho^\strplq_{\lsym} O_Z)= \Tr(\rho^\plq_{\lsym} O_Z) .
   \end{align*}
  Furthermore, all marginals of the reduced Gibbs state on $R, R' \subseteq \lsym $ commute:
   \begin{equation}\label{eq:commred}
    [\rho^\str_R , \rho^\plq_{R'}] = 0
    \quad \text{and} \quad
    [ \rho^\strplq_R , \rho^\strplq_{R'}] = 0.
  \end{equation}
\end{theorem}
\begin{proof}

Since star and plaquette operators commute, and for any  $s\in \strset_{\lsym}$, we have $e^{\beta A_s} = \cosh(\beta)\identity + \sinh(\beta)A_s$, since $A_s^2=\identity$ and analogously for plaquettes, we can expand the partition function:
  \begin{align}
    Z^\strplq_{\lsym} &= \Tr e^{-\beta H^\str_{\lsym}}e^{-\beta H^\plq_{\lsym}} =  \Tr \left(\prod_{s\in \strset_{\lsym}} (\cosh(\beta)\identity + \sinh(\beta)A_s) \prod_{p\in \plqset_{\lsym}} (\cosh(\beta)\identity + \sinh(\beta)B_p)\right) \notag\\
    &= \cosh(\beta)^{|\strset_{\lsym}| +  |\plqset_{\lsym}|} \sum_{T\subseteq \strset_{\lsym}} \tanh(\beta)^{|T|} \sum_{Q\subseteq \plqset_{\lsym}} \tanh(\beta)^{|Q|} \Tr\left(\bigotimes_{v\in \cobdy T} X_v  \bigotimes_{v' \in \bdy Q}Z_v\right) . \label{eq:HTexp}
  \end{align}
  The last line also used that $\prod_{s\in T}A_s = \bigotimes_{v\in \cobdy T}X_v$ by definition of the coboundary $ dT $ of the star subset. 
  
  Since the Pauli operators are all traceless, the trace in the left side of~\eqref{eq:HTexp} is non-zero if and only if the product $\bigotimes_{v\in \cobdy T} X_v \bigotimes_{v' \in \bdy Q}Z_{v'} = \identity$.
  Expressing the trace in the $ X $-basis, and using the fact that, in this basis, the diagonal elements of any product over Pauli $ Z $-operators is independent of the $ X $-state, we thus arrive at
  \begin{equation*}
    \Tr\left(\bigotimes_{v\in \cobdy T} X_v  \bigotimes_{v' \in \bdy Q}Z_v\right) = \Tr\left(\bigotimes_{v\in \cobdy T} X_v \right)
    \Tr\left(  \bigotimes_{v' \in \bdy Q}Z_v\right)
     2^{-|\lsym|}
  \end{equation*}
  where $2^{-|\lsym|}$ corrects for the second $\Tr \identity$ term if both traces are non-zero. Plugging this back into the expansion of the partition function shows the factorization.
  The factorization of the Gibbs state follows directly from the factorization of the partition function in case $ R = \lsym$.  To decompose the reduced Gibbs state, we perform a similar expansion, replacing $\Tr$ with $\Tr_{R^c}$. We then use
  \begin{equation*}
    \Tr_{R^c}\left(\bigotimes_{v\in \cobdy T} X_v  \bigotimes_{v' \in \bdy Q}Z_v\right) = \Tr_{R^c}\left(\bigotimes_{v\in \cobdy T} X_v \right)
    \Tr_{R^c}\left(  \bigotimes_{v' \in \bdy Q}Z_v\right)
     2^{-|R^c|} , 
  \end{equation*}
  since the trace is only non-zero if both $\cobdy T \cap R^c=\emptyset$ and $\bdy Q \cap R^c = \emptyset$. Together with the decomposition of the partition function, this yields the overall correction factor of $2^{|\lsym|-|R^c|}=2^{|R|}$ in~\eqref{eq:reducedstatedecomp}.
  Furthermore, considering only the star part, we find
  \begin{equation*}
    \Tr_{R^c}\left(\prod_{s\in T} A_s\right) = 1_{\cobdy T\cap R^c=\emptyset} 2^{|R^c|}\prod_{s\in T} A_s ,
  \end{equation*}
  since $\prod_{s\in T} A_s$ has no support on $R^c$ iff $\cobdy T\cap R^c=\emptyset$. Hence, the reduced star Gibbs state is still a polynomial of star operators
  \begin{equation}\label{eq:traceoutGibbs}
    \Tr_{R^c}(e^{-\beta H^\str_{\lsym}}) = \cosh(\beta)^{|\strset_{\lsym}|} 2^{|R^c|} \sum_{T\subseteq \strset_{\lsym}} \tanh(\beta)^{|T|}  1_{\cobdy T\cap R^c=\emptyset} \prod_{s\in T} A_s ,
  \end{equation}
  and thus commutes with any plaquettes. Similarly, the marginal of any plaquette Gibbs state is a polynomial of plaquettes and commutes with any star. This yields~\eqref{eq:commred}.

  For the expectation of $X$- or $Z$-diagonal operators,  the above expansion can be used on $\Tr(e^{-\beta H^\strplq_{\lsym}} O_X)$, which, together with
  \begin{equation*}
    \Tr\left(\bigotimes_{v\in \cobdy T} X_v  \bigotimes_{v' \in \bdy Q}Z_v O_X\right) = \Tr\left(\bigotimes_{v\in \cobdy T} X_v O_X\right)
    \Tr\left(  \bigotimes_{v' \in \bdy Q}Z_v\right)
     2^{-|\lsym|}
  \end{equation*}
  establishes the claim. The last identity is again most easily checked by expressing the trace in the $ X $-basis, and using the fact that, in this basis, the diagonal elements of any product over Pauli $ Z $-operators is independent of the $ X $-state.

 Finally, for the von Neumann entropy, we use the factorization of the Gibbs state and its expectation to find
  \begin{align*}
    S(\rho^\strplq_R) &= -\Tr ( \rho^\strplq_R \ln ( \rho^\str_R)) -\Tr ( \rho^\strplq_R \ln ( \rho_R^\plq)) - \ln 2^{|R|}\\
    &= -\Tr (\rho^\str_R \ln ( \rho_R^\str)) -\Tr ( \rho_R^\plq \ln ( \rho_R^\plq)) - \ln 2^{|R|} ,
  \end{align*}
  since $\ln \rho_R^\str$ and $\ln \rho_R^\plq$ are diagonal in $X$- and $Z$-basis, respectively. 
\end{proof}

The entropy relation~\eqref{eq:entropy} expresses the Maassen-Uffink uncertainty principle \cite{maassen_Generalizedentropicuncertainty_1988} as an identity. It also immediately  implies the relation
\begin{equation*}
I(R:R')_{\rho_\lsym^\strplq} = I(R:R')_{\rho_\lsym^\str} + I(R:R')_{\rho_\lsym^\plq}
\end{equation*}
for the mutual information  $I(R:R')_{\sigma} \coloneqq S(\sigma_{R}) + S(\sigma_{R'}) - S(\sigma_{R\cup R'})$ of two disjoint subsets $ R , R' \subseteq \lsym$.\\

Using the above structure of the Gibbs state, it is easy to see that the expectation value of the logical operators $X_L, Y_L, Z_L$ is $0$ at any positive temperature. In other words, the Gibbs state is maximally mixed in the codespace.
\begin{corollary}\label{cor:logicalexpectation}
    For a CSS Hamiltonian, all logical operators $O_L \in \{X_L, Y_L, Z_L\}$ vanish in expectation over the Gibbs state at any inverse temperature   $0 \leq \beta < \infty$:
    \begin{equation*}
        \Tr(\rho^\strplq_{\lsym} O_L)  = 0 \ .
    \end{equation*}
Moreover, we also have $  \Tr(\rho^\str_{\lsym} X_L) = 0 $ and $  \Tr(\rho^\plq_{\lsym} Z_L) = 0 $. 
\end{corollary}
\begin{proof}
We recall our conventions for the logical operators spelled above \cref{lem:nonselfcorr}.
  For the logical $ X$-operator we compute: 
    \begin{equation}
        \Tr(\rho^\strplq_{\lsym}X_L) = \Tr(\rho^\str_{\lsym} X_L) = \cosh(\beta)^{|\strset_{\lsym}|}\sum_{T\subseteq \strset_{\lsym}} \tanh(\beta)^{|T|} \Tr(X_L\prod_{s\in T}A_s) \ .
    \end{equation}
    This expression is $0$, since $X_L\prod_{s\in T}A_s$ is a product of Pauli $X$-operators. The trace of such a product is $0$ unless it is the identity, i.e.\  unless $X_L = \prod_{s\in T}A_s$.  However, as a logical operator, $X_L$ cannot be expressed as a product of star operators. Analogously $\Tr(\rho^\strplq_{\lsym} Z_L) = \Tr(\rho^\plq_{\lsym} Z_L) =0$. For the expectation of $Y_L=iX_L Z_L$ in $ \rho^\strplq_{\lsym} $, we expand both the star and plaquette Gibbs state, yielding a sum over some constants and traces of the form
    \begin{equation*}
        \Tr(X_L Z_L \prod_{s\in T}A_s \prod_{p\in Q}B_p)
    \end{equation*}
    which is $0$ for the same reasons. 
\end{proof}

\subsection{Dobrushin-Shlosman condition} 
\label{sec:decaycond}

Since the star and similarly the plaquette Gibbs state is equivalent to a 
classical Gibbs measure, one may also relate the reduced states and the local Gibbs state featuring in the DS-condition (\cref{def:dscond}) to this classical measure. We will spell this for the star case. The plaquette case is similar. 

  The marginals of the Gibbs measure $ \mu^\str(\textbf{x}) = \bra{\textbf{x}} \rho^{\str}_{\lsym} \ket{\textbf{x}} $, $ \textbf{x} \in 2^\lsym$, associated with a subset $R\subseteq \lsym$  can be expressed in terms of the reduced Gibbs state on $ R$:
  \begin{equation}\label{eq:defcondGibbs}
    \mu_R^\str(\textbf{x}_{R}) \coloneqq \sum_{\textbf{y} \in 2^{R^c}} \mu^\str(\textbf{x}_R\textbf{y}) = \bra{\textbf{x}_{R}} \rho^{\str}_{R} \ket{\textbf{x}_{R}} , 
  \end{equation}
    where $\textbf{x}_{R} \in 2^{R}$ and we again, by a slight abuse of notation, suppress the dependence on $\lsym$ in the left side.
    
    The conditional expectation, conditioned on the spin configurations on the complement of a  subset $ R' \subseteq \lsym$, which contains $ R \subseteq R' $, is related to the state $ \hat\rho_{R'}^{\str} $ (defined in~\eqref{def:restorGNS} for $ R $ instead of $ R' $):
    \begin{equation}\label{eq:relcondGS}
   \mu^\str_{R(R')^c}(\textbf{x}_{R}|\textbf{x}_{(R')^c}) \coloneqq \frac{\mu^\str_{R(R')^c}(\textbf{x}_{R} \textbf{x}_{(R')^c})}{\mu^\str_{(R')^c}(\textbf{x}_{(R')^c})} = \bra{\textbf{x}_{R}\textbf{x}_{(R')^c}} \Tr_{R'\cap R^c} \hat\rho_{R'}^{\str} \ket{\textbf{x}_{R}\textbf{x}_{(R')^c}} .
    \end{equation}
\begin{figure}
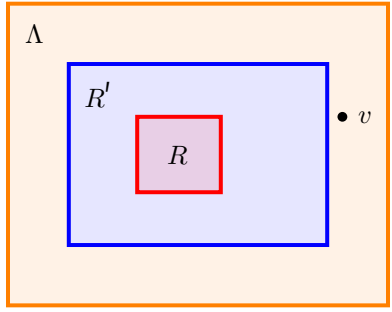

  \centering
  \include{tikz_dobrushin_shlosman}
  \caption{The geometry for \cref{def:DSIIId}. Two nested regions $R\subset R'\subseteq \lsym$ and a spin $v$ in the complement $(R')^c$. In \cref{def:DSIIId}, flipping the spin at $v$ changes the (conditional) Gibbs measure in $R$ by a term decaying exponentially in the distance between $v$ and $R$.}
  \label{fig:dobrushin_shlosman}
\end{figure}
    
    The nested geometry (shown in \cref{fig:dobrushin_shlosman}) with conditioning on the complement of the larger region is in fact exactly what is addressed in the high-temperature conditions by Dobrushin-Shlosman \cite{dobrushin_CompletelyAnalyticalInteractions_1987}. They spelled out a number of equivalent conditions. 
    While in our proof, we require condition \emph{IIId} \cite[Eq. (2.25)]{dobrushin_CompletelyAnalyticalInteractions_1987}, which we recall in the following definition,  any of these conditions would be sufficient.
\begin{definition}\label{def:DSIIId}
    The Gibbs measure $\mu^{\str}$ satisfies \emph{Dobrushin-Shlosmann condition IIId}  if there exist constants $K<\infty, \gamma>0$ such that for all $\lsym$, $R'\subseteq \lsym$, $R\subset R'$, all $\textbf{x}_{R}, \textbf{x}_{(R')^c}$ and all $v\in (R')^c$: 
  \begin{equation*}
    \left| \frac{\mu^{\str}_{R(R')^c}(\textbf{x}_{R}|F_{v}(\textbf{x}_{(R')^c}))}{\mu^{\str}_{R(R')^c}(\textbf{x}_{R}|\textbf{x}_{(R')^c})} -1\right| \leq K e^{-\gamma \dist(v, R)}
  \end{equation*}
  where $F_v$ flips the spin at $v$. 
\end{definition}
The condition for the plaquette Gibbs measure $\mu^\plq$ is identical up to changing $\mathbf{x}\leftrightarrow \mathbf{z}$. 

The following reveals that the DS-condition (\cref{def:dscond}) is in fact the high-temperature condition IIId by Dobrushin-Shlosman for $H^\str_{\lsym}$ or $H^\plq_{\lsym}$ when viewed as a Hamiltonian on classical Ising configurations.  

\begin{proposition}\label{lem:weakDS}
  Let $\wildcard\in \{\str,\plq\}$ and fix some inverse temperature $\beta> 0$. If  $ \mu^\wildcard_\lsym$ satisfies the Dobrushin-Shlosman condition {IIId}, \cref{def:DSIIId}, 
  then there exist constants $K^\wildcard<\infty$, $\xi^\wildcard > 0$ such that $\mathrm{DS}_\beta^{\wildcard}(K^\wildcard, \xi^\wildcard)$ holds on any length scale.
\end{proposition}
\begin{proof} 
In this proof, we drop the $ \str $ from the classical Gibbs measure. 
  Employing~\eqref{eq:relcondGS} in case of $R =U $ and $ R' = UVW$, respectively $ R'= UV $, and setting $C \coloneqq (UVW)^c $,  we arrive at \eqref{eq:DScondprob}.
  Using the rules for the conditional probabilities on its right side, one may write:
  \begin{equation*}
    \frac{\mu_{UC}(\textbf{x}_{U}|\textbf{x}_{C})}{\mu_{UWC}(\textbf{x}_{U}|\textbf{x}_{W}\textbf{x}_{C})} = \sum_{\textbf{y}_W \in 2^{W}} \mu_{WC}(\textbf{y}_W|\textbf{x}_{C}) \frac{\mu_{UWC}(\textbf{x}_{U}|\textbf{y}_W\textbf{x}_{C})}{\mu_{UWC}(\textbf{x}_{U}|\textbf{x}_{W}\textbf{x}_{C})} \ .
  \end{equation*}
  As is shown below, the Dobrushin and Shlosman condition~\emph{IIId} translated to the present situation implies that for some $ K^{\wildcard}, \xi^{\wildcard} \in (0,\infty) $: 
  \begin{equation}\label{eq:dobrushin_summed}
    \left| \frac{\mu_{UWC}(\textbf{x}_{U}|\textbf{y}_W\textbf{x}_{C})}{\mu_{UWC}(\textbf{x}_{U}|\textbf{x}_{W}\textbf{x}_{C})} - 1 \right| \leq K^{\wildcard} e^{-\xi^{\wildcard} \dist(U,W)} 
  \end{equation}
  uniformly in the spin configurations.
  Hence $ \mathrm{DS}_\beta^{\wildcard}(K^\wildcard, \xi^\wildcard)$ holds. 

  In the rest of the proof, we show that \eqref{eq:dobrushin_summed} follows indeed from \cite[Eq. (2.25)]{dobrushin_CompletelyAnalyticalInteractions_1987}.
  To this end, we decompose $W$ by distance to $U$: 
  \begin{equation*}
      W = \bigcupplus_{d=d_{min}}^{d_{max}} W_d,\quad W_d\coloneqq \{v\in W| \dist(v,U)=d\}
  \end{equation*} 
  where $d_{min} = \min_{v\in W} \dist(v,U)$ and $d_{max} = \max_{v\in W} \dist(v,U)$. For any $W_d$, we pick an arbitrary enumeration of spins. We fix two spin configurations $\textbf{x}_W,\textbf{y}_W\in 2^W$, and enumerate all spins in the symmetric difference $\textbf{x}_W\Delta\textbf{y}_W$, that is, those with different configurations, thereby respecting distance and the enumeration on $W_d$ we picked.
  Let $F_{i}$ for $i=1,\ldots, |\textbf{x}_W\Delta\textbf{y}_W|$ be the operation that flips spin $i$. Let $F_{\leq i} = \prod_{j\leq i}F_j$. Then, in particular $\textbf{y}_W= F_{\leq |\textbf{x}_W\Delta\textbf{y}_W|}\textbf{x}_W$.
  We now use a standard telescopic expansion of the fraction in the left-hand side of~\eqref{eq:dobrushin_summed}:
  \begin{align*}
    \frac{\mu_{UWC}(\textbf{x}_{U}|\textbf{y}_W\textbf{x}_{C})}{\mu_{UWC}(\textbf{x}_{U}|\textbf{x}_{W}\textbf{x}_{C})} & =  \prod_{i=1}^{|\textbf{x}_W\Delta\textbf{y}_W|} \frac{\mu_{UWC}(\textbf{x}_{U}|F_{\leq i}(\textbf{x}_W)\textbf{x}_{C})}{\mu_{UWC}(\textbf{x}_{U}|F_{\leq i-1}(\textbf{x}_{W})\textbf{x}_{C})} . 
  \end{align*}
  By assumption $\mu$ satisfies the Dobrushin-Sholsmann condition IIId (\cref{def:DSIIId}).
  For the current geometry this implies that there exist $\hat{K}<\infty, \gamma>0$ such that for all $U,V,W$, all $\textbf{x}_{U}, \textbf{x}_{W}, \textbf{x}_{C}$ and all $i$: 
  \begin{equation*}
    \left| \frac{\mu_{UWC}(\textbf{x}_{U}|F_{\leq i}(\textbf{x}_W)\textbf{x}_{C})}{\mu_{UWC}(\textbf{x}_{U}|F_{\leq i-1}(\textbf{x}_{W})\textbf{x}_{C})} -1\right| \leq \hat{K} e^{-\gamma \dist(i, U)} .
  \end{equation*}
  Inserting this into the left side of \eqref{eq:dobrushin_summed}, we arrive at
  \begin{align*}
    &\frac{\mu_{UWC}(\textbf{x}_{U}|\textbf{y}_W\textbf{x}_{C})}{\mu_{UWC}(\textbf{x}_{U}|\textbf{x}_{W}\textbf{x}_{C})}
    \leq  \prod_{i=1}^{|\textbf{x}_W\Delta\textbf{y}_W|} \left( 1 + \hat{K} e^{-\gamma \dist(i, U)} \right)
    \leq  \prod_{d=d_{min}}^{d_{max}} \left( 1 + \hat{K} e^{-\gamma d} \right)^{|W_d|}\\
    &\leq  \exp \left( \sum_{d=d_{min}}^{d_{max}} |W_d| \ln\left( 1 + \hat{K} e^{-\gamma d} \right)\right)
    \leq  \exp \left( \max_{d} |W_d| \hat{K} e^{-\gamma d_{min}} \sum_{d=0}^{\infty}   e^{-\gamma d} \right)\\
    &\leq 1 +  \frac{\max_{d} |W_d| \hat{K}}{\gamma} e^{-\gamma d_{min}} \exp \left( \frac{\max_{d} |W_d| \hat{K}}{\gamma} e^{-\gamma d_{min}} \right) . 
  \end{align*}
  If $U,V,W$ are rectangles as assummed, and if $\diam(W)\leq 2 \dist(U,W)^2$, then  the size of $W_d$ is bounded by 
  \begin{equation*}
      |W_d|\leq |W|\leq 2^{D} \dist(U,W)^{2D} C_Q 
  \end{equation*} 
  where $C_Q$ is the number of qubits per unit cell. This
  yields~\eqref{eq:dobrushin_summed} as an the upper bound with $ K^{\wildcard} = \frac{2^{D}  \hat{K} C_Q}{\gamma} f_1 e^{f_2}$ and $ \xi^{\wildcard} = \gamma/2 $, where
   \begin{equation*}
       f_1 \coloneqq \max_{d} d^{2D} e^{-\gamma/2 d} \quad   f_2 \coloneq \max_{d} f(d) , \quad  f(d) \coloneqq \frac{2^{D} d^{2D} C_Q \hat{K}}{\gamma} e^{-\gamma d} \ .
  \end{equation*}
  Since this upper bound is independent of the spin configurations, by taking quotients, it also implies a lower bound on the left side in terms of $ 1 - K^{\wildcard} e^{-\xi^{\wildcard} \dist(U,W)} $, which completes the argument for~\eqref{eq:dobrushin_summed}. 
\end{proof}

\subsection{DS-condition at any positive temperature}\label{sec:trivhight}
We show that the DS-condition holds at any positive temperature for models where each qubit is connected to at most two interactions. This class includes both the stars and plaquettes of the $2$D toric code, the stars of the $3$D toric code, and the slightly more general $2$D and $3$D tessellation models, see \cref{fig:code_examples}. 
The key point is that the high-temperature expansion of these models has a particularly simple form, akin to that of the $1$D Ising model. This allows for good control over the marginals of the Gibbs state used in the DS-condition.

\begin{definition}\label{def:pseudo1D}
  For a family of CSS Hamiltonians on a $D$-dimensional lattice, we say that the star part admits a \emph{trivial high-temperature expansion} if, for any $\lsym\in \FF$ and any star-connected region $R\subseteq \lsym$, the marginal of the (unnormalized) Gibbs state on $R$ satisfies at any $ \beta < \infty $:
  \begin{equation}\label{eq:pseudo1Dcase1}
    \Tr_R e^{-\beta H^\str_R}
    = 2^{|R|} \cosh(\beta)^{|\strset_R|} \left( \identity + 1_{\cobdy \strset_R \cap R=\emptyset}  \prod_{s\in \strset_R} \tanh(\beta)A_s \right) .
  \end{equation}
  The definition for the plaquette part is analogous.
\end{definition}
For any star-connected region $ R $ the condition $ \cobdy \strset_R \cap R=\emptyset $ on the coboundary of the star set of $ R $ can be fulfilled or not. 
Generalizations to non-connected $ R $ are possible and would involve one term of the type \eqref{eq:pseudo1Dcase1} per connected component. 

To derive the DS-condition for models with a trivial high-temperature expansion, we further require that all rectangles are connected. The main obstacle to further generalizations is a large number of small connected components at the boundary of the rectangle. While one could deal with a constant number of connected components, or require it only for large enough rectangles, we opt not to for clarity.

\begin{theorem}\label{thm:trivexpansionDS}
  For a translation-invariant family of CSS Hamiltonians on a $D$-dimensional lattice and $\wildcard \in \{\str, \plq\}$, assume that the $\wildcard$-part has a trivial high-temperature expansion and that all rectangles are $\wildcard$-connected. 
  Then there exist constants $K^\wildcard<\infty$, $\xi^\wildcard > 0$ such that $\mathrm{DS}_\beta^{\wildcard}(K^\wildcard, \xi^\wildcard)$ holds at any $ \beta < \infty $ on any length scale. 
\end{theorem}
The proof of this theorem is given at the end of the section.
Let us first relate the trivial high-temperature expansion to the number of interactions per qubit.
\begin{corollary}\label{cor:interactioncount}
  For a family of CSS Hamiltonians on a $D$-dimensional lattice, $\beta<\infty$,
  assume that for all $\lsym \in \FF $ and any $v\in \lsym$ the number of stars supported on $v$ is at most $2$, i.e.\ $|\bdy v |\leq 2$. Furthermore, assume that all rectangles are star-connected. 
  Then, the star part satisfies the DS-condition on length scale $L^\str_0$. Similarly, if $|\cobdy v|\leq 2$ for all $v\in \lsym$, and if all rectangles are plaquette-connected, then the plaquette part satisfies the DS-condition.
\end{corollary}
\begin{proof}
  Using \cref{thm:trivexpansionDS}, it is sufficient to show that the model admits a trivial high-temperature expansion. We give the proof for the star case; the plaquette case is analogous.
  Similarly to~\eqref{eq:traceoutGibbs}, we expand the marginal of the unnormalized Gibbs state:
  \begin{equation*}
    \Tr_R e^{-\beta H^\str_R} = 2^{|R|} \cosh(\beta)^{|\strset_R|} \sum_{T\subseteq \strset_R} 1_{\cobdy T\cap R=\emptyset} \prod_{s\in T} \tanh(\beta) A_s \ .
  \end{equation*}
  The condition $\cobdy T \cap R = \emptyset$ holds only if $T=\emptyset$ or $T=\strset_R$. 
  Assume $\cobdy T \cap R = \emptyset$ but neither  $T=\emptyset$ nor $T=\strset_R$. Then, there are $s\in T, s'\in \strset_R\setminus T$. 
  Since $R$ is star-connected, there is a finite path $s,v_1,s_2,v_2, \ldots, v_k, s'$ connecting $s$ to $s'$ with $v_i \in R$ for all $i$. Since $s\in T, s'\notin T$ there exists some $v_i$ with $s_{i-1}\in T $ and $s_{i}\notin T$. Since, by assumption, there are only two stars connected to $v_i$, it is connected to an odd number of stars in $T$ and hence $v_i\in dT\cap R$, which is a contradiction. 
\end{proof}

The trivial high-temperature expansion is illustrated well with the $1$D Ising model on a chain, which can be viewed as a trivial CSS code with Ising interactions as stars and no plaquettes. 
The marginals of its Gibbs state on an interval $R$ contain two terms, one proportional to the identity and the other is the product of all interactions, which is a Pauli-$X$ on both ends of the interval; interactions cannot form loops. This is in stark contrast to the $2$D Ising model, where the expansion contains an extensive number of closed loops of interactions which multiply to identity.

The more relevant examples of CSS codes, however, are the $2$D and $3$D toric code and the wider class of tessellation models \cref{fig:code_examples}. 
For a $2$D tessellation model, such as the $2$D toric code, qubits are placed on edges of a graph, stars on vertices and plaquettes on faces. Thus, there are at most two stars and two plaquettes connected to a single qubit. Since, by assumption, every unit cell is connected, so are rectangles. Thus,  \cref{cor:interactioncount} applies to both the star and plaquette parts and they satisfy the DS-condition at any positive temperature. 

For a $3$D tessellation model, such as the $3$D toric code, qubits are also placed on edges of a graph, with stars on vertices and plaquettes on faces. 
The star part of these models thus satisfies the DS-condition at any positive temperature, using \cref{cor:interactioncount} and the fact that each edge (i.e.\ qubit) is connected to two vertices (i.e.\  stars). The plaquette part has, in general, many non-trivial loops in the high-temperature expansion, similar to the $2$D Ising model, and can not be analysed with this method.

\begin{proof}[Proof of \cref{thm:trivexpansionDS}]
  Let $U, V, W\subseteq \lsym$ be as in \cref{def:dscond} and let $\beta<\infty$.
  Since, by assumption, any rectangle is star-connected, $UV$, $VW$, and $UVW$ are star-connected. 
  The rectangle $V$ has either one or two connected components; the latter occurs if $UVW$ has periodic boundary conditions (see \cref{fig:approxtensor}). First, consider the case where $V$ is star-connected. 
  We use that $H^\str_{UVW}-H^\str_{VW} = H^\str_{UV}-H^\str_{V}$ has no support on $VW$ to rewrite the fraction in the DS-condition:
  \begin{align*}
    &\frac{\Tr_{VW} \hat{\rho}^\str_{UVW}}{\Tr_{V} \hat{\rho}^\str_{UV}}= \frac{\Tr_{VW} e^{-\beta H^\str_{VW}}}{\Tr_{UVW} e^{-\beta H^\str_{UVW}}} \frac{\Tr_{UV} e^{-\beta H^\str_{UV}}}{\Tr_{V} e^{-\beta H^\str_{V}}} .
  \end{align*}
  On each of the traces on the right side, we use the trivial high-temperature expansion \eqref{eq:pseudo1Dcase1}.
  To simplify notation let $\tau =\tanh(\beta)$ and $A^{\strset_R}\coloneqq 1_{\cobdy \strset_R \cap R=\emptyset}  \prod_{s\in \strset_R}A_s$, then the above fraction equals:
  \begin{align*}
    \frac{2^{|VW|} \cosh(\beta)^{|\strset_{VW}|}}{2^{|UVW|} \cosh(\beta)^{|\strset_{UVW}|}}
    \frac{\identity + \tau^{|\strset_{VW}|}A^{\strset_{VW}} }{\identity + \tau^{|\strset_{UVW}|}A^{\strset_{UVW}}} 
    \cdot
    \frac{2^{|UV|} \cosh(\beta)^{|\strset_{UV}|}}{2^{|V|} \cosh(\beta)^{|\strset_{V}|}}
    \frac{\identity + \tau^{|\strset_{UV}|}A^{\strset_{UV}} }{\identity + \tau^{|\strset_{V}|}A^{\strset_{V}}} &\\
    =\frac{\identity + \tau^{|\strset_{VW}|}A^{\strset_{VW}} }{\identity + \tau^{|\strset_{UVW}|}A^{\strset_{UVW}}}
    \frac{\identity + \tau^{|\strset_{UV}|}A^{\strset_{UV}} }{\identity + \tau^{|\strset_{V}|}A^{\strset_{V}}} & ,
  \end{align*}
  where the last step also used $|\strset_{UV}| + |\strset_{VW}| = |\strset_{UVW}| + |\strset_{V}|$.
  This expression is diagonal in the $X$ basis. Moreover, since the spectrum of $A^{\strset_R}$ is one of $\{\pm 1\}, \{1\}, \{0\}$ for any $R$,
  we can bound the norm by replacing any $A^{\strset_R}$ by either $ +1 $ in the numerator and $- 1$ in the denominator, hence taking the maximum over all such configurations.   This yields
  \begin{align*}
    \left\| \frac{\Tr_{VW} \hat{\rho}^\str_{UVW}}{\Tr_{V} \hat{\rho}^\str_{UV}} -1 \right\| 
    & \leq \frac{\tau^{|\strset_{VW}|} + \tau^{|\strset_{UV}|} + \tau^{|\strset_{UV}|+|\strset_{VW}|} + \tau^{|\strset_{UVW}|} +\tau^{ |\strset_{V}|} + \tau^{|\strset_{UVW}| + |\strset_{V}|}}
    {\left( 1- \tau^{|\strset_{UVW}|} \right) \left(1-\tau^{|\strset_{V}|}\right)}
    \notag \\ 
    & \leq \frac{ 6 \tau^{ |\strset_{V}|}}
    { \left(1-\tau^{|\strset_{V}|}\right)^2} , 
  \end{align*}
    since $ \tau \in [0, 1) $.  The last term is upper-bounded by 
    $
     6 \left(1-\tau\right)^{-2} e^{-|\ln(\tau)| |\strset_{V}|}
   $. 
   
  If $V$ consists of two star-separated, connected components, $V=V_1\uplus V_2$, the Gibbs state on $V$ factorizes $\Tr_V e^{-\beta H^{\str}_V} = \Tr_{V_1} e^{-\beta H^{\str}_{V_1}} \otimes \Tr_{V_2} e^{-\beta H^{\str}_{V_2}} $. The other terms remain unchanged. Proceeding as above, we obtain three factors in the denominator and hence find
  \begin{align*}
    \left\| \frac{\Tr_{VW} \hat{\rho}^\str_{UVW}}{\Tr_{V} \hat{\rho}^\str_{UV}} -1 \right\| &\leq \frac{ 6 }
    { \left(1-\tau\right)^3} e^{-|\ln(\tau)| |\strset_{V}|} \ .
  \end{align*}
  Taking the maximum over both cases yields $K^\str = 6 (1-\tau)^{-3}$ for \eqref{eq:corrdecaywildcard}.

  To obtain a decay rate  $\xi^\str> 0$, we bound the number of stars on $ V $:
  \begin{equation*}
    |\strset_{V}| \geq \frac{|V|}{\max_{s\in \strset_V}|\cobdy s|}\geq \frac{\dist(U,W)}{C_Q \max_{s\in \strset_V}|\cobdy s|} , 
  \end{equation*}
  where the maximal number of qubits in a star, $\max_{s\in \strset_V}|\cobdy s|$ is bounded for translational invariant systems and where $C_Q<\infty$ is the number of qubits per unit cell. Thus
  \begin{equation*}
      \xi^\str = \frac{|\ln(\tau)|}{C_Q \max_{s\in \strset_V}|\cobdy s|} 
  \end{equation*}
  and the model satisfies $DS^\str(K^\str, \xi^\str)$ at any length scale.
\end{proof}

\section{Davies Lindbladians}\label{sec:davies}

We recall from~\eqref{def:LL}--\eqref{def:jumpops} the definitions of the Davies Lindbladians for CSS codes. 
In this section, we will briefly gather some properties of these operators which are of relevance for our proof of the main result (see also \cite{davies_Generatorsdynamicalsemigroups_1979,alicki_Quantumdynamicalsemigroups_1987,temme_ThermalizationTimeBounds_2017, alicki_thermalizationKitaevs2D_2009,breuer_theoryopenquantum_2007,attal_OpenQuantumSystems_2006a} for more information). In particular, we show that the full Lindbladian $ \LL_\lsym^\strplq = \LL_\lsym^\str + \LL_\lsym^\plq$ is the sum of two commuting parts.

\subsection{Self-adjointness and stationary states}

While the Lindbladians,
$$
\LL_R^\wildcard = \sum_{v\in R } \LL_v^\wildcard , \qquad \wildcard \in \{ \strplq, \str, \plq \} ,
$$
with $ \LL_v^\wildcard $ from~\eqref{def:LLv} and $ R \subseteq \lsym$ arbitrary, are generally not self-adjoint with respect to the Hilbert-Schmidt inner product, they do share this property if one endows the space $ \BB(\HH_\lsym) $ with a scalar product  with respect to the Gibbs equilibrium state $ \rho_\lsym^\strplq$ and any interpolation parameter $ s \in [0,1]$:
\begin{equation}\label{def:sscalar}
    \langle O_1, O_2 \rangle_{\rho_\lsym^\strplq, s} \coloneq  \Tr O_2 (\rho_\lsym^\strplq)^s O_1^\dagger (\rho_\lsym^\strplq)^{1-s}  . 
\end{equation}
The case $ s= \tfrac{1}{2}$ corresponds for the KMS-scalar product, and $ s = 1 $ is the GNS-scalar product. 

Properties of the Lindbladian, such as its self-adjointness with respect to the above family of inner products, are most easily seen by identifying it with a symmetric Dirichlet form. To do so, one fixes $v \in \lsym  $ and recalls that the spectral projections, which enter the definitions~\eqref{def:jumpops} of the jump operators, are polynomials in the star and plaquette operators:
\begin{align}\label{eq:repPpol}
    P^\str_{\bdy v}(\omega) &=  \sum_{\substack{\mathbf{a}\in \{\pm 1\}^{\bdy v}\\ \sum_{s\in \bdy v} \mathbf{a}_s = -\omega/2}} \prod_{s\in \bdy v} \frac{\identity + \mathbf{a}_s A_s}{2} \notag \\
    P^\plq_{\cobdy v}(\omega) &=  \sum_{\substack{\mathbf{b}\in \{\pm\}^{\cobdy v}\\ \sum_{p\in \cobdy v} \mathbf{b}_p = -\omega/2}} \prod_{p\in \cobdy v} \frac{\identity + \mathbf{b}_p B_p}{2} .
  \end{align}
 This allows us to verify that the jump operators are eigenvectors of the modular operator corresponding to $ \rho_\lsym^\strplq$, i.e., in the star case for any $ s \in \mathbb{R} $: 
  \begin{align*}
    L_v^{\str}(\omega)(\rho^{\strplq}_{\lsym})^s &= Z_v P^\str_{\bdy v}(\omega) (\rho^{\strplq}_{\lsym})^s\\
    &= \frac{e^{-\beta s (H_{\lsym}^{\strplq}-H^\str_v)}}{(Z^\strplq_\beta)^s}Z_v P^\str_{\bdy v}(\omega)e^{-\beta s H^\str_v} \\
    &= (\rho^{\strplq}_{\lsym})^s e^{2\beta s H^\str_v}Z_v P^\str_{\bdy v}(\omega) \\
    &= (\rho^{\strplq}_{\lsym})^s Z_v e^{-2\beta s H^\str_v} P^\str_{\bdy v}(\omega)
    = e^{-\beta s \omega}(\rho^{\strplq}_{\lsym})^s L_v^{\str}(\omega)
  \end{align*}
  where we also used that $Z_v$ flips the sign of exactly the stars connected to $v$ and that $P^\str_{\bdy v}(\omega)$ are the spectral projections of $2H_v^\str$ of eigenvalue $-\omega$. This also implies, e.g.\ that $ L_v^{\str, \dag}(\omega)L^{\str}_v(\omega)$ commutes with $(\rho^\strplq)^s$. The relation
  \begin{equation}\label{eq:adjumps}
    L_v^{\str,\dag}(\omega) = P^\str_{\bdy v}(\omega) Z_v = Z_v P^\str_{\bdy v}(-\omega) = L_v^{\str}(-\omega) \ ,
  \end{equation}
  together with the detailed balance condition \eqref{eq:jumprate} may be used to write the Lindbladian as 
  \begin{equation}\label{eq:single_site_lindbladian}
     \LL_v^{\str}(O) = \frac{1}{2} \sum_\omega h_v^\str(\omega)  \left(  L_v^{\str,\dag}(\omega) \big[O, L_v^{\str}(\omega)  \big] 
     + e^{-\beta \omega} \big[L_v^{\str}(\omega) ,O \big] L_v^{\str,\dag}(\omega) \right) ,
  \end{equation}
where the summation is over the eigenvalues of $ -2\sum_{s\in \partial v} A_s $ which is symmetric with respect to changing $ \omega \to - \omega $. A similar result applies to the plaquette part.

Using the above properties and following the arguments in \cite[Lemma 5.2]{carlen_Gradientflowentropy_2017}, one thus arrives at the advertised Dirichlet form.
\begin{proposition}[cf.~\cite{carlen_Gradientflowentropy_2017}] 
    For any $ s \in [0,1 ] $, all $ v \in \lsym $ and both $ \wildcard \in \{ \str, \plq \} $ the single-site Lindbladian $ \LL_v^\wildcard $ corresponds to the Dirichlet form
\begin{equation}\label{eq:Dirichlet}
\frac{1}{2} \sum_\omega h_v^\wildcard(\omega) e^{\beta \omega (s-1)}  \left\langle \big[ L_v^\wildcard(\omega), O_1\big] , \big[ L_v^\wildcard(\omega), O_2\big] \right\rangle_{\rho_\lsym^\strplq, s} = - \left\langle O_1, \LL_v^\wildcard(O_2) \right\rangle_{\rho_\lsym^\strplq, s}
\end{equation}
defined on all $ O_1,O_2 \in \BB(\HH_\lsym)$. 
\end{proposition}
An immediate consequence is the following 
\begin{corollary}\label{pro:fullrankfixedpoint}
    Let $ R \subseteq \lsym $. Then the full Davies Lindbladians of a CSS code and its parts, $\LL^\wildcard_R$ with $\wildcard \in \{\str, \plq,\strplq\}$ share the following properties:
   \begin{enumerate}
\item Self-adjointness with respect to $  \langle \cdot , \cdot \rangle_{\rho_\lsym^\strplq, s} $ for any $ s \in [0,1]$. 
\item The full Gibbs state is a stationary state:  $
    \LL^{\wildcard,*}_R(\rho_\lsym^{\strplq}) = 0 $.
\end{enumerate} 
For both $ \wildcard \in \{ \str, \plq \} $, the $ \wildcard $-part of the full Gibbs state is also a stationary state of its Linbladian, $\LL^{\wildcard,*}_R(\rho_\lsym^{\wildcard}) = 0 $, and the same applies to its $ R' $-local version (defined in~\eqref{def:restorGNS}) 
    $$
    \LL^{\wildcard,*}_R(\hat{\rho}_{R'}^{\wildcard}) = 0 
    $$
    as long as $ R \subseteq R' $.
\end{corollary}
\begin{proof}
1. The self-adjointness of $ \LL_v^\wildcard $ with $\wildcard \in \{\str, \plq\} $ is immediate from~\eqref{eq:Dirichlet}. The claim then follows by linearity.

\noindent
2. We use self-adjointness with respect to the GNS-inner product ($s=1$) and the unitality of the semigroup to conclude
          \begin{align*}
            \Tr(O \LL^{\wildcard, *}_R(\rho^{\strplq}_{\lsym})) = \Tr(\LL^{\wildcard}_R(O)\rho^{\strplq}_{\lsym} \identity) = -\Tr(O\rho^{\strplq}_{\lsym} \LL^{\wildcard}_R(\identity))=0
          \end{align*}
          for any $O\in \BB(\HH_{\lsym})$.

\noindent
For the proof of the remaining assertions, we focus on the $ \str $-case, and recall from \cref{thm:equi:main} that $ \rho^{\strplq}_{\lsym} = 2^{|\lsym|} \rho^{\plq}_{\lsym}\rho^{\str}_{\lsym}$. Since the jump operators $ L_v^\str (\omega)$ and their adjoints commute with the plaquette operators $ B_p $ and hence $ \rho^{\plq}_{\lsym} $, the second item implies that 
\begin{equation*} 
0 = \LL^{\str, *}_R(\rho^{\strplq}_{\lsym}) = 2^{|\lsym|} \rho^{\plq}_{\lsym} \LL^{\str,*}_R(\rho_\lsym^{\str}) 
\end{equation*}
and hence $  \LL^{\str,*}_R(\rho_\lsym^{\str}) = 0 $. 

Finally, we may write  $ e^{-\beta H_\lsym^\str} = e^{-\beta (H_{\lsym}^\str - H_{R'}^\str)} e^{-\beta H_{R'}^\str} $. The first factor commutes with the jump operators $ L_v(\omega)^\str $ with $ v \in R $ and its adjoint. Since it is not supported on $R$ it commutes with all $Z_v$ for $v\in R$ and, since they are built from stars, it commutes with all star projections $P^\str_{\bdy v}(\omega)$.
 By a similar argument as above, one hence concludes $  \LL^{\str,*}_R( e^{-\beta H_{R'}^\str} ) = 0 $, and hence the claim.
\end{proof}
The second assertion in this corollary is a weak form of frustration-freeness. Subsequently, we also need a stronger statement: any fixed point of the Lindbladian on  $R\subset \lsym$ is a fixed point of the Lindbladian on any subregion $R'\subseteq R$. This will be established as the last item in \cref{pro:condexpprop} below.

\subsection{Commuting star and plaquette parts}
Splitting the Lindbladian into a sum of two Lindbladians is possible in general. In our case, the two parts satisfy a much stronger condition: they commute. As in the equilibrium setting, this allows us to decompose most properties into a star and a plaquette part. Unlike for the Gibbs state, this, however, does not mean that the two Lindbladian parts act trivially as an entirely classical dynamics on general states. Only for density matrices diagonal in the $ X $- or $ Z $-basis they behave as Glauber dynamics.  
\begin{lemma}\label{lem:lindcom}
  Let $R,R'\subseteq \lsym$. Then $\LL^\str_R$ and $\LL^\plq_{R'}$ commute.
\end{lemma}
\begin{proof}
It suffices to show that the single-site Lindbladians, $\LL_v^\str$ and $\LL_{v'}^\plq$, commute for any two sites $v,v'\in \lsym$.  This is immediate from the representation~\eqref{eq:single_site_lindbladian} together with the fact that
the jump operators either commute or anti-commute:
  \begin{align}\label{eq:commjump}
    L_{v'}^{\plq}(\omega') L_v^{\str}(\omega) \ & =  X_{v'} P^\plq_{\cobdy v'}(\omega') Z_v P^\str_{\bdy v}(\omega)  = X_{v'} Z_v P^\plq_{\cobdy v'}(\omega')  P^\str_{\bdy v}(\omega) \notag \\  &= (-1)^{\delta_{v,v'}} L_v^{\str} (\omega) L_{v'}^{\plq}(\omega') \ .
  \end{align}
  This follows from the commutation rules of the constituents: as polynomials in star and plaquette operators (cf.~\eqref{eq:repPpol}), the projections $P^\str$ and $P^\plq$ commute. Moreover, $P^\str$ commutes with all $X$-operators, and $P^\plq$ commutes with all $Z$-operators. The same relation holds if we replace one or both jumps by their adjoints, using \eqref{eq:adjumps}. The jump operators thus share the commutation relation of $Z_v$ and $X_{v'}$. 
\end{proof}

\subsection{Kernels as commutants}
Determining the kernels of the parts of the Davies Lindbladian will be key for our explicit formulas for the conditional expectations in \Cref{sec:condexp}. 
\begin{proposition}\label{pro:lindbladkern}
  For $R, R'\subseteq \lsym$, the kernels of the star and plaquette Lindbladians are given by
  \begin{equation}\label{eq:lkern}
    \ker{\LL^{\str}_R} = \left\{Z_v, \sum_{s\in \bdy v} A_s \right\}'_{v\in R} \quad \text{and} \quad
    \ker{\LL^{\plq}_R} = \left\{X_v, \sum_{p\in \cobdy v} B_p \right\}'_{v\in R}
  \end{equation}
  where $\{\ \cdot \ \}'$ denotes the commutant.
\end{proposition}
\begin{proof} 
  The kernel of a Lindbladian with a full-rank fixed point is known \cite[Thm. 7.2]{wolf_QuantumChannelsOperations_2012} to be given by the commutant of the set of jump operators and their adjoints. Since the Gibbs state $ \rho_\lsym^{\strplq} $ is a full rank fixed point (\cref{pro:fullrankfixedpoint}) and the set of jump operators is self-adjoint by~\eqref{eq:adjumps}, we have:
  \begin{equation*}
    \ker{\LL^{\str}_R} = \left\{L_v(\omega)^\str \right\}'_{v\in R, \omega}
    =\left\{Z_vP_{\bdy v}^\str(\omega) \right\}'_{v\in R, \omega} \ .
  \end{equation*}
  We now use the fact that the commutant of a set is equal to the commutant of the algebra generated by this set to replace the jump operators with simpler operators.
  In a first step, we express  
  \begin{equation*}
    \sum_{\omega} L^\str_v(\omega) = Z_v \sum_{\omega} P_{\bdy v}^\str(\omega) = Z_v \ ,
  \end{equation*}
  since $P_{\bdy v}^\str(\omega)$ are the spectral projections of $-2\sum_{s\in\bdy v} A_s$.
  This provides the first inclusion ($\subseteq$) of the equality
  \begin{equation*}
    \left\{L_v(\omega)^\str \right\}'_{v\in R, \omega}
    =\left\{Z_v, P_{\bdy v}^\str(\omega) \right\}'_{v\in R, \omega} \ .
  \end{equation*}
  The reverse inclusion ($ \supseteq $) follows by multiplying the projection by $Z_v$ from the left.

  In a second step, we again use that $P_{\bdy v}^\str(\omega)$ are spectral projections to obtain the sum of the star operators:
  \begin{equation*}
    \sum_{s\in\bdy v} A_s = \sum_{\epsilon} \epsilon P_{\bdy v}^\str(-2\epsilon) \ .
  \end{equation*}
  Here the sum is over all eigenvalues $\epsilon$ of $\sum_{s\in\bdy v} A_s$. This yields the first inclusion ($\subseteq$) in the equality
  \begin{equation*}
    \left\{Z_v, P_{\bdy v}^\str(\omega) \right\}'_{v\in R, \omega} =\left\{Z_v, \sum_{s\in\bdy v} A_s \right\}'_{v\in R}\ .
  \end{equation*}
  The reverse inclusion ($ \supseteq $) follows from the representation~\eqref{eq:repPpol} 
  of the spectral projection $ P_{\bdy v}^\str(\omega) $ as a polynomial. By the same argument, we also show the identity claimed for $\ker\mathcal{L}_R^\plq$.
\end{proof}

\section{Conditional expectations}\label{sec:condexp}
Key quantities in the study of the long-time limits of the Davies Lindbladian will be the projections related to their kernels, which we study in this section.  
\begin{definition}
  Let $\wildcard \in \{\str,\plq,\strplq\}$ and $ R \subseteq \lsym $. Then, the local Lindbladian projector, or conditional expectation, $\E^{\wildcard}_R$ of $\LL_R$ is the infinite-time limit defined as
  \begin{equation}\label{def:condexp}
    \E^{\wildcard}_R(O) \coloneqq \lim_{t\to \infty} e^{t\LL^{\wildcard}_R}(O)
  \end{equation}
  for any $O \in \BB(\HH_{\lsym})$.
\end{definition}
The limit in~\eqref{def:condexp} exists in either case $ \wildcard \in \{\str,\plq,\strplq\}$  since the Davies Lindbladians are self-adjoint with respect to $  \langle \cdot , \cdot \rangle_{\rho_\lsym^\strplq, s} $.  As the name suggests, the maps defined through this limit share all the defining properties of a conditional expectation, which we recall in the following~\cite{takesaki_TheoryOperatorAlgebras_2003}. 
\begin{definition}
  Let $\NN\subseteq \BB(\HH_{\lsym})$ be a von Neumann subalgebra. A completely positive, unital map $E_{\NN}:\BB(\HH_{\lsym})\to \NN$ is called a conditional expectation  onto $\NN$ if
  \begin{align}\label{eq:condexpdef}
      \forall O \in \NN &\colon \ E_{\NN}(O) = O \notag\\
      \forall l, r \in \NN,  O \in \BB(\HH_{\lsym})&\colon \ E_{\NN}(l O r)= l E_{\NN}(O)r \ .
  \end{align}
\end{definition}
Here, the two conditions in \eqref{eq:condexpdef} are in fact equivalent.
Basic properties of the local Lindbladian projectors are collected in the following:
\begin{proposition}\label{pro:condexpprop}
  For any  $\wildcard \in \{\str,\plq,\strplq\}$ and $R\subseteq \lsym$:
  \begin{enumerate}
    \item $\E^{\wildcard}_R$ is unital and completely positive,
    \item $\E^{\wildcard}_R$ is a conditional expectation onto the von Neumann subalgebra $\ker{\LL^{\wildcard}_R}$. In particular, it is an orthogonal projection.
    \item $\E^{\wildcard}_R$ is self-adjoint with respect to the inner product $\langle \cdot,\cdot\rangle_{\rho^\strplq_\lsym, s}$ for any $ s \in [0,1]$.
    \item\label{item4} $\E^{\wildcard}_R\circ\E^{\wildcard}_{R'} = \E^{\wildcard}_{R'}\circ \E^{\wildcard}_R = \E^{\wildcard}_R$ for any $R'\subseteq R$ .
  \end{enumerate}
\end{proposition}
    The first three items are well known and follow the self-adjointness of the Lindbladian (\cref{pro:fullrankfixedpoint}) and the unitality and complete positivity of the semigroup. The last item is well known for the full Davies Lindbladian (cf.~\cite{kastoryano_NoncommutativeNashInequalities_2015,kochanowski_RapidThermalizationDissipative_2025}). The proofs for the star and plaquette parts are analogous. We spell it out for completeness.
\begin{proof}[Proof of~\ref{item4}.~in \cref{pro:condexpprop}]
Recall that $\E^\wildcard_{R}$ and $\E^\wildcard_{R'}$ are projections onto the kernels of $\LL^\wildcard_{R}$ and $\LL^\wildcard_{R'}$, respectively. By \cref{pro:lindbladkern}, the commutant of the set of jump operators gives these kernels. The kernel of the Lindbladian on the bigger set is hence contained in the kernel of the one on the smaller set: $   \ker \LL^\wildcard_{R} \subseteq \ker \LL^\wildcard_{R'} $.
Consequently, the image of $\E^\wildcard_{R}$ is contained in the image of $\E^\wildcard_{R'}$. In particular, $ \E^\wildcard_{R'}\circ \E^\wildcard_{R} = \E^\wildcard_{R} = \E^\wildcard_{R}\circ \E^\wildcard_{R'} $, since these operators are orthogonal projections on nested spaces. 
This completes the proof.
\end{proof}

Having control over the conditional expectations $\E^\str_R$ and $\E^\plq_R$ is a key ingredient in the proof of the MLSI. As a main result of this section, we will derive simple, explicit expressions of these conditional expectations. These are then used to establish other properties, which will be very helpful in \Cref{sec:approxtensor}.

\subsection{Partition by support}

As preparation for the explicit representations of the conditional expectations, we regroup the terms in the star and plaquette part of the Hamiltonian. This will be done by partitioning the sets of stars or plaquettes of a subset $R\subseteq \lsym$ according to the support of the respective operators within this subset. 

\begin{figure}[h]
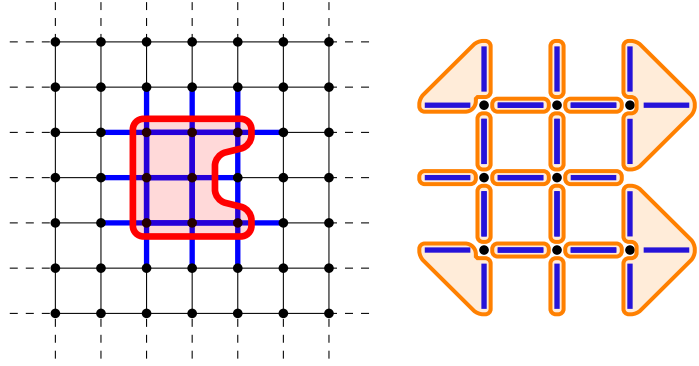

  \centering
  \include{tikz_ising_partition}
  \caption{A subset of the Ising model and the corresponding partition of the interactions. For illustration purposes, the Ising interactions are taken as the star part of a CSS code. Plaquette interactions are not discussed. Left: Qubits (black dots) and interactions (lines) of the $2$D Ising model. A subset of qubits $R$  (red box) and the interactions $\strset_R$ (blue, thick lines) with support intersecting $R$.  Right: Only the qubits in $R$ and their interactions. The interactions are partitioned by support in $R$ (orange sets). Most sets of this partition contain only a single interaction, only sets at the boundary contain multiple interactions. }
  \label{fig:ising}
\end{figure}

\noindent Given any non-empty subset $r\subseteq R$, we denote the set of stars whose support in $R$ is exactly $r$ by
\begin{equation}
  \strset_{r|R} \coloneqq \left\{s\in \strset_R \ \middle| \ \cobdy s\cap R = r\right\} \ .
\end{equation}
Analogously, for the plaquettes, we define
$
  \plqset_{r|R} \coloneqq \left\{p\in \plqset_R\middle| \bdy p\cap R = r\right\} $. 
In general, there are many subsets $r\subseteq R$ such that $\strset_{r|R}=\emptyset$ or $\plqset_{r|R}=\emptyset$. We stress that we do not allow $r=\emptyset$, and abbreviate the set of supports by
$$J^\str_R \coloneqq \left\{ r \subseteq R \middle| r \neq \emptyset \wedge \strset_{r|R}  \neq \emptyset \right\} \subset2^R ,$$ 
and similarly for plaquettes. 
The following lemma shows that this set partitions the set of stars disjointly. 
It also ensures that enlarging a set from $R'$ to $R$ only refines this partition. An illustration of these partitions for the case of Ising interactions (as the stars) can be found in \cref{fig:ising}. 
The last item, which will become important in  \cref{lem:pinchingoverlap} below, expresses the fact that if $r$ is far enough from the boundary of two sets, its star set looks identical to a star set of the bulk.
\begin{lemma}\label{lem:partitionproperties}
  For any non-empty $R\subseteq \lsym$:
  \begin{enumerate}
    \item $\displaystyle \bigcupplus_{r\in J_R^\str} \strset_{r|R} = \strset_R$,
    \item For any $R'\subseteq R$, and any  $r\in J^\str_R$, $r'\in J^\str_{R'}$:\begin{equation*}
        \strset_{r|R}\subseteq \strset_{r'|R'} \quad \mathrm{or} \quad \strset_{r|R}\cap \strset_{r'|R'} = \emptyset . 
    \end{equation*}
    \item Let $R_1, R_2\subseteq \lsym$ be two sets of qubits. For $i=1,2$, let 
    \begin{equation}\label{def:Rminus} 
    R_i^{-}\coloneqq\{v\in R_i\mid \dist(v, (R_1\cup R_2)\setminus R_i )>2\}
    \end{equation}
    be the interior of $R_i$ with respect to $R_1\cup R_2$. Then for $i=1,2$ and any $r\in J_{R_1\cap R_2}^\str$:
  \begin{equation}\label{eq:pinchingoverlap:proof:implication}
    r\cap R_i^{-} \neq \emptyset \  \Rightarrow  \ r\subseteq R_i \ \mathrm{and} \ \strset_{r|R_1\cup R_2} = \strset_{r|R_i} \ .
  \end{equation}
  \end{enumerate}
  Analogous results apply to plaquettes.
\end{lemma}
  \begin{figure}[h]
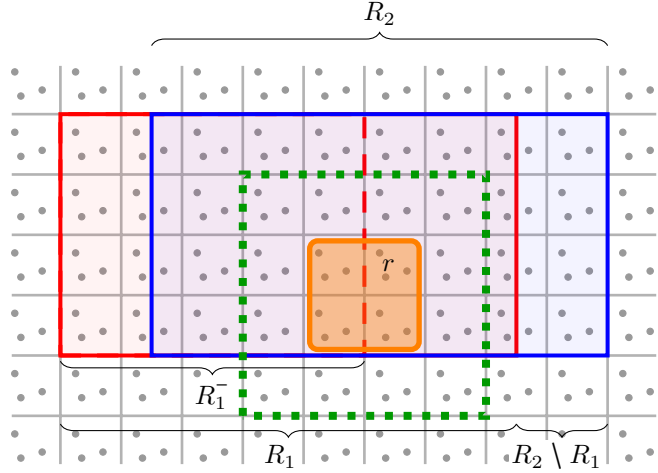

    \centering
    \include{tikz_partition_overlap}
    \caption{The geometry for \eqref{eq:pinchingoverlap:proof:implication} and  \cref{lem:pinchingoverlap}. The unit cells are given by the grid, each one contains $3$ qubits (grey dots). Stars or plaquettes are not drawn. The red and blue boxes are two overlapping regions $R_1$ and $R_2$. The thin dashed lines show the boundary of $R^{-}_1$. The orange box shows a set $r$ intersecting $R^{-}_1$ with diameter one. Any such set lies fully within $R_1$. Any star or plaquette with support intersecting $r$ can only have support in the green, dotted region. In particular, it cannot have support in $R_2\setminus R_1$. }
    \label{fig:partition_overlap}
  \end{figure}
\begin{proof}
  \begin{enumerate}
    \item Since $\strset_R$ is the set of stars $ s $ with $\cobdy s\cap R \neq \emptyset$, the representation of $ \strset_R $ as a union is immediate. The disjointness of this union follows from the fact that $\strset_{r|R}  \cap \strset_{w|R} = \emptyset$  if $r\neq w$. This holds, since $r=\cobdy s \cap R=w$ for any $s$ in both sets.
    \item Assume $\strset_{r|R}\cap \strset_{r'|R'} \neq \emptyset$. Then $r\cap R'=r'$, since for any $s\in \strset_{r|R}\cap \strset_{r'|R'}$:
      \begin{equation*}
        r\cap R' = (\cobdy s\cap R)\cap R' = \cobdy s \cap R' = r' \ .
      \end{equation*}
      Hence any $s\in \strset_{r|R}$ has support in $R'$ and thus is an element of $\strset_{R'}$.
      Using these two facts and the definition, we may rewrite:
      \begin{align*}
        \strset_{r|R} &= \left\{s\in \strset_{R}\middle| \cobdy s\cap R = r\right\}= \left\{s\in \strset_{R'}\middle| \cobdy s\cap R = r\right\}\\
        &=  \left\{s\in \strset_{R'}\middle| \cobdy s\cap R' = r' \ \mathrm{and} \  (\cobdy s\cap R)\setminus R' = r\setminus R'\right\} \subseteq \strset_{r'|R'} \ .
      \end{align*}
      \item By symmetry, we only need to prove the result in case $ i = 1$. We  pick any $v_0\in r\cap R^{-}_1$ and any $s\in \strset_{r|R_1\cup R_2}$. Recall that the support of any star has diameter at most $1$. Thus, the distance between $v_0$ and any $v\in r$ is at most one (cf.\ \cref{fig:partition_overlap}), since they all lie in the support $ds$ of $s$. By assumption, the distance of $v_0$ to $R_2\setminus R_1$ is greater than $2$. Hence the distance $\dist(r,R_2\setminus R_1)$ is greater than $1$ and $r\subseteq R_1$. The latter also implies that $\strset_{r|R_1\cup R_2}\subset \strset_{R_1}$ and, moreover, that $\strset_{r|R_1}$ is well defined. Then, for every $s\in \strset_{r|R_1\cup R_2}$
  \begin{equation*}
    \cobdy s\cap R_1 = (\cobdy s\cap(R_1\cup R_2))\cap R_1 = r\cap R_1 = r \ .
  \end{equation*}
  Thus $\strset_{r|R_1\cup R_2}\subseteq \strset_{r|R_1}$. The other inclusion follows from the fact that the diameter of a star is at most $1$. More precisely,  since the distance from $r$ to $R_2\setminus R_1$ is greater than $1$, no star with support in $r$ can have support in $R_2\setminus R_1$. In particular, for any $s\in \strset_{r|R_1}$:
  \begin{equation*}
    r = \cobdy s \cap R_1 = (\cobdy s \cap R_1 ) \cup (\cobdy s \cap (R_2\setminus R_1)) = \cobdy s \cap (R_1\cup R_2) , 
  \end{equation*}
  which concludes the proof of \cref{eq:pinchingoverlap:proof:implication} for $ i = 1 $. 
  \end{enumerate}
    This completes the proof.
\end{proof}
  
The partitions by support give rise to the operator sums 
\begin{equation}\label{eq:sumsPbS}
   \Sigma^\str_{r|R} \coloneqq \sum_{s\in \strset_{r|R} } A_s ,  
\end{equation} 
and similarly $ \Sigma^\plq_{r|R}  \coloneqq \sum_{p\in \plqset_{r|R} } B_p$. These sums will play a key role in our construction of explicit representations of the conditional expectations. 
They inherit the commutation properties of their constituents.

\begin{lemma}\label{lem:partitionproperties2}
 For any $v\in R \subseteq \lsym$, $Z_v$ either commutes or anti-commutes with $\Sigma_{r|R}^\str$: 
    $$ Z_v \Sigma_{r|R}^\str Z_v = \pm \Sigma_{r|R}^\str$$ with $-$ if and only if $v\in r$.
    Analogous results apply for the plaquette sums. 
\end{lemma}
\begin{proof}
For any individual star $s\in \strset_R$ the operator $ A_s $ either commutes or anti-commutes with any $Z_v$: $Z_v A_s Z_v = \pm A_s$, with $-$ if and only if $v$ is in the support of $s$. Since all stars $ s \in \strset_{r|R}$ share the same sites $ v $ in their support when restricted to $R$, any site $v$ is either in the support of all of them or not in the support of any.
\end{proof}

The conditional expectation is, among other things, a projection onto the kernel of the local Lindbladian, which according to \cref{pro:lindbladkern} is the commutant of all local Pauli $Z$-operators and all sums of the form $\sum_{s\in \bdy v}A_s$. 
Regrouping those terms according to the partition by support into the sums~\eqref{eq:sumsPbS} still allows to express the kernel as a commutant. 
\begin{lemma}\label{lem:kernelDbS}
Let $R\subseteq \lsym$ and let $J_R^\str$ index be the partition of the star set $\strset_R$ by support
 \begin{equation*}
    \ker \LL_R^\str 
    = \left\{Z_v, \Sigma_{r|R}^\str\right\}_{v\in R, \ r\in J_R^\str}' \ .
  \end{equation*}
 Analogous results apply to plaquettes.
\end{lemma}
\begin{proof}
Recall from Proposition~\ref{pro:lindbladkern} that $\ker \LL_R^\str = \left\{Z_v, \sum_{s\in \bdy v} A_s\right\}_{v\in R}' $.  We keep the $Z_v$ as they are and replace the sums $\sum_{s\in \bdy v}A_s$ with those of the partition by support $\Sigma_{r|R}^\str$:
\begin{equation}\label{eq:equalcomm}
 \left\{Z_v, \sum_{s\in \bdy v} A_s\right\}_{v\in R}' = \left\{Z_v, \sum_{s\in \strset_{r|R}} A_s\right\}_{v\in R, \ r\in J_R^\str}' \ .
  \end{equation}
The proof of this identity will be based on the fact that if two sets generate the same algebra, they also have the same commutant. 

It is easy to see that the algebra generated by $\sum_{s\in \bdy v}A_s$ with $ v \in R $ is included in the one generated by $\Sigma^\str_{r|R} $ with $ r\in J^\str_R $, since the former is the sum over all stars whose support intersects $v$ while the latter is the sum over all stars whose support in $R$ is exactly $r$:
  \begin{equation*}
    \sum_{s\in \bdy v}A_s = \sum_{r\in J^\str_R\colon \ v\in r} \Sigma_{r|R}^\str \ .
  \end{equation*}

  The other direction requires more work and involves the $Z_v$ operators. Our goal is, given $r\in J^\str_R$, to construct $\Sigma_{r|R}^\str$. First, we pick any $v_0\in r$ (recall that $r$ is never empty). We then observe that for any other $v\in r$, conjugating the sum $\sum_{s\in \bdy v_0}A_s$ with $Z_v$ adds a negative sign to all stars supported on $v$:
  \begin{equation*}
    \frac{1}{2}\left(\sum_{s\in \bdy v_0} A_s -  Z_{v}\sum_{s\in \bdy v_0} A_s Z_{v}\right) = \sum_{s\in \bdy v_0 \cap\bdy v} A_s
  \end{equation*}
  where $s\in \bdy v_0 \cap\bdy v$ if and only if $s$ is supported on both $v$ and $v_0$ (but not necessarily exactly).
  Repeating this step for all $v\in r$, we can construct the sum of all stars with $r\subseteq  \cobdy s$. However, there may still be stars in this sum that have support on some $v\in R\setminus r$. To exclude these, we again conjugate the sum with $Z_v$, this time for $v\in R\setminus r$ and with a plus sign:
  \begin{equation*}
    \frac{1}{2}\left(\sum_{s\in \bdy v_0, \ r\subseteq \cobdy s} A_s +  Z_{v}\sum_{s\in \bdy v_0, \ r\subseteq \cobdy s} A_s Z_{v}\right) = \sum_{s\in \bdy v_0, \ r\subseteq \cobdy s,\ s\notin \bdy v} A_s \ .
  \end{equation*}
  Repeating this for all $v\in R\setminus r$ yields a sum over all stars $s$ with $r= \cobdy s$. 
  This concludes the proof of~\eqref{eq:equalcomm}.
\end{proof}

\subsection{Explicit representations}
We are now ready to spell out the explicit expressions for the conditional expectations. They involve three ingredients. One will be the $ Z$- and $ X$-pinchings of observables $ O \in \BB(\HH_{\lsym})$:
\begin{align}\label{def:ZXpinch}
    \TT_R^\str(O) & \coloneq 2^{|R|}\sum_{\mathbf{z}\in 2^R} \ketbra{\mathbf{z}}{\mathbf{z}}O\ketbra{\mathbf{z}}{\mathbf{z}} , \notag \\
    \TT_R^\plq(O) & \coloneq 2^{|R|}\sum_{\mathbf{x}\in 2^R} \ketbra{\mathbf{x}}{\mathbf{x}}O\ketbra{\mathbf{x}}{\mathbf{x}} . 
  \end{align}
Here and in the following, we denote by $ \ketbra{\mathbf{z}}{\mathbf{z}} \equiv \ketbra{\mathbf{z}}{\mathbf{z}} \otimes \identity_{\lsym\backslash R}$ the orthogonal projection onto the subspace spanned by the joint eigenvectors $\ket{\mathbf{z}} \in \HH_R $ of all $ Z_v $, $ v \in R $, corresponding to the eigenvalues $ z_v $. The latter are the components of the vector $ \mathbf{z} $.   Similarly and  by a slight abuse of notation, $\ket{\mathbf{x}}$ marks the joint eigenvectors of $ X_v $, $ v \in R $, and eigenvalue $ x_v $. 

The second ingredient are the star or plaquette,  $\wildcard \in \{\str, \plq\}$, pinchings
\begin{equation}\label{def:Sigmapinch}
    \PP^\wildcard_R(O) \coloneq \sum_{\boldsymbol{\omega} \in \Omega_R^\wildcard} \mathbf{\Pi}_R^\str(\boldsymbol{\omega}) \, O \,  \mathbf{\Pi}_R^\str(\boldsymbol{\omega}) 
  \end{equation}
   corresponding to  the spectral projections 
    $$
     \mathbf{\Pi}_R^\wildcard(\boldsymbol{\omega}) \coloneqq \prod_{r\in J_{R}^\wildcard}\Pi_{r|R}^\wildcard(\omega_r) 
    $$
    associated with the collection of commuting operators $ \Sigma^\wildcard_{r|R} $ with $ r\in J_{R}^\wildcard $ and eigenvalues $ \boldsymbol{\omega} = (\omega_r)_{ r \in J^\wildcard_R} \in \Omega_R^\wildcard \coloneqq \bigtimes_{r\in J^\wildcard_R} \Omega^\wildcard_{r|R} $. The latter 
stands for the Cartesian product of the spectra $\Omega^\wildcard_{r|R}$ of $\Sigma^\wildcard_{r|R}$ and $\Pi^\wildcard_{r|R}(\omega)$ is the projection onto the eigenvalue $\omega \in \Omega^\wildcard_{r|R}$. 

The last piece in our construction of the conditional expectations is the insertion of the reduced Gibbs state $ \hat \rho_R^\wildcard $ (defined in~\eqref{def:restorGNS}) 
  to ensures the self-ajointness with respect to the scalar product~\eqref{def:sscalar}. The following is our key result. Note that this representation does not rely on $ D$-dimensionality of the underlying graph, and hence extends straightforwardly to general graphs (e.g.\ expanders, which are used in low-density parity-check codes).
\begin{theorem}\label{lem:condexp}
  Let $\wildcard \in \{\str, \plq\}$ and $R\subseteq \lsym$. The Davies conditional expectation $\E^{\wildcard}_R$ can be written as:
  \begin{equation}\label{eq:condexpexpression}
    \E^{\wildcard}_R(O)  = \TT^{\wildcard}_R \circ \PP^{\wildcard}_R (\hat\rho^{\wildcard}_R O) . 
  \end{equation}
  for any $ O \in \BB(\HH_{\lsym})$.
\end{theorem}
\begin{proof}
  We will use the fact that orthogonal projections are uniquely determined by the subspaces onto which they project.
  From  \cref{pro:condexpprop} and \cref{lem:kernelDbS}, we recall that $ \E^\str_R$ is the self-adjoint projection onto 
  \begin{equation*}
    \ker \LL^\str_R  = \left\{Z_v, \Sigma^\str_{r|R}\right\}_{v\in R, \ r\in J_R^\str}' \ . 
  \end{equation*}
  We will show that the expression in~\eqref{eq:condexpexpression} is also a self-adjoint projection onto the same image.
    For a proof, we first note that by \cref{lem:pinchingcommute} below, the two pinchings $\TT^\str_R$ and $\PP^\str_R$ commute. Thus, any operator $O\in \Img(\TT^\str_R \circ \PP^\str_R(\hat\rho_R \ \cdot \  ))$ commutes with any $Z_v$ for any $v\in R$, as well as with any $\Sigma^\str_{r|R}$ for any $ r\in J_R^\str $. Hence $O\in \ker \LL_R^\str$, which establishes the inclusion $ \Img(\TT^\str_R \circ \PP^\str_R(\hat\rho_R^\str \ \cdot \  )) \subseteq  \ker \LL_R^\str $. For the reverse inclusion, let $O\in \ker \LL_R^\str$. Then $O$ commutes with any projection $\ketbra{\mathbf{z}}{\mathbf{z}}$  onto any $Z$-basis state with $\mathbf{z}\in 2^R$. It also commutes with the spectral projections $\Pi_{r|R}^\str(\omega)$ for any $r\in J_R^\str$ and any $\omega \in \Omega_{r|R}^\str$
    justifying the first equality in
  \begin{align*}
    \TT^\str_R \circ \PP^\str_R(\hat\rho^\str_R O ) = \TT^\str_R( \PP^\str_R(\hat\rho^\str_R) O ) = \TT^\str_R \circ \PP^\str_R(\hat\rho^\str_R ) O = O . 
  \end{align*}
  Here, the last equality, which expresses the unitality of $ \TT^\str_R \circ \PP^\str_R(\hat\rho_R^\str \ \cdot \  ) $, follows from \cref{lem:pinchtotrace} below together with the definition of $\hat\rho^\str_R$, which imply
  \begin{equation*}
    \TT^\str_R(\hat\rho^\str_R) = \TT^\str_R\left(\frac{e^{-\beta H^\str_R}}{\Tr_R e^{-\beta H^\str_R}}\right) = \frac{1}{\Tr_R e^{-\beta H^\str_R}} \Tr_{R}(e^{-\beta H^\str_R}) = \identity \ .
  \end{equation*}

  Next, we show that $\TT^\str_R \circ \PP^\str_R(\hat\rho^\str_R \ \cdot\  )$ is a projection. Let $O\in \BB(\HH_{\lsym})$, then, by a repeated application \cref{lem:pinchingcommute} below and using $\TT_R(\hat\rho_R^\str)=\identity$, we find
  \begin{align*}
    \TT^\str_R \circ \PP^\str_R(\hat\rho^\str_R \cdot \TT^\str_R \circ \PP^\str_R(\hat\rho^\str_R  O )) &= \PP^\str_R \circ \TT^\str_R(\hat\rho^\str_R \cdot \TT^\str_R \circ \PP^\str_R(\hat\rho^\str_R O ))\\
    &= \PP^\str_R( \TT^\str_R(\hat\rho^\str_R)  \cdot \TT^\str_R \circ\PP^\str_R(\hat\rho^\str_R O ))\\
    &= \PP^\str_R(\identity  \cdot \TT^\str_R \circ\PP^\str_R(\hat\rho^\str_R O ))\\
    &= \TT^\str_R \circ\PP^\str_R(\hat\rho^\str_R O ) .
  \end{align*}
  In the second line we also used that $[\TT^\str_R(O),\ketbra{\mathbf{z}}{\mathbf{z}}] = 0 $ for any $O \in \BB(\HH_{\lsym}) $ and any $\mathbf{z}\in 2^R$.

  To show self-adjointness with respect to the scalar product~\eqref{def:sscalar} with $ s \in [0,1] $ arbitrary, we note that for any  $O\in \BB(\HH_{\lsym})$
  \begin{equation}\label{eq:pinchingsonrho}
      \PP^\str_R(\hat\rho^\str_R O )= \PP^\str_R((\hat\rho^\str_R)^{1-s} O (\hat\rho^\str_R)^s)
  \end{equation}
  since by the explicit form of the states $ \hat\rho^\str_R$ and the fact that $ H^\str_R = - \sum_{r \in J_R^\str}  \Sigma^\str_{r|R} $ one may let the star-pinching act on the conjugating operators. 
  We now use the definition~\eqref{def:sscalar} of the scalar product to conclude that for all $O_1, O_2 \in \BB(\HH_{\lsym})$:
  \begin{align*}
   & \mkern-50mu \langle O_1, \TT^\str_R \circ \PP^\str_R(\hat\rho^\str_R O_2 )\rangle_{\rho_\lsym^\strplq,s}
  = \Tr \left((\rho_\lsym^\strplq)^{s}  O_1^\dag (\rho_\lsym^\strplq)^{1-s} \TT^\str_R \circ \PP^\str_R((\hat\rho^\str_R)^{1-s} O_2 (\hat\rho^\str_R)^s) \right)\\
    &= \Tr \left( \frac{e^{- s\beta H^\str_R}}{(\Tr_{R} e^{-\beta H^\str_R})^{s}} \TT^\str_{R}\circ \PP^\str_{R}\left( \frac{e^{-s \beta H_\lsym^\strplq}}{(Z_\lsym^\strplq)^s} O_1^\dag \frac{e^{- (1-s)\beta H_\lsym^\strplq}}{(Z_\lsym^\strplq)^{1-s}}\right)  \frac{e^{- (1-s)\beta H^\str_R}}{(\Tr_{R} e^{-\beta H^\str_R})^{1-s}}   O_2   \right)\\
    &= \Tr \left( \frac{e^{-s \beta H_\lsym^\strplq}}{(Z_\lsym^\strplq)^s} \TT^\str_{R}\circ \PP^\str_{R}\left( \frac{e^{- s\beta H^\str_R}}{(\Tr_{R} e^{-\beta H^\str_R})^{s}}  O_1^\dag \frac{e^{- (1-s)\beta H^\str_R}}{(\Tr_{R} e^{-\beta H^\str_R})^{1-s}} \right) \frac{e^{- (1-s)\beta H_\lsym^\strplq}}{(Z_\lsym^\strplq)^{1-s}}   O_2   \right)\\
    &=  \langle \TT^\str_R \circ \PP^\str_R(\hat\rho^\str_R  O_1 ), O_2 \rangle_{\rho_\lsym^\strplq,s} . 
  \end{align*}
  In the second equality, we commuted the pinchings using \cref{lem:pinchingcommute}, and let them act as their Hilbert-Schmidt duals. The next line relies on the fact that $(\Tr_{R} e^{-\beta H^\str_R})^{-1}$ and $e^{-\beta (H^\str_{\lsym}-H^\str_R)}$ have no support on $R$, and thus commute with any projection $\ketbra{\mathbf{z}}{\mathbf{z}}$ supported on $R$. Thus, they can be pulled in and out of $\TT^\str_R$. Since they also commute with any star operator, they furthermore can be pulled in and out of $\PP^\str_R$.  
  Moreover, for any $O\in \BB(\HH_{\lsym})$
  \begin{equation*}
    \TT^\str_R \circ \PP^\str_R(e^{-s\beta H^\plq} O e^{-(1-s) \beta H^\plq} )= e^{-s\beta H^\plq} \TT^\str_R \circ \PP^\str_R(O ) e^{-\beta (1-s) H^\plq} ,
  \end{equation*}
  since every plaquette operator $ B_p $ commutes with every star operator $ A_s $ and every $Z_v$.
  The final step in the above chain of equalities is again based on~\eqref{eq:pinchingsonrho} and the fact that both pinchings are $\dag$-homomorphisms. 
  This concludes the proof.
\end{proof}

\subsection{Auxiliary results}
This subsection contains more technical lemmas, on which the proof of \cref{lem:condexp} relied. We start by showing that the pinchings introduced in~\eqref{def:ZXpinch} and~\eqref{def:Sigmapinch} all commute. 
\begin{lemma}\label{lem:pinchingcommute}
  Let $R_1\subseteq R_2\subseteq \lsym$ and  $R_3\subseteq R_4\subseteq \lsym$. Then $\TT_{R_1}^\str$, $\PP_{R_2}^\str$, $\TT_{R_3}^\plq$ and $\PP_{R_4}^\plq$ all commute.
\end{lemma}
\begin{proof}
  Since star operators are diagonal in $X$, plaquettes are diagonal in $Z$ and stars and plaquettes commute, the only non-trivial pairings are $[\TT_{R_1}^\str, \TT_{R_3}^\plq]$, $[\TT_{R_1}^\str, \PP_{R_2}^\str]$ and $[\TT_{R_3}^\plq, \PP_{R_4}^\plq]$ and the latter two are equivalent. The fact that $X$- and $Z$-pinchings commute is easy to check using the fact that $|\braket{\mathbf{x}|\mathbf{z}}|^2 $ is independent of $\mathbf{x} \in 2^R$ and $\mathbf{z}\in 2^R$. It thus remains to prove the last case.
  
  From \cref{lem:partitionproperties2} we recall that 
   $Z_v \Sigma_{r|R_{2}}^\str Z_v=\pm \Sigma_{r|R_{2}}^\str$ for any $r\in J^\str_{R_2} $ and $v\in R_1$.  Conjugating the corresponding spectral projections may hence only flip the sign of the eigenvalue  $\omega \in \Omega_{r|R}^\str$:
   $$Z_v \Pi_{r|R}^\str(\omega) Z_v=\Pi_{r|R}^\str(\pm\omega). $$
  For the rest of the proof, we abbreviate $\mathbf{\Pi}_R(\pm_v\boldsymbol{\omega}) \coloneqq Z_v\mathbf{\Pi}_R^\str(\boldsymbol{\omega})Z_v  $. 
  
  We now compute the action of $\TT_{R_1}^\str \circ \PP_{R_2}^\str $ in the $Z$-basis, that is, for any $\mathbf{z}_1, \mathbf{z}_2, \mathbf{z}_3, \mathbf{z}_4 \in 2^{\lsym}$: 
  \begin{align*}
    &2^{-|R_1|} \bra{\mathbf{z}_1{}}\TT_{R_1}^\str \circ \PP_{R_2}^\str (\ketbra{\mathbf{z}_2{}}{\mathbf{z}_3{}})\ket{\mathbf{z}_4{}} \\
    &= 1[\mathbf{z}_1|_{R_1} = \mathbf{z}_4|_{R_1}] \ \bra{\mathbf{z}_1{}} \PP_{R_2}^\str (\ketbra{\mathbf{z}_2{}}{\mathbf{z}_3{}})\ket{\mathbf{z}_4{}}\\
    &= 1[\mathbf{z}_1|_{R_1} = \mathbf{z}_4|_{R_1}]  (\mathbf{z}_1)_v (\mathbf{z}_2)_v (\mathbf{z}_3)_v (\mathbf{z}_4)_v  \bra{\mathbf{z}_1{}} Z_v \PP_{R_2}^\str (Z_v \ketbra{\mathbf{z}_2{}}{\mathbf{z}_3{}} Z_v)Z_v\ket{\mathbf{z}_4{}}\\
    &= 1[\mathbf{z}_1|_{R_1} = \mathbf{z}_4|_{R_1}] (\mathbf{z}_2)_v (\mathbf{z}_3)_v \sum_{\boldsymbol{\omega}\in \Omega^\str_{R_2}}
    \bra{\mathbf{z}_1{}}
    \mathbf{\Pi}_{R_2}(\pm_v\boldsymbol{\omega})
    \ketbra{\mathbf{z}_2{}}{\mathbf{z}_3{}}
    \mathbf{\Pi}_{R_2}(\pm_v\boldsymbol{\omega})   \ket{\mathbf{z}_4{}}\\
    &= 1[\mathbf{z}_1|_{R_1} = \mathbf{z}_4|_{R_1}]  (\mathbf{z}_2)_v (\mathbf{z}_3)_v \bra{\mathbf{z}_1{}} \PP_{R_2}^\str (\ketbra{\mathbf{z}_2{}}{\mathbf{z}_3{}})\ket{\mathbf{z}_4{}} .
  \end{align*}
  Here $1[\dots] $ denotes the indicator function of the set in brackets. In the first step, the indicator function resulted from the $Z$-pinching. 
  For the second step, we used that for $v\in R_i$ the values $(\mathbf{z}_i)_v \in \{\pm 1\}$ are the eigenvalues of $Z_v$ on the vector $\ket{\mathbf{z}_i}$. In the last step, 
   we absorbed the sign $\pm_v$ into the sum since the spectra $\Omega_{r|R}^\str$ are symmetric.
  To proceed, we compare the second with the last line in the equation above, finding
  \begin{align*}
      1[\mathbf{z}_1|_{R_1} = \mathbf{z}_4|_{R_1}]\bra{\mathbf{z}_1{}} \PP_{R_2}^\str (\ketbra{\mathbf{z}_2{}}{\mathbf{z}_3{}})\ket{\mathbf{z}_4{}} =  (\mathbf{z}_2)_v (\mathbf{z}_3)_v  1[\mathbf{z}_1|_{R_1} = \mathbf{z}_4|_{R_1}]\bra{\mathbf{z}_1{}} \PP_{R_2}^\str (\ketbra{\mathbf{z}_2{}}{\mathbf{z}_3{}})\ket{\mathbf{z}_4{}}
  \end{align*}
  and hence $ 1[\mathbf{z}_1|_{R_1} = \mathbf{z}_4|_{R_1}] \bra{\mathbf{z}_1{}} \PP_{R_2}^\str (\ketbra{\mathbf{z}_2{}}{\mathbf{z}_3{}})\ket{\mathbf{z}_4{}} = 0 $ unless $ (\mathbf{z}_2)_v (\mathbf{z}_3)_v = 1 $. 
  Since $v \in R_1$ was arbitrary, we conclude 
  \begin{multline*}
    2^{-|R_1|} \bra{\mathbf{z}_1{}}\TT_{R_1}^\str \circ \PP_{R_2}^\str (\ketbra{\mathbf{z}_2{}}{\mathbf{z}_3{}})\ket{\mathbf{z}_4{}} \\
    = 1[\mathbf{z}_1|_{R_1} = \mathbf{z}_4|_{R_1}]   1[\mathbf{z}_2|_{R_1} = \mathbf{z}_3|_{R_1}] \bra{\mathbf{z}_1{}} \PP_{R_2}^\str (\ketbra{\mathbf{z}_2{}}{\mathbf{z}_3{}})\ket{\mathbf{z}_4{}} . 
  \end{multline*}
  Repeating the same argument for $\PP_{R_2}^\str \circ\TT_{R_1}^\str$ gives an indicator function on $\mathbf{z}_2, \mathbf{z}_3$ from the $Z$-pinching and another indicator on $\mathbf{z}_1,\mathbf{z}_4$ from inserting $Z_v$. Thus, the two expressions are identical, as claimed.
\end{proof}

In what follows, it will be convenient to abbreviate the abelian algebras generated by $\{\Sigma^\str_{r|R}\}_{r\in J^\str_R}$ or $\{\Sigma^\plq_{r|R}\}_{r\in J^\plq_R}$ by $\AAA_R^\str$ and $\AAA_R^\plq$, respectively.
\begin{lemma}\label{lem:pinchtotrace}
  Let $ R_1 \subseteq R_2 \subseteq \lsym$, let $\wildcard \in \{\str, \plq\}$ and $O\in \AAA^\wildcard_{R_2}$. Then $\TT^\wildcard_{R_1}(O) = \Tr_{R_1}(O) \in \AAA^\wildcard_{R_2}$. In particular, $\hat\rho^\wildcard_{R}$ as defined in~\eqref{def:restorGNS} is an element of $\AAA^\wildcard_R$ for any $ R \subseteq \lsym$.
\end{lemma}
\begin{proof}
  Let $\Xi_{\boldsymbol{r}}^\str \coloneqq\Sigma_{r_1|R_2}^\str \cdot \Sigma_{r_2|R_2}^\str \cdot \ldots$ be a product of $\Sigma_{r_i|R_2}^{\str}$ with $r_i\in J^\str_{R_2}$. Then the operator $\bra{\mathbf{z}} \Xi_{\boldsymbol{r}}^\str \ket{\mathbf{z}} \in \BB(\HH_{R_1^c}) $ is independent of $\mathbf{z} \in 2^{R_1} $, which follows from the fact that any $\mathbf{z}' \in 2^{R_1}$ may be reached by spin flips, $\ket{\mathbf{z}'} = \bigotimes_{v \in \mathbf{z}\Delta \mathbf{z}'} X_v \ket{\mathbf{z}}$. Since $[\Xi_{\boldsymbol{r}}^\str,X_v]=0$ for any $v$, this implies
  \begin{equation*}
    \bra{\mathbf{z}} \Xi_{\boldsymbol{r}}^\str \ket{\mathbf{z}} = \bra{\mathbf{z}} \left(\bigotimes_{v \in \mathbf{z}\Delta \mathbf{z}'} X_v\right)^2 \Xi_{\boldsymbol{r}}^\str \ket{\mathbf{z}} = \bra{\mathbf{z}'} \Xi_{\boldsymbol{r}}^\str \ket{\mathbf{z}'}\ .
  \end{equation*}
  Now, any operator $O\in \AAA_{R_2}$ can be written as a finite linear combination $O = \sum_{\boldsymbol{r}} o_{\boldsymbol{r}}\Xi^\str_{\boldsymbol{r}}$ with some $o_{\boldsymbol{r}}\in \C$, which implies:
  \begin{align*}
    \TT^\str_{R_1}(O) &= 2^{|R_1|}\sum_{\mathbf{z}\in 2^{R_1}}\ketbra{\mathbf{z}}{\mathbf{z}}
    \sum_{\boldsymbol{r}} o_{\boldsymbol{r}} \Xi_{\boldsymbol{r}}^\str
    \ketbra{\mathbf{z}}{\mathbf{z}} \\
    &= \sum_{\mathbf{z}\in 2^{R_1}}\sum_{\boldsymbol{r}} o_{\boldsymbol{r}}\sum_{\mathbf{z}'\in 2^{R_1}} \ketbra{\mathbf{z}}{\mathbf{z'}}
    \Xi_{\boldsymbol{r}}^\str
    \ketbra{\mathbf{z}'}{\mathbf{z}} \\
    &= \identity_{R_1}\otimes\sum_{\boldsymbol{r}} o_{\boldsymbol{r}} \Tr_{R_1}\Xi_{\boldsymbol{r}}^\str
    = \Tr_{R_1} O \ .
  \end{align*}
  To show that $\Tr_{R_1}(O)\in\AAA_{R_2}^\str$, we recall from \cref{lem:partitionproperties2} that $Z_v \Sigma^\str_{r|R_2} = \pm\Sigma^\str_{r|R_2} Z_v$ for any $v\in R_1$ and any $r\in J^\str_{R_2}$, and thus also $Z_v\Xi^\str_{\boldsymbol{r}} = \pm \Xi^\str_{\boldsymbol{r}} Z_v$. We now consider the $Z$-pinching of $\Xi^\str_{\boldsymbol{r}}$ on a single qubit $v\in R_1$:
  \begin{align*}
    \TT^\str_{v}(\Xi^\str_{\boldsymbol{r}}) &= 2\left(\frac{\identity + Z_v}{2} \Xi^\str_{\boldsymbol{r}}  \frac{\identity + Z_v}{2} +  \frac{\identity - Z_v}{2} \Xi^\str_{\boldsymbol{r}}\frac{\identity - Z_v}{2}\right)\\
    &= 2\, \Xi^\str_{\boldsymbol{r}} \frac{\identity \pm Z_v}{2}\frac{\identity + Z_v}{2} +  2\, \Xi^\str_{\boldsymbol{r}} \frac{\identity  \mp Z_v}{2}\frac{\identity - Z_v}{2}\\
    &=
    \begin{cases}2\, \Xi^\str_{\boldsymbol{r}} &+\\ 0 & -
    \end{cases} \ .
  \end{align*}
  Thus, $\Tr_{R_1} O$ is indeed an element of $\AAA_{R_2}^\str$.

  For a proof of the last claim, note that $ H^\str_R = - \sum_{r \in J^\str_R}\Sigma^\str_{r|R}$ and thus $\exp\left(-\beta H^\str_R \right) \in \AAA^\str_R$. By the above argument, $\Tr_R(e^{-\beta H^\str_R})\in \AAA^\str_{R}$, and thus $\hat\rho^\str_R\in \AAA^\str_{R}$ as well.
\end{proof}

Our final lemma concerns the star or plaquette pinchings of two overlapping sets. It states that the product of these pinchings is equal to the pinching of the union of the two sets if the overlap is sufficiently wide -- made precise in terms of the definition~\eqref{def:Rminus} of interior sets. Simply put, this holds since any qubit in the union of the two sets is in the bulk of at least one of them, and the projections agree in the bulk. An illustration of the geometry can be found in \cref{fig:partition_overlap}. 
\begin{lemma}\label{lem:pinchingoverlap}
  Let $R_1, R_2\subseteq \lsym$ be two sets of qubits. If $R_1\cup R_2 = R_1^{-} \cup R_2^{-}$, then
  \begin{equation}
    \PP^\wildcard_{R_1}\circ \PP^\wildcard_{R_2} = \PP^\wildcard_{R_1\cup R_2}
  \end{equation}
  for both $\wildcard \in \{\str,\plq\}$.
\end{lemma}
\begin{proof}
  The proof is based on \cref{lem:partitionproperties}. 
  Since any $r\in J^\str_{R_1\cup R_2}$ intersects at least one of $R_1^{-}$ or $R_2^{-}$ by assumption, the partition of the stars of $R_1\cup R_2$ is a subset of the partitions of $R_1$ and $R_2$:
  \begin{equation}\label{eq:pinchingoverlapproof1}
    \{\strset_{r|R_1\cup R_2}\}_{r\in J_{R_1 \cup R_2}^\str} \subseteq \{\strset_{r|R_1}\}_{r\in J_{R_1}^\str} \cup  \{\strset_{r|R_2}\}_{r\in J_{ R_2}^\str} \ .
  \end{equation}
  Denoting the remaining terms in the union of the partitions by $\tilde r_i \in \tilde J^\str_{R_i}$ and the corresponding projections by $\PP^\str_{\tilde J^\str_{R_i}}$, we may rewrite:
  \begin{equation}\label{eq:pinchingoverlapproof2}
    \PP^\str_{R_1}\circ \PP^\str_{R_2} = \PP^\str_{R_1\cup R_2} \circ \PP^\str_{\tilde J^\str_{R_1}} \circ \PP^\str_{\tilde J^\str_{R_2}} \ .
  \end{equation}
  To prove that the right-hand side of \eqref{eq:pinchingoverlapproof2} equals $\PP^\str_{R_1\cup R_2}$, it suffices to show that for a single $\tilde r_1 \in \tilde J_{R_1}^\str$, its pinching
  \begin{equation*}
    \PP^\str_{\tilde r_1} (O)= \sum_{\omega\in \Omega_{\tilde r_1|R_1}^\str} \Pi_{\tilde r_1|R_1}^\str(\omega) \ O\  \Pi_{\tilde r_1|R_1}^\str(\omega)
  \end{equation*}
  satisfies $\PP^\str_{R_1\cup R_2}\circ \PP^\str_{\tilde r_1} = \PP^\str_{R_1\cup R_2}$. In turn, by \cref{lem:partitionproperties} we know that any $\strset_{\tilde r_1|R_1}$ is a disjoint union
  \begin{equation*}
    \strset_{\tilde r_1|R_1} = \bigcupplus_{k} \strset_{r_k|R_1\cup R_2}
  \end{equation*}
  for a finite set $\{r_k\}\subset J^\str_{R_1\cup R_2}$.
  Thus, the spectral projection $\Pi_{\tilde r_1|R_1}^\str(\omega)$ projects onto the sum of the eigenvalues of the $\{\Sigma^\str_{r_k|R_1\cup R_2}\}$. In other words, for any $\boldsymbol{\omega} \in \Omega_{R_1\cup R_2}^\str$ and $\omega\in\Omega_{\tilde r_i}^\str$
  \begin{equation*}
    \prod_{r\in J_{R_1\cup R_2}^\str} \Pi_{r|R_1\cup R_2}^\str(\omega_r) \Pi_{\tilde r_1|R_1}^\str(\omega)= \prod_{r\in J_{R_1\cup R_2}^\str} \Pi_{r|R_1\cup R_2}^\str(\omega_r)  1[\Sigma_k\omega_{r_k}=\omega] \ .
  \end{equation*}
  The sum over $\omega\in \Omega_{\tilde r_i}^\str$ removes the indicator function in the pinching. Thus
  \begin{equation*}
    \PP^\str_{R_1\cup R_2} \circ \PP_{\tilde r_1}^\str = \PP^\str_{R_1\cup R_2} \ 
  \end{equation*}
  as claimed. 
\end{proof}

\section{Approximate tensorization}\label{sec:approxtensor}
A key technical contribution in the proof of the main result will be the approximate tensorization of the relative entropy. 
It concerns the approximate additivity of the relative entropy of the star or plaquette Gibbs states  from~\eqref{def:restorGNS} associated with three disjoint subsets  $U, V ,W\subset \lsym$. They define two  subsets $ UV \coloneqq U \uplus V $, and $  VW  \coloneqq  V\uplus W  $, 
which overlap in $ V $, and their union $ UVW  \coloneqq  U \uplus V\uplus W $. An illustration of this geometry is found in \cref{fig:approxtensor} for the example of rectangular sets, which is the relevant case in the main result.

\begin{theorem}\label{thm:approxtensor}
Let $U, V ,W\subset \lsym$ be disjoint non-empty subsets with $ \dist(U,W)>4 $. Assume that the star or plaquette Gibbs states, $\wildcard \in \{\str, \plq\}$, restricted to the unions $ UVW $ and $UV$ satisfy the norm bound
  \begin{equation}
      \left\| \frac{\Tr_{VW} \hat{\rho}^\wildcard_{UVW}}{\Tr_{V} \hat{\rho}^\wildcard_{UV}} -1 \right\| \leq \varepsilon\leq \frac{1}{28} .
  \end{equation}
  Then for any full-rank state $ \sigma \in \SSS(\HH_{\lsym})$:
  \begin{equation}
    D(\sigma \|\E^{\wildcard,*}_{UVW}(\sigma))\leq \left(1+28 \varepsilon\right) \left(D(\sigma \|\E^{\wildcard,*}_{VW}(\sigma)) + D(\sigma \|\E^{\wildcard,*}_{UV}(\sigma)) \right).
  \end{equation}
\end{theorem}
The proof will be spelled at the end of this section. 
The two main ingredients of this proof are (i) a result from \cite{gao_Completepositivityorder_2025} reducing approximate tensorization to a positivity order of the conditional expectations, and (ii) the explicit expressions and properties of the conditional expectations derived in \cref{sec:condexp}. They will be the topic of the next subsection.

\subsection{Auxiliary results}

We start recalling a result from \cite{gao_Completepositivityorder_2025} on a rather general bound of the relative entropy involving essentially two conditional expectations, which are nested and related by positivity order. 
\begin{proposition}[{{\cite[Lemma 2.3]{gao_Completepositivityorder_2025}}}]\label{prop:lem2.3}
  Let $\EE$ be a conditional expectation and $\Psi$ be a unital and completely positive map defined over the set of bounded operators of a finite dimensional Hilbert space $\HH$, such that $\EE\circ \Psi = \EE$ and 
  \begin{equation}
    (1- \varepsilon) \EE \leq \Psi \leq (1+\varepsilon) \EE
  \end{equation}
   in the sense of complete positivity order. Then for any $\sigma \in \SSS(\HH)$
  \begin{equation}\label{eq:approxtensorfull}
    D(\sigma \| \EE^*(\sigma)) \leq \left(\frac{1-\varepsilon}{1+\varepsilon} - \frac{\varepsilon}{(1-\varepsilon)(2\ln2-1)}\right)^{-1} D(\sigma \|\Psi^*(\sigma)).
  \end{equation}
  In particular, if $\varepsilon \leq \frac{1}{14}$ we have
  \begin{equation}\label{eq:approxtensorbound}
    D(\sigma \| \EE^*(\sigma)) \leq (1+14\varepsilon) D(\sigma \|\Psi^*(\sigma)).
  \end{equation}
\end{proposition}
\begin{proof}
  The proof of \eqref{eq:approxtensorfull} is in~\cite{gao_Completepositivityorder_2025}. The bound~\eqref{eq:approxtensorbound} follows from elementary estimates: setting $k\coloneqq2\ln2-1$, we have $1/k\leq 3$, $\varepsilon^2\leq\varepsilon\leq \frac{1}{14}$ and
  \begin{equation*}
    \left(\frac{1-\varepsilon}{1+\varepsilon} - \frac{\varepsilon}{(1-\varepsilon)k}\right)^{-1} = \frac{1-\varepsilon^2}{1-\varepsilon(\frac{1}{k}+2)-\varepsilon^2(\frac{1}{k}-1)} \leq \frac{1}{1-7\varepsilon} \leq 1+14 \varepsilon . 
  \end{equation*}
\end{proof}

Next, we relate the positivity order of the conditional expectations back to a DS-type condition. The relation~\eqref{eq:posorderexp} is based on the explicit expressions (\cref{lem:condexp}). 

\begin{lemma}\label{lem:positivitybound}
  Let $U, V, W\subseteq \lsym$ be non-empty and disjoint with $\dist(U,W)>4$ and let $\wildcard \in \{\str, \plq\}$. Assume that the Gibbs states restricted to $ UVW $ and $ UV$ satisfy
  \begin{equation}
    \left\| \frac{\Tr_{VW} \hat{\rho}^{\wildcard}_{UVW}}{\Tr_{V} \hat{\rho}^{\wildcard}_{UV}} -1 \right\| \leq \varepsilon
  \end{equation}
  for some $0<\varepsilon\leq \frac{1}{2}$. Then
  \begin{equation}\label{eq:posorderexp}
    (1-2\varepsilon)\E_{UVW}^{\wildcard} \leq \E_{UV}^{\wildcard}\circ \E_{VW}^{\wildcard} \leq (1+2\varepsilon) \E_{UVW}^{\wildcard}
  \end{equation}
  where $\leq$ is in the sense of complete positivity.
\end{lemma}
\begin{proof}  The proof is spelled out for the star case, and, to make things more readable, we will drop the superscript $\str$ from $\E, \TT, \PP, \AAA,  H, \Sigma,J$ and $\hat \rho$ for the rest of this proof.
  We will use the explicit expressions from \cref{lem:condexp} and the following facts:
  \begin{enumerate}
    \item\label{item1} By \cref{lem:pinchtotrace} the partial traces $\Tr_{UVW}e^{-\beta H_{UVW}} $ and $\Tr_{VW}e^{-\beta H_{UVW}} $, which as partial traces of Gibbs states are also invertible, are elements of $\AAA_{UVW}$, the algebra generated by $\{\Sigma_{r|UVW}\}_{r\in J_{UWV}}$.  They are hence polynomials in the star sums $ \Sigma_{r|UVW} $ with $ r\in J_{UVW}$, and commute with each other and with any spectral projector of the star sums. In particular, they can be pulled in and out of the star pinching $\PP_{UVW}$. 
    In addition, both partial traces have no support on $VW$ and can thus be pulled in and out of any $Z$-pinching $\TT_{VW}$.
    \item\label{item2} By \cref{lem:pinchingcommute}: $[\TT_{VW}, \PP_{UVW}]=0$ and $[\TT_{V}, \PP_{UVW}]=0$.
    \item\label{item3} $\PP_{UVW} = \PP_{UV}\circ \PP_{VW} = \PP_{UVW} \circ \PP_{VW}$ by \cref{lem:pinchingoverlap}, as the distance between $U$ and $W$ is larger than $4$ and thus, any $v\in V$ is either in the interior of $UV$ or in the interior of $VW$.
  \end{enumerate}
  Applying $\E_{UVW}$ to $O \in \BB(\HH_{\lsym})$ and using Fact~\ref{item1} in the second and forth subsequent equality, and Facts~\ref{item2}-\ref{item3} for the fourth and fifth equality, we obtain:
  \begin{align}\label{eq:cond3}
    \E_{UVW}(O)
    &= \TT_{UVW}\left(\frac{\Tr_{VW}e^{-\beta H_{UVW}}}{\Tr_{VW}e^{-\beta H_{UVW}}} \PP_{UVW}\left(\frac{e^{-\beta H_{UVW}}}{\Tr_{UVW}e^{-\beta H_{UVW}}}  O \right)\right)  \notag\\
    &= \TT_{UVW}\left(\frac{\Tr_{VW}e^{-\beta H_{UVW}}}{\Tr_{UVW}e^{-\beta H_{UVW}}} \PP_{UVW}\left(\frac{e^{-\beta H_{UVW}}}{\Tr_{VW}e^{-\beta H_{UVW}}}  O \right)\right)  \notag\\
    &= \TT_{U}\circ \TT_{VW}\left(\Tr_{VW}(\hat\rho_{UVW}) \PP_{UVW}\left(\hat\rho_{VW}  O \right)\right)  \notag \\
    &= \TT_{U} \left(\Tr_{VW}(\hat\rho_{UVW}) \PP_{UVW}\circ \TT_{VW}\left(\hat\rho_{VW}  O \right)\right)  \notag\\
    &= \TT_{U} \left(\Tr_{VW}(\hat\rho_{UVW}) \PP_{UVW}\circ \PP_{VW}\circ \TT_{VW}\left(\hat\rho_{VW}  O \right)\right) \notag \\
    &= \TT_{U} \left(\Tr_{VW}(\hat\rho_{UVW}) \PP_{UVW}\circ \E_{VW}(O)\right) .
  \end{align}
  The third equality is also based on the observation that
  \begin{equation*}
    \hat\rho_{VW} = e^{-\beta H_{VW}}\left(\Tr_{VW}e^{-\beta H_{VW}}\right)^{-1} = e^{-\beta H_{UVW}}\left(\Tr_{VW}e^{-\beta H_{UVW}}\right)^{-1}
  \end{equation*}
  since $H_{UVW}-H_{VW}$ has no support on $VW$. 
  
  Applying $\E_{UV}\circ \E_{VW}$ to $O\in \BB(\HH_{\lsym}) $ and using the facts above, we similarly obtain: 
  \begin{align}\label{eq:conditer}
    \E_{UV}\circ \E_{VW}(O) &= \TT_{UV}\circ \PP_{UV}\left(\hat\rho_{UV} \E_{VW}(O) \right)\notag\\
    &=  \TT_{UV}\left(\hat\rho_{UV} \PP_{UV}\circ\E_{VW}(O) \right) \notag\\
    &= \TT_{U}\circ \TT_{V}\left(\hat\rho_{UV} \PP_{UVW}\circ(2^{-|V|}\TT_{V})\circ\E_{VW}(O)\right)\notag\\
    &= \TT_{U}\circ \TT_{V}\left(\hat\rho_{UV} (2^{-|V|}\TT_{V})\circ\PP_{UVW}\circ\E_{VW}(O)\right)\notag\\
    &= \TT_{U}\left(\TT_{V}(\hat\rho_{UV}) \PP_{UVW}\circ\E_{VW}(O)\right) \notag\\
    &= \TT_{U}\left(\Tr_V(\hat\rho_{UV}) \PP_{UVW}\circ\E_{VW}(O)\right)  . 
  \end{align}
  The factor $2^{|V|}$ on lines 3 and 4 normalizes $\TT_{V}$ into a $Z$-pinching. In particular, we use that $\TT_{V}\circ\TT_{V} = 2^{|V|}\TT_{V}$ to pull a $Z$-pinching out of $\E_{VW}(O)$ on line 3 and absorb it again on line 5.
  In order to compare the expressions in~\eqref{eq:cond3} and \eqref{eq:conditer} in the sense of positivity order for non-negative $ O \geq 0 $, we note that the following objects all commute:
  \begin{enumerate}
    \item $[\Tr_{VW}(\hat\rho_{UVW}),\Tr_V(\hat\rho_{UV})]=0$, since both are contained in $\AAA_{UVW}$ by \cref{lem:pinchtotrace}.
    \item $[\Tr_{VW}(\hat\rho_{UVW}), \PP_{UVW}\circ \E_{VW}(O)]=0$, since $\Tr_{VW}(\hat\rho_{UVW})$ is an element of $\AAA_{UVW}$ and $\PP_{UVW}$ projects onto the commutant $(\AAA_{UVW})'$.
    \item $[\Tr_V(\hat\rho_{UV}),\PP_{UVW}\circ \E_{VW}(O)]=0$ for the same reason.
  \end{enumerate}

  Thus, if
  \begin{equation*}
    (1-\varepsilon) \Tr_V(\hat\rho_{UV}) \leq \Tr_{VW}(\hat\rho_{UVW}) \leq (1+\varepsilon) \Tr_V(\hat\rho_{UV})
  \end{equation*} 
  and since $\PP_{UVW}\circ \E_{VW}(O)$ is non-negative by assumption on $ O\geq 0 $ and commutes with the partial traces inside the quantum channel $\TT_{U}$ featuring on the outside of \eqref{eq:cond3} and \eqref{eq:conditer}, we arrive at the bounds
  \begin{equation*}
    (1-\varepsilon) \E_{UV}\circ \E_{VW}(O) \leq \E_{UVW}(O) \leq (1+\varepsilon) \E_{UV}\circ \E_{VW}(O) .
  \end{equation*}
  To establish the claimed bound, we note that $\varepsilon\leq \frac{1}{2}$ implies the bounds $\frac{1}{1-\varepsilon}\leq 1+2\varepsilon$ and $\frac{1}{1+\varepsilon}\geq 1-2\varepsilon$. This finishes the proof.
\end{proof}

\subsection{Proof of approximate tensorization}
\begin{proof}[Proof of \cref{thm:approxtensor}]
  Combining \cref{prop:lem2.3} and \cref{lem:positivitybound} yields the bound:
  \begin{equation*}
    D(\sigma \|\E^{\str,*}_{UVW}(\sigma))\leq (1+28\varepsilon) D(\sigma \|\E^{\str,*}_{VW}\circ\E^{\str,*}_{UV}(\sigma)) \ .
  \end{equation*}
  We now use the fact that $\ln\E^{\str,*}_{VW}(\sigma) - \ln \E^{\str,*}_{VW}\circ\E^{\str,*}_{UV}(\sigma)$ is an element of the von Neumann subalgebra of $\E^{\str}_{VW}$ (cf.\ \cref{pro:logdiff}) to rewrite the relative entropy for full-rank states $\sigma$ (cf. \cref{lem:condexpkern}):
  \begin{align*}
    D(\sigma \|\E^{\str,*}_{VW}\circ\E^{\str,*}_{UV}(\sigma)) &= \Tr(\sigma (\ln\sigma -\ln \E^{\str,*}_{VW}\circ\E^{\str,*}_{UV}(\sigma)))\\
    &= \Tr(\sigma (\ln\sigma -\ln\E^{\str,*}_{VW}(\sigma)))\\
    &\phantom{=} + \Tr(\sigma(\ln\E^{\str,*}_{VW}(\sigma) - \ln \E^{\str,*}_{VW}\circ\E^{\str,*}_{UV}(\sigma)))\\
    &=D(\sigma \|\E^{\str,*}_{VW}(\sigma)) +  D(\E^{\str,*}_{VW}(\sigma) \|\E^{\str,*}_{VW}\circ\E^{\str,*}_{UV}(\sigma))\\
    &\leq D(\sigma \|\E^{\str,*}_{VW}(\sigma)) + D(\sigma \|\E^{\str,*}_{UV}(\sigma)) . 
  \end{align*}
  The equality from the first to the fourth line is sometimes refered to as the chain rule.
  The last inequality is the data processing inequality for the relative entropy.
\end{proof}
\section{Proof of the MLSI}\label{sec:mainresult}
In this section, we finally spell the proof of the main result and its immediate consequences. 
To lay out the structure of the argument, we start with the proof of \cref{thm:main}, postponing the heart of the argument, the multiscale analysis, to the next subsection.

\subsection{Proof of Theorem~\ref{thm:main} and Corollary~\ref{corMLSIfull} }
  The proof of the main theorem (\cref{thm:main}), which is spelled for the star case, consists of two steps. First, \cref{lem:mlsirecursion} below establishes a recursive bound of the MLSI constant on different length scales. Second, a lower bound for the MLSI constant on some fixed, finite length scale $L^\str_0$ is derived from a spectral gap of the Lindbladian $ \LL_\lsym $.
  
  In this context, we define the (star) MLSI constant of any rectangle $ R \subseteq \lsym$ as
  \begin{equation}\label{eq:CMLSI}
    \alpha^\str(R):= \inf_{\sigma \in \SSS(\HH_{\lsym}), \sigma>0} \frac{EP^\str_R(\sigma)}{2 D(\sigma \| \E^{\str,*}_{R}(\sigma))} \ .
  \end{equation}
  where the numerator is the entropy production 
  \begin{equation*}
   EP_{R}^\str (\sigma)\coloneqq -\Tr\left(\LL^{\str,*}_{R}(\sigma)\left(\ln \sigma - \ln \E^{\str,*}_{R}(\sigma)\right)\right) 
  \end{equation*}
  related to the star-Lindbladian on $ R $ and a full-rank state $ \sigma $. For the convenience of the reader, some basic properties of this quantity are collected in~\cref{sec:appendix:full-rank}. In particular, it is proven there (\cref{lem:condexpkern}), that:
  \begin{itemize}
      \item $ \E^{\str,*}_{R}(\sigma) $ is again of full rank (\cref{lem:condexpkern});
      \item on the right-hand side, we may exchange $ \ln \rho_R^\strplq $ for $ \ln \E^{\str,*}_{R}(\sigma) $ (\cref{pro:logdiff}).
  \end{itemize} 
  Since the projected state $ \E^{\str,*}_{R}(\sigma) $ is again a state on the full~region $\lsym $, taking a strict subset $ R\subset \lsym $ means that the constant in~\eqref{eq:CMLSI} is, in the notion of \cite{bardet_Estimatingdecoherencetime_2017,gao_FisherInformationLogarithmic_2020,gao_Completeentropicinequalities_2022}, bounded from below by the optimal constant in the complete modified logarithmic Sobolev inequality (CMLSI) for $ \LL_R^\str$. \\

Postponing the details of the multiscale method which underlies \cref{lem:mlsirecursion},  the proof of the main result proceeds as follows. We let $\RR_L$ be the set of rectangles of diameter at most $L$. The MLSI constant of a length scale $L$ is given by the minimum over all rectangles $$ \alpha^\str(L) \coloneqq \min_{R\in \RR_L} \alpha^\str(R) , $$ 
(where the same symbol $ \alpha^\str $ is used by a slight abuse of notation).
\begin{proof}[Proof of \cref{thm:main}]
  By \cref{lem:mlsirecursion} below and assuming \eqref{eq:corrdecaywildcard}, there exists a length scale $L^\str_0$, which is larger or equal to the scale in~\cref{def:dscond} (and, by a slight abuse of notation, denoted by the same symbol),  such that for $ J $ such that $R\in \RR_{2^J L_0}$:
  \begin{equation}
      \alpha^\str(R) \geq \alpha^\str(2^J L_0^{\str}) \geq \prod_{j=0}^{J-1} \left(1+ \frac{12 D}{(L_0^\str)^{1/3} 2^{j/3}} \right)^{-1} \alpha^\str(L^\str_0)
  \end{equation}
  where we iterated \eqref{eq:mlsirecursion} below exactly $J$ times.
  Next, we bound the product in the right side using $\ln(1+x)\leq x$ to arrive at:
  \begin{equation}\label{eq:MLSIlb}
    \alpha^\str(R) \geq \exp\left[-\frac{12 D}{(L_0^\str)^{1/3}} \sum_{j=0}^{\infty}2^{-j/3}\right] \alpha^\str(L_0^\str)\geq \exp\left[\frac{-60 D}{(L_0^\str)^{1/3}} \right]\alpha^\str(L^\str_0) \ .
  \end{equation}
  Note that this bound is independent of the system size and $J $.

It remains to establish a uniform lower bound on $\alpha^\str(L^\str_0) $ for the initial length scale. 
To do so, we use the lower bound from~\cite[Cor.~3.3 and Eq.~(19)]{gao_Completeentropicinequalities_2022} for any $ R \in \mathcal{R}_{L^\str_0} $:
\begin{equation}\label{eq:gaorouze}
    \alpha^\str(R) \geq \frac{\gap \LL_{R}^\str}{C_{cb}(\E^{\str}_{R})}  ,
\end{equation}
which involve a spectral gap and the completely bounded Pimsner-Popa index $ C_{cb}(\E^{\str}_{R}) $ associated with the conditional expectation. By~\cite[Eq.~(14) and (20)]{gao_Completeentropicinequalities_2022} the latter is bounded from above 
\begin{equation*}
C_{cb}(\E^{\str}_{R}) \leq (\dim \HH_{R'})^2 (\hat{\rho}^\str_{R})_{\textrm{min}}^{-1} ,    
\end{equation*}
$R'$ being the set of qubits at distance at most $1$ from $R$, and 
since \cref{pro:fullrankfixedpoint} ensures that the $ R $-local Gibbs state is invariant, 
$ \hat{\rho}^\str_{R} =\E^{\str,*}_{R}(\hat{\rho}_{R}^\str) $. 
Its minimal eigenvalue is lower bounded  according to 
$$
(\hat{\rho}^\str_{R})_{\textrm{min}}^{-1} \leq \left\| e^{\beta H_R^\str} \Tr_R e^{-\beta H_R^\str} \right\| \leq 2^{|R|} \exp\left( 2 \beta \| H_{R}^\str \|\right) .
$$ 
Since the Hamiltonian is a sum of bounded local terms, $\| H_{R}^\str \|$ is estimated by a polynomial in the size $ | R' |$, which is uniformly bounded in terms of $ L_0^\str$. 

Next, we recall from \cite{gao_Completeentropicinequalities_2022} the details concerning the spectral gap featuring in~\eqref{eq:gaorouze} and provide a lower bound, which is uniformly positive. 
Since we may exchange \cite[cf.~Lemma 3.2]{gao_Completeentropicinequalities_2022} the invariant state $ \E^{\str,*}_{R}(\sigma) $ by the full Gibbs state $\rho^\strplq_\lsym $, the notion of spectral gap from \cite[Eq.~(33)]{gao_Completeentropicinequalities_2022} agrees with the lowest non-zero eigenvalue associated with the Dirichlet form~\eqref{eq:Dirichlet} averaged over $ s \in [0,1] $:
\begin{equation*}
    \gap  \LL_{R}^\str \coloneqq \inf_{O \in \BB(\HH_R)}  \int_0^1 \frac{- \langle O, \LL^\str_R(O) \rangle_{\rho_\lsym^\strplq,s}}{\| O -\E^{\str}_{R}(\sigma)(O) \|_{\textrm{BKM}} } ds . 
\end{equation*}
The average corresponds the Bogoliubov-Kubo-Mori scalar product whose norm is abbreviated by $ \| \cdot \|_{\textrm{BKM}} $. 
By self-adjointness of $ \LL^\str_R $ (\cref{pro:fullrankfixedpoint}), the above is indeed the second largest eigenvalue of the linear operator $ \LL^\str_R $ on the finite-dimensional Hilbert space $ \BB(\HH_R)$.

By the assumption on the jump rates (\cref{def:jumprates}) and using the representation~\eqref{eq:Dirichlet}, the Dirichlet form $ - \langle O, \LL^\str_R(O) \rangle_{\rho^\strplq_\lsym,s} $
is lower bounded in terms of a Dirichlet form, in which the jump rates are set to $\widetilde h^\str(\omega) \coloneqq g^\str e^{\beta\omega/2} $. Since these jump rates still satisfy detailed balance, they correspond to a Davies Lindbladian $ \widetilde{\LL}_{R}^\str $, whose jump rates are independent of $v$, and for which, by the above reasoning:
\begin{equation}
    \gap  \LL_{R}^\str \geq \gap  \widetilde \LL_{R}^\str .
\end{equation}
The right side is bounded uniformly in all $ R \in \mathcal{R}_{L^\str_0} $, since by translation invariance of $ \widetilde\LL_{R}^\str$, one only needs to consider finitely many rectangles of the given length scale, and on which the second largest eigenvalue of $  \widetilde\LL_{R}^\str $ does not vanish. This completes the proof of $\alpha^\str(L^\str_0) > 0 $.
\end{proof}

\begin{figure}
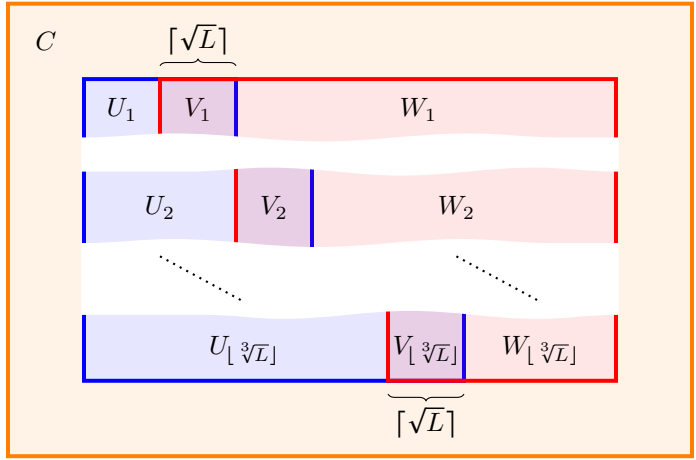

  \centering
  \include{tikz-martingale}
  \caption{Geometry for the multiscale analysis. A sequence of partitions $U_iV_iW_i$ such that $V_i$ and $V_{i+1}$ do not overlap.}
  \label{fig:martingale}
\end{figure}

Having laid out the argument for the MLSI of the star or plaquette part of the Lindbladian, the argument for the full Lindbladian is as follows.

\begin{proof}[Proof of \cref{corMLSIfull}]
 By assumption, both $\LL_\lsym^\str$ and $\LL^\plq_\lsym$ satisfy an MLSI. Since they commute by \cref{lem:lindcom}, we have $\E^{\strplq,*}_R = \E^{\str,*}_R\circ \E^{\plq,*}_R$ together with \cref{pro:logdiff} for any full rank state $\sigma$:
    \begin{align*}
        D(\sigma \| \E^{\strplq,*}_R(\sigma) )&\leq D(\sigma \| \E^{\str,*}_R(\sigma))+D(\sigma \| \E^{\plq,*}_R(\sigma))\\
        &\leq \frac{EP_R^\str(\sigma)}{2\alpha^\str(R)} + \frac{EP_R^\plq(\sigma)}{2\alpha^\plq(R)}\\
        &\leq \frac{1}{2\min\{\alpha^\str(R), \alpha^\plq(R)\}}EP_R^\strplq(\sigma) . 
    \end{align*}
    In the last step, we used \cref{pro:entropyprodprop}, which guaranteed that the entropy production of $\LL_R^\str$ and $\LL^\plq_R $ is additive.
\end{proof}

\subsection{Multiscale analysis}
This subsection is dedicated to the derivation of the recursive bound on the MLSI constant used in the proof of the main theorem. The multiscale argument employed here goes back to the classical works \cite{stroock_equivalencelogarithmicSobolev_1992, stroock_logarithmicSobolevinequality_1992, martinelli_ApproachEquilibriumGlauber_1994, martinelli_LecturesGlauberDynamics_1999,cesi_Quasifactorizationentropylogarithmic_2001}. A related argument for quantum systems has been discussed in \cite{kastoryano_QuantumGibbsSamplers_2016}. One main ingredient is the approximate tensorization (\cref{thm:approxtensor}), which yields good control over the relative entropy between a state and its infinite time limit on different scales. 

The proof of the following result follows approximately \cite{martinelli_LecturesGlauberDynamics_1999, kastoryano_QuantumGibbsSamplers_2016}.
\begin{proposition}\label{lem:mlsirecursion}
  Let $\{\LL^\strplq_{\lsym}\}_{\lsym \in \famsym}$ be as in \cref{thm:main}, pick $\wildcard \in \{\str,\plq\}$ and assume that it satisfies $\mathrm{DS}^\wildcard(K^\wildcard, \xi^\wildcard)$ with constants $K^\wildcard <\infty$, $ \xi^\wildcard>0$.  
  Then, there exists $L^\wildcard_0>0$ independent of the system size such that for all $L\geq L^\wildcard_0$, $\alpha^\wildcard(L)$ satisfies the following recursion relation:
  \begin{equation}\label{eq:mlsirecursion}
    \alpha^\wildcard(2L) \geq \left(1+ \frac{12 D}{L^{1/3}} \right)^{-1}\alpha^\wildcard(L)
  \end{equation}
  where $D$ is the dimension of the lattice.
\end{proposition}
\begin{proof}
  We fix any rectangle $R\in \RR_{2L}$. If no side of $R$ is longer than $L$, $R\in\RR_{L}$, and thus $\alpha^\str(R) \geq \alpha^\str(L)$ by definition. For the rest of the proof, we hence assume that one side of $R$, which we will call $l_1$,  is longer than $L$, that is $L<l_1\leq 2L$. To simplify notation, we abbreviate in this proof for any $R$:
  $$D_R \coloneqq D(\sigma \| \E^{\str,*}_{R}(\sigma)). $$

 \noindent  We partition $R = U\uplus V\uplus W$ into overlapping rectangles $UV$ and $VW$ along the direction of $l_1$. If the boundary is periodic in the direction of $l_1$, $R$ might have a winding number in this direction. In this case, $V$ contains two components $V=V_1\uplus V_2$, see \cref{fig:approxtensor}.  Assume furthermore that the length of the overlap $V$, that is, the distance $\dist(U,W)$, is at least $\sqrt{L}$ and larger than $4$.
  Then the diameter of $W$ is shorter than $2L\leq 2 \dist(U,W)^2 $ and the DS bound \eqref{eq:corrdecaywildcard} holds. 
  We now take $L_0^{\str}$ large enough such that the DS bound \eqref{eq:corrdecaywildcard} satisfies for all $L\geq L_0^{\str}$:
  \begin{equation*}
      K^\str e^{-\xi^\str \sqrt{L}}\leq \frac{1}{28} .
  \end{equation*}
  By approximate tensorization, \cref{thm:approxtensor} above, we may then bound the relative entropy with $K'=28K^\str$:
  \begin{align*}
    D_R &\leq (1 + K' e^{-\xi^{\str} \sqrt{L}}) (D_{UV} + D_{VW})\\
    &\leq (1 + K' e^{-\xi^{\str} \sqrt{L}}) \left(\frac{EP^\str_{UV}}{2 \alpha^\str(UV)} + \frac{EP^\str_{VW}}{2 \alpha^\str(UV)} \right)\\
    &\leq (1 + K' e^{-\xi^{\str} \sqrt{L}}) \frac{1}{2 \min\{\alpha^\str(UV), \alpha^\str(VW)\}}\left(EP^\str_{UV} + EP^\str_{VW} \right) .
  \end{align*}
  Here we dropped the argument $\sigma$ not only in the relative entropy, but also in the entropy production.

  Our next aim is to use the subadditivity of the entropy production (\cref{pro:entropyprodprop}). However, due to the overlap $V$, this is not possible for $EP^\str_{UV} + EP^\str_{VW}$. One way to get around this  is to average over many partitions as was done in \cite{martinelli_LecturesGlauberDynamics_1999}.
  
  Consider a family $(U_i, V_i, W_i)_{i=1}^I$ of partitions with $I=\lfloor\sqrt[3]{L}\rfloor$ and the property that for all $i=1,\ldots, I-1$ the regions $U_iV_i$ and $V_{i+1}W_{i+1}$ do not overlap (cf. \cref{fig:martingale}). By \cref{pro:entropyprodprop} the sum of the entropy productions of disjoint regions can be bounded by the entropy production of their union:
  \begin{equation*}
      EP^\str_{U_i V_i} + EP^\str_{V_{i+1} W_{i+1}} \leq EP^\str_{R} \ .
  \end{equation*}
  Averaging over all $i$ and pairwise combining entropy productions of $i$ and $i+1$ except $EP^\str_{V_1W_1}$ and $EP^\str_{U_IV_I}$, which we both bound individually by $EP^\str_{R}$,  we find:
  \begin{align}\label{eq:afterav}
    D_R \leq (1 + K' e^{-\xi^{\str} \sqrt{L}}) \frac{1}{2 \min\{\alpha^\str(U_iV_i), \alpha^\str(V_iW_i)\}_{i=1}^{I}} \left(1+L^{-1/3}\right) EP^\str_R \ .
  \end{align}

  Next, we let $R' \in\{U_iV_i, V_iW_i\}_{i=1}^I$ be the rectangle minimizing $\alpha^\str$.
  We will argue below that for $L$ large enough, the width of $R'$ in the coordinate direction of $l_1$ can be taken at most $3/2 L$.
  Then, we distinguish two cases: 
  \begin{enumerate}
    \item the size of $R'$ in all other directions is less than $\frac{3}{2}L$ and hence $R'\in \RR_{\frac{3}{2}L}$,
    \item $R'$ is longer than $\frac{3}{2}L$ in at least one direction.
  \end{enumerate}
  In the second case, we repeat the above procedure a maximum of $D-1$ times, since every coordinate direction in which we cut will then be shorter than $\frac{3}{2}L$.
  We then conclude from \eqref{eq:afterav}:
  \begin{equation}\label{eq:recursionthird}
    \alpha^\str(3/2L) \leq (1 + K' e^{-\xi^{\str} \sqrt{L}})^D \left(1+L^{-1/3}\right)^D \alpha^\str(2L).
  \end{equation}
  Setting $L_0^{\str}$ such that $K' e^{-\xi^{\str} \sqrt{L}} \leq L^{-1/3}\leq \frac{\ln 2}{6D}$ for all $L\geq 3/2 L_0^{\str}$ and iterating \eqref{eq:recursionthird} three times, we then arrive at
  \begin{align*}
    \alpha^\str(L) &\leq (1 + K' e^{-\xi^{\str} \sqrt{\frac{2}{3} L}})^{3D}\left(1+\left(\frac{2}{3} L\right)^{-1/3}\right)^{3D} \alpha^\str(64/27 L)\\
    &\leq \left(1+ L^{-1/3}\right)^{6D}\alpha^\str(2 L) \leq \left(1+12 D L^{-1/3}\right)\alpha^\str(2 L).
  \end{align*}

  It remains to be shown that such a family of partitions exists. There are two cases. If $R$ does not wrap around the torus along $l_1$, $V$ only has one connected component, and we choose the following partition:
  Let the length of $V_i$ be $\lceil \sqrt{L}\rceil$,
  let the length of $U_i$ be $\lfloor l_{0} \rfloor + i \lceil\sqrt{L}\rceil$ and let the length of $W_i$ be $\lceil l_0  \rceil + (I-i+1) \lceil\sqrt{L}\rceil$, where
  \begin{equation*}
    l_0 \coloneqq \frac{l_1- (I+2)\lceil\sqrt{L}\rceil}{2}\ .
  \end{equation*}
  This defines a valid partition if $\lceil \sqrt L \rceil >4$ and if the lengths of $U_1$ and $W_I$ are at least $1$. The latter condition is implied (using $l_1\geq L$) by \eqref{eq:afterav}: $ L - I \lceil\sqrt{L}\rceil\geq 4$. Both conditions are satisfied for all $L\geq L_0^{\str}$ if we pick $L_0^{\str}\geq16$.

  In case $R$ wraps around the torus along $l_1$, then $V$ contains two connected components, $V= V^1 \sqcup V^2$ (see \cref{fig:approxtensor}). The length of $V^1$ and $V^2$ are taken both as $\lceil\sqrt{L}\rceil$ which also agrees with $\dist(U,W)$. Similarly to the first case, we define a partition with a growing $U_i$ and shrinking $W_i$. However, in this case, $U_i$ grows on both ends and $W_i$ shrinks accordingly. That is, the length of $U_i$ is $\lfloor l_{0}' \rfloor + 2i \lceil\sqrt{L}\rceil$ and the length of $W_i$ is $\lceil l_0'  \rceil + 2(I-i+1) \lceil\sqrt{L}\rceil$ with
  \begin{equation*}
    l_0' \coloneqq \frac{l_1- 2(I+2)\lceil\sqrt{L}\rceil}{2}\ .
  \end{equation*}
  Requiring that the lengths of $U_i$ and $W_i$ are at least $1$ leads to the following condition: $ L - 2 I \lceil\sqrt{L}\rceil\geq 4$. This holds for $L_0^{\str}\geq 125$.

  Finally, the longest component in both cases is $W_1$. Its length is at most $L+\frac{\sqrt[3]{L}(\sqrt{L}+1)}{2}+1$ in the first case or $L+\sqrt[3]{L}(\sqrt{L}+1)+1$ in the second case. This length is shorter than $3/2 L$ if $L_0^{\str}\geq 1000$.
\end{proof}

\begin{remark}
    The proof works identically if we replace the entropy production by another subadditive, positive quantity. One particular choice is the approximate tensorization constant of a region $R$, defined as
    \begin{equation*}
        AT^\wildcard(R) \coloneqq \inf_{\sigma \in \SSS(\HH_{\lsym}), \sigma>0} \frac{ \sum_{v\in R}D(\sigma \| \E^{\wildcard,*}_{v}(\sigma))}{ D(\sigma \| \E^{\wildcard,*}_{R}(\sigma))} \ .
    \end{equation*}
    A recursive relation of this constant gives an alternative proof for the MLSI. Another choice is the entropy production of a Lindbladian $\hat{\LL}$ which has the same conditional expectation as the Davies-Lindbladian $\LL^\wildcard$. 
\end{remark}

\subsection{Proof of \cref{lem:nonselfcorr}}
We finally spell the proof of the erasure of quantum information in case the star Lindbladian is rapidly mixing. 
\begin{proof}[Proof of \cref{lem:nonselfcorr}]
    By an approximation argument, it suffices to establish the bound for any full-rank state  $\sigma \in \SSS(\HH_{\lsym})$.
    Since $\LL^\str_{\lsym}$ satisfies an MLSI:
    \begin{equation*}
        \|e^{t \LL^{\str, *}_{\lsym}}(\sigma) - \E^{\str,*}_{\lsym}(\sigma) \|_1 \leq e^{-\alpha^{\str} t} \sqrt{ 2 D\left(\sigma\| \E^{\str,*}_{\lsym}(\sigma)\right)} \ .
    \end{equation*}
    Using the fact that $\ln \E^{\str,*}_{\lsym}(\sigma) - \ln \rho^\strplq_{\lsym}$ is an element of the subalgebra associated with $ \E^{\str}_{\lsym}$ (\cref{pro:logdiff}), we rewrite the relative entropy as a difference and drop the non-positive term:
    \begin{align*}
        D\left(\sigma\| \E^{\str,*}_{\lsym}(\sigma)\right) &= D\left(\sigma\| \rho^\strplq_{\lsym}\right) - D\left(\E^{\str,*}_{\lsym}(\sigma)\| \rho^\strplq_{\lsym}\right) \leq D\left(\sigma\| \rho^\strplq_{\lsym} \right) \ . 
    \end{align*}
    By the standard bound 
    $$D\left( \sigma \| \rho^{\strplq}_{\lsym} \right) \leq \ln((\rho^{\strplq}_{\lsym})^{-1}_{\min}) = \mathrm{poly}(|\lsym|) , 
    $$
    we thus conclude that the star part mixes rapidly to the kernel of $\LL^{\str}_{\lsym}$. By construction, the logical operators $X_L$ and $Y_L=iX_LZ_L$ commute with all stars and plaquettes, and are annihilated by the star dynamics:
    \begin{align*}
        \E^\str_{\lsym}(X_L) &= \PP^\str_{\lsym} \circ \TT^\str_{\lsym}(\hat{\rho}^{\str}_{\lsym}X_L)
        =\PP^\str_{\lsym}\left(\Tr(\hat{\rho}^{\str}_{\lsym}X_L)\right) =0 ,\\
        \E^\str_{\lsym}(Y_L) &= i \PP^\str_{\lsym} \left( \TT^\str_{\lsym}(\hat{\rho}^{\str}_{\lsym}X_L)Z_L\right)
        =i \PP^\str_{\lsym} \left( \Tr(\hat{\rho}^{\str}_{\lsym}X_L)Z_L\right) =0 .
    \end{align*}
    This follows from the fact that $\TT^\str_{\lsym}$ is the $Z$-pinching on the full lattice and $\hat{\rho}^{\str}_{\lsym}X_L$ is diagonal in the $X$-basis. The above trace is zero by \cref{cor:logicalexpectation}.

    The above facts then allow us to conclude that for both $O_L\in \{X_L, Y_L\} $:
    \begin{align*}
        \left|\Tr( e^{t\LL^{\strplq, *}_{\lsym}}(\sigma)O_L )\right| &=  \left|\Tr( \left(e^{t\LL^{\strplq,*}_{\lsym}}(\sigma)- \E^{\str,*}_{\lsym}\circ e^{t\LL^{\plq,*}_{\lsym}}(\sigma)\right)O_L)\right|\\
        &= \left|\Tr( e^{t\LL^{\plq,*}_{\lsym}}\left(e^{t\LL^{\str,*}_{\lsym}}(\sigma)- \E^{\str,*}_{\lsym}(\sigma)\right)O_L)\right|\\
        &\leq \left\| e^{t\LL^{\plq,*}_{\lsym}}\left(e^{t\LL^{\str,*}_{\lsym}}(\sigma)- \E^{\str,*}_{\lsym}(\sigma)\right) \right\|_1\\
        &\leq \left\| e^{t\LL^{\str,*}_{\lsym}}(\sigma)- \E^{\str,*}_{\lsym}(\sigma) \right\|_1 \leq \mathrm{poly}(|\lsym|)e^{-\alpha^{\str} t}.
    \end{align*}
    Here, the second line follows from the fact that the star and plaquette parts of the Lindbladian commute (\cref{lem:lindcom}). The second inequality stems from the contractivity of any Lindbladian evolution.
\end{proof}

\appendix

\section{Conditional expectations and entropy production}\label{sec:appendix:full-rank}
For the convenience of the reader, this section compiles some properties of conditional expectations and the entropy production, which are well-known but partially hard to locate. For simplicity, we will first focus on full-rank input states only. In \cref{pro:nonfullrank} we will then show that we can extend a decay of the relative entropy of full-rank states to all states, justifying this simplification. 

First, let us show that taking the conditional expectation can not decrease the rank of a state:
\begin{lemma}\label{lem:condexpkern}
  Let $\LL$ be a Lindbladian (in Heisenberg picture) on a finite-dimensional Hilbert space $\HH$ with a full rank fixed point $\rho\in \SSS(\HH)$, $\rho>0$, $\LL^*(\rho)=0$, and with a well-defined conditional expectation. Then, for any state $\sigma \in \SSS(\HH)$:
  \begin{equation}
    \ker\E_{\LL}^*(\sigma) \subseteq \ker\sigma
  \end{equation}
  where $\displaystyle \E_{\LL}^* = \lim_{t\to \infty} e^{t\LL^*}$ is the dual of the conditional expectation of $\LL$.
\end{lemma}
\begin{proof}

  As $\LL^*$ has a full rank fixed point, there exists a decomposition of the Hilbert space
  \begin{equation*}
    \HH = \bigoplus_{i\in I} \HH_i \otimes \KK_i
  \end{equation*}
  such that the conditional expectation has the following decomposition \cite{deschamps_Structureuniformlycontinuous_2016}:
  \begin{equation*}
    \E^*_{\LL} (\sigma) = \sum_{i\in I} \Tr_{\KK_i}(P_i \sigma P_i)\otimes \tau_i
  \end{equation*}
  where $P_i$ is the projection onto the $i$-th block $\HH_i\otimes \KK_i$ and where $\tau_i$ are full-rank states on $\KK_i$.
  First, since $\E_{\LL}^*(\sigma)$ is block-diagonal, it's kernel is
  \begin{equation*}
    \ker(\E^*_{\LL}(\sigma)) = \bigoplus_{i\in I} \ker\left(\Tr_{\KK_i}(P_i \sigma P_i)\otimes \tau_i\right) \ .
  \end{equation*}
  Without loss of generality, fix some $i\in I$. Since $\tau_i$ is full-rank, its kernel is empty and thus
  \begin{equation*}
    \ker\left(\Tr_{\KK_i}(P_i \sigma P_i)\otimes \tau_i\right) = \ker\left(\Tr_{\KK_i}(P_i \sigma P_i)\right) \otimes \KK_i \ .
  \end{equation*}
  Let $\ket{\psi}\in \HH_i$ be an element of the kernel of $\Tr_{\KK_i}(P_i \sigma P_i)$. Then
  \begin{equation*}
    0 = \Tr_{\HH_i}(\Tr_{\KK_i}(P_i\sigma P_i) \ketbra{\psi}{\psi}) =  \Tr_{\HH_i\otimes\KK_i}(\ketbra{\psi}{\psi}P_i\sigma P_i\ketbra{\psi}{\psi})
  \end{equation*}
  and, since $\ketbra{\psi}{\psi}P_i\sigma P_i\ketbra{\psi}{\psi}=O^\dag O$, $0=O=\sqrt{\sigma}P_i\ketbra{\psi}{\psi}$ and thus also
  $\sigma P_i\ketbra{\psi}{\psi} = \sigma \ketbra{\psi}{\psi} =0$. Where we embedded $\ket{\psi}$ into $\HH$ in the last step. Thus, for any element $ \sum_{k}\ket{\psi_k}\otimes \ket{\phi_k}\in \ker\left(\Tr_{\KK_i}(P_i \sigma P_i)\right) \otimes \KK_i$ it holds that
  \begin{equation*}
    \sigma \sum_k \ket{\psi_k}\otimes \ket{\phi_k} = \sum_k \sigma \ketbra{\psi_k}{\psi_k}\ket{\psi_k}\otimes \ket{\phi_k} =0 , 
  \end{equation*}
  and hence $\sum_{k}\ket{\psi_k}\otimes \ket{\phi_k}  \in \ker \sigma $. 
\end{proof}

\begin{corollary}\label{cor:relentropywelldef}
  Let $\LL$ be a Lindbladian (in Heisenberg picture) on a finite-dimensional Hilbert space $\HH$ with a full rank fixed point $\rho\in \SSS(\HH)$, $\rho>0$, $\LL^*(\rho)=0$, and with a well-defined conditional expectation. Then, for any state $\sigma \in \SSS(\HH)$ and any $t\geq 0$:
  \begin{equation}
    D(e^{t\LL} \sigma \| \E_{\LL}^*(\sigma)) < \infty
  \end{equation}
\end{corollary}
\begin{proof}
  By \cref{lem:condexpkern}, $\ker\E_{\LL}^*(\sigma) \subseteq \ker\sigma$ for any $\sigma$. In particular,
  \begin{equation}
    \ker(\E_{\LL}^*( \sigma)) = \ker(\E_{\LL}^* \circ e^{t\LL}(\sigma)) \subseteq \ker (e^{t\LL} \sigma) .
  \end{equation}
  This completes the proof, since the relative entropy is infinite only if the kernel of the second argument is not contained in the kernel of the first.
\end{proof}

The following proposition collects mostly from \cite{spohn_Entropyproductionquantum_1978} some results about the entropy production, particularly its additivity. Here, the assumption of a full-rank state is very helpful, since otherwise the logarithm of $\sigma$ diverges.
\begin{proposition}\label{pro:entropyprodprop}
  Let $\LL =\sum_{v\in \lsym}\LL_v$ be a local, frustration-free Lindbladian (in Heisenberg picture) on a finite-dimensional Hilbert space $\HH$ with  full rank fixed point $\rho\in \SSS(\HH)$, that is $\rho>0$, $\LL^*(\rho)=0$, and with a well-defined conditional expectation. Then, for any other full rank state $\sigma \in \SSS(\HH)$ and any two regions $R_1, R_2 \subseteq \lsym$ the entropy production satisfies:
  \begin{enumerate}
    \item positivity: $EP_{R_1}(\sigma)\geq0$,
    \item linearity: if $R_1$ and $R_2$ are disjoint, $EP_{R_1}(\sigma) + EP_{R_2}(\sigma) = EP_{R_1 \cup R_2}(\sigma)$,
    \item monotonicity: if $R_1\subseteq R_2$, $EP_{R_1}(\sigma) \leq EP_{R_2}(\sigma)$ .
  \end{enumerate}
\end{proposition}
\begin{proof}
  As is shown in \cite{spohn_Entropyproductionquantum_1978} positivity follows from the monotonicity of $ t \mapsto D(e^{t\LL^*_{R}}(\sigma)\|\E^*_{R}(\sigma))$, which is finite by \cref{cor:relentropywelldef}. The monotonicity in the region follows directly from linearity and positivity.
  To show linearity, consider the explicit expression from \cite{spohn_Entropyproductionquantum_1978}:
  \begin{equation*}
    EP_{R_i}(\sigma) = -\Tr\left(\LL^*_{R_i}(\sigma)\left(\ln(\sigma) - \ln(\E^*_{R_i}(\sigma))\right)\right)\ ,
  \end{equation*}
  where $\E^*_{R_i}:=\E^*_{\LL_{R_i}}$ is the conditional expectation for the region $R_i$.
  Since the Lindbladian is frustration-free, $\rho$
  is a full rank fixed point of $\LL^*_{R_1}$, $\LL^*_{R_2}$ and $\LL^*_{R_1 \cup R_2}$. Using \cref{pro:logdiff} below for each of the three Lindbladians, we can swap the fixed points to find:
  \begin{align*}
    EP_{R_1}(\sigma) + EP_{R_2}(\sigma)
    &= -\Tr\left((\LL^*_{R_1}(\sigma) +\LL^*_{R_2}(\sigma) )\left(\ln(\sigma) - \ln(\rho)\right)\right) \\
    & =  -\Tr\left(\LL^*_{R_1\cup R_2}(\sigma)\left(\ln(\sigma) - \ln(\rho)\right)\right) \\
    &= -\Tr\left(\LL^*_{R_1 \cup R_2}(\sigma) \left(\ln(\sigma) - \ln(\E^*_{R_1\cup R_2}(\sigma))\right)\right) \\
    &= EP_{R_1 \cup R_2}(\sigma)\ .
  \end{align*}
  On the second line, we used that $R_1$ and $R_2$ are disjoint. In the last step, we also used \cref{lem:condexpkern}, which ensures that $ \E^*_{R_1\cup R_2}(\sigma)$ is full rank. 
\end{proof}

\begin{lemma}
    \label{pro:logdiff}
  Let $\LL$ be a Lindbladian (in Heisenberg picture) on a finite-dimensional Hilbert space $\HH$ with a full rank fixed point $\rho\in \SSS(\HH)$, $\rho>0$, $\LL^*(\rho)=0$, and with a well-defined conditional expectation. Then, for any state $\sigma \in \SSS(\HH)$ and any two full rank fixed points $\rho_1$ and $\rho_2$ of $\LL^*$:
  \begin{equation}
    \Tr\left(\LL^*(\sigma)\left(\ln(\rho_1) - \ln(\rho_2)\right)\right) =0 \ .
  \end{equation}
\end{lemma}
\begin{proof}
  As $\LL^*$ has a full-rank fixed point, there exists a decomposition of the Hilbert space
  \begin{equation*}
    \HH = \bigoplus_{i\in I} \HH_i \otimes \KK_i
  \end{equation*}
  such that the conditional expectation has the following decomposition \cite{deschamps_Structureuniformlycontinuous_2016}:
  \begin{equation*}
    \E^*_{\LL} (\sigma) = \sum_{i\in I} \Tr_{\KK_i}(P_i \sigma P_i)\otimes \tau_i
  \end{equation*}
  where $P_i$ is the projection onto the $i$-th block $\HH_i\otimes \KK_i$ and where $\tau_i$ are full-rank states on $\KK_i$. Furthermore, in this decomposition, the kernel of $\LL$ is given by $\bigoplus_{i\in I} \BB(\HH_i)\otimes \identity_{\KK_i}$.
  Since $\rho_1$ and $\rho_2$ have full rank, we can directly evaluate the difference of logarithms, using that $\rho_j = \E^*_{\LL}(\rho_j)$ for $j=1,2$:
  \begin{equation*}
    \ln(\rho_1) - \ln(\rho_2) = \sum_{i\in I}\left(\ln(\Tr_{\KK_i}(P_i \rho_1 P_i))-\ln(\Tr_{\KK_i}(P_i \rho_2 P_i))\right)\otimes\identity_{\KK_i}
  \end{equation*}
  which is an element of the kernel of $\LL$.
  Thus,
  \begin{equation*}
    \Tr\left(\LL^*(\sigma)\left(\ln(\rho_1) - \ln(\rho_2)\right)\right) = \Tr\left(\sigma \LL\left(\ln(\rho_1) - \ln(\rho_2)\right)\right) = 0 . 
  \end{equation*}
\end{proof}

The last proposition of this section is a density argument, showing that a decay of the relative entropy and thus also rapid mixing can be lifted from full-rank states to all states, using a continuity bound on the relative entropy \cite{hiai_Differentquantumfdivergences_2017}. 
\begin{proposition}\label{pro:nonfullrank}
  Let $\LL$ be a Lindbladian (in Heisenberg picture) on a finite-dimensional Hilbert space $\HH$ with a full rank fixed point $\rho\in \SSS(\HH)$, $\rho>0$, $\LL^*(\rho)=0$, and with a well-defined conditional expectation. Assume that the following inequality holds for all $t\geq 0$, all full rank states $\sigma\in \SSS(\HH)$ and some positive function $f(t)$:
  \begin{equation}
    D(e^{t\LL^*}\sigma \| \E^*_{\LL}\sigma)\leq f(t) D(\sigma \| \E^*_{\LL}\sigma)
  \end{equation}
  Then the inequality holds for all states $\sigma\in \SSS(\HH)$, not necessarily of full rank.
\end{proposition}
\begin{proof}
  By \cite[Proposition 3.8]{hiai_Differentquantumfdivergences_2017} the relative entropy of any two states $\sigma, \sigma'\in \SSS(\HH)$ can be approximated by 
  \begin{equation}\label{eq:approxentr}
    D(\sigma\|\sigma') = \lim_{n\to\infty} D(\sigma + \xi_n\|\sigma'+\xi_n)
  \end{equation}
  for any sequence $\xi_n$ of positive, bounded operators with $\lim_{n\to\infty}\xi_n=0$ and as long as $\sigma+\xi_n>0$ and $\sigma'+\xi_n>0$ for all $n$.
  Recall that $\rho>0$ is a fixed point of $e^{t\LL^*}$ and thus also of $\E^*_{\LL}$. Consider $\xi_n=n^{-1}\rho>0$ and note that $\sigma+\frac{1}{n}\rho$ is of full rank for any state $\sigma$.Thus, for any (not necessarily full rank) state $\sigma$, we find
  \begin{align*}
    D(e^{t\LL^*}\sigma \| \E^*_{\LL}\sigma) &= \lim_{n\to\infty} D(e^{t\LL^*}(\sigma) + n^{-1}\rho\| \E^*_{\LL}(\sigma)+n^{-1}\rho)\\
    &= \lim_{n\to\infty} D(e^{t\LL^*}(\sigma + n^{-1}\rho)\| \E^*_{\LL}(\sigma+n^{-1}\rho))\\
    &\leq f(t)  \lim_{n\to\infty} D(\sigma + n^{-1}\rho\| \E^*_{\LL}(\sigma+n^{-1}\rho))\\
    &= f(t) D(\sigma \| \E^*_{\LL}\sigma) \ ,
  \end{align*}
  where the last step is based on the linearity of the conditional expectation and~\eqref{eq:approxentr}. Note that we assume the latter only for normalized states but applied it to non-normalized states. However, for any two normalized states $\sigma_1,\sigma_2$ and any positive number $c$, it holds that
  \begin{equation*}
      D(c\sigma_1\|c\sigma_2) = \Tr\left(c\sigma_1 (\ln\sigma_1+\ln c -\ln\sigma_2 -\ln c)\right) = c D(\sigma_1\|\sigma_2) \ .
  \end{equation*}
  Thus, \eqref{eq:approxentr} can be extended to non-normalized states $\sigma$.
\end{proof}

\section*{Acknowledgments}
 The DFG supported this work under the grants TRR 352–Project-ID 470903074 (A.\,C., C.\,R., S.\,S., S.\,W.) and EXC-2111-390814868 (S.W.). A.\,L.~acknowledges support from the Italian Ministry of University and Research (MUR), through ``Programma per Giovani Ricercatori Rita Levi Montalcini'', as well as the grant ``Dipartimento di Eccellenza 2023-2027'' of Dipartimento di Matematica, Politecnico di Milano. L.\,G.~acknowledges support from the National Natural Science Foundation of China (NNSFC) through program "Excellent Young Scientists Fund (Overseas)" and the grant NNSFC-12401163. 
S.\,S. and S.\,W. would like to thank Princeton University for its hospitality. S.\,S. would like to thank the University of Tübingen for its hospitality. D.\,P-G. and A.\,P-H.~acknowledge financial support from the following grants PID2020-113523GB-I00 and PID2023-146758NB-I00 and funded by MICIU/AEI/10.13039/501100011033. D.P-G.acknowledges support from grant  TEC-2024/COM-84-QUITEMAD-CM, funded by Comunidad de Madrid, and grant CEX2023-001347-S, funded by MICIU/AEI/10.13039/501100011033. This work has been financially supported by the Ministry for Digital Transformation and the Civil Service of the Spanish Government through the QUANTUM ENIA project call – Quantum Spain project, and by the European Union through the Recovery, Transformation and Resilience Plan – NextGenerationEU within the framework of the Digital Spain 2026 Agenda. This project was funded within the QuantERA II Programme which has received funding from the EU’s H2020 research and innovation programme under the GA No 101017733.
\sloppy
        \bibliographystyle{abbrvurl}
        \bibliography{lit}

        \end{document}

%% file: tikz_css_code_examples.tex
\begin{tikzpicture}[scale=0.7]

    \pgfmathsetmacro{\sizex}{4}
    \pgfmathsetmacro{\sizey}{4}
    \pgfmathsetmacro{\poscellx}{1.5}
    \pgfmathsetmacro{\poscelly}{-1.5}
    
    \def\region{33,34,35,36,44,45,46,55,56}


    
    \definecolor{strbgcolor}{HTML}{86abbf}
    \definecolor{plqbgcolor}{HTML}{ab8d67}
    \definecolor{strfgcolor}{HTML}{00A6FF}
    \definecolor{plqfgcolor}{HTML}{FF9000}
    \definecolor{qubitcolor}{HTML}{000000}

    \pgfdeclarelayer{background}
    \pgfdeclarelayer{foreground}
    \pgfdeclarelayer{nodes}
    \pgfdeclarelayer{hexagon}
    \pgfdeclarelayer{toric}
    \pgfdeclarelayer{3dtoric}
    \pgfdeclarelayer{tesselation}
    \pgfsetlayers{background, foreground, nodes, hexagon, toric, 3dtoric, tesselation}
    \begin{pgfonlayer}{nodes}
        \foreach \row in {0,...,\sizex} {
            \foreach \column in {0,..., \sizey}{
                \fill[qubitcolor] (\column+1,\row) circle (3pt);
                \node[coordinate] (\column\row) at (\column+1,\row) {};
            }
        }
    \end{pgfonlayer}
    \begin{pgfonlayer}{background}
        \foreach \row in {1, ..., \sizex} {
            \foreach \column in {1, 3,..., \sizey} {
                \fill[strfgcolor] ({\column + mod(\row,2)}, \row) rectangle +(1,-1);
                \fill[plqfgcolor] ({\column - mod(\row,2)+1}, \row) rectangle +(1,-1);
            }
        }
        \foreach \x in {1, 3,..., \sizex} {
            \fill[strfgcolor] (1,\x+1) -- (1,\x) arc[start angle=90, delta angle=180, radius=0.5] -- cycle;
            \fill[strfgcolor] (\sizey+1,\x+1) -- (\sizey+1,\x) arc[start angle=270, delta angle=180, radius=0.5] -- cycle;
        }
        \foreach \y in {1, 3,..., \sizey} {
            \fill[plqfgcolor] (\y+2,0) -- (\y+1,0) arc[start angle=180, delta angle=180, radius=0.5] -- cycle;
            \fill[plqfgcolor] (\y+2,\sizex) -- (\y+1,\sizex) arc[start angle=0, delta angle=180, radius=0.5] -- cycle;
        }

        \fill[strfgcolor] (\poscelly, \poscellx-1) rectangle +(1,1);
        \fill[plqfgcolor] (\poscelly, \poscellx+1) rectangle +(1,1);
    \end{pgfonlayer}

    \begin{pgfonlayer}{foreground}
        \draw[red, line width=4pt] (20)--(24);
        \draw[red, line width=4pt] (03)--(43);
        \node at (\poscelly, \poscellx-1) {$Z$};
        \node at (\poscelly, \poscellx) {$Z$};
        \node at (\poscelly+1, \poscellx-1) {$Z$};
        \node at (\poscelly+1, \poscellx) {$Z$};

        \node at (\poscelly, \poscellx+1) {$X$};
        \node at (\poscelly, \poscellx+2) {$X$};
        \node at (\poscelly+1, \poscellx+1) {$X$};
        \node at (\poscelly+1, \poscellx+2) {$X$};

        \node at (-1.1, -1.5) {\large\textbf{(a)}};
        \node at (3, -1.5) {\large\textbf{(b)}};
        \node at (8.5, -1.5) {\large\textbf{(c)}};
        \node at (14, -1.5) {\large\textbf{(d)}};
        \node at (1.5, -7.5) {\large\textbf{(e)}};
        \node at (8.5, -7.5) {\large\textbf{(f)}};
    \end{pgfonlayer}

    \begin{pgfonlayer}{toric}
        \begin{scope}[xshift=7cm, yshift=1.5cm]

        \draw[line width=3pt, plqfgcolor] (0,1)--(2,1);
        \draw[line width=3pt,plqfgcolor] (1,2)--(1,0);
        \node at (-0.3,1) {$X$};
        \node at (2.3,1) {$X$};
        \node at (1,2.3) {$X$};
        \node at (1,-0.3) {$X$};

        \fill[strfgcolor] (1.7,0.3) rectangle ++(1.8,-1.8);
        \node[circle, fill=white, inner sep=1pt] at (1.7, -0.6) {$Z$};
        \node[circle, fill=white, inner sep=1pt] at (1.7+1.8, -0.6) {$Z$};
        \node[circle, fill=white, inner sep=1pt] at (1.7+0.9, 0.3) {$Z$};
        \node[circle, fill=white, inner sep=1pt] at (1.7+0.9, 0.3-1.8) {$Z$};

        \end{scope}
    \end{pgfonlayer}

    \begin{pgfonlayer}{hexagon}
        \begin{scope}[xshift=13cm]

        \coordinate (c1) at (0, 0);
        \coordinate (c2) at (0, 1.5);
        \coordinate (c3) at (1.3, 0.75);
        \coordinate (c4) at (1.3, 2.25);
        \coordinate (c5) at (2.6, 1.5);
        \coordinate (c6) at (2.6, 3);
        \coordinate (c7) at (0, 3);
        \coordinate (c8) at (1.3, 3.75);
        
        \foreach \c in {c1,c2,c3,c4,c5,c6,c7,c8}{
            \node[draw, regular polygon, regular polygon sides=6, minimum width=1.21cm, fill=plqfgcolor!90!white, line width=1pt] at (\c) {};
        }
        \foreach \c in {c1,c2,c3,c4,c5,c6,c7,c8}{
            \fill[qubitcolor] ($ (\c) +(0:0.85cm)$) circle (3pt);
            \fill[qubitcolor] ($ (\c) +(60:0.85cm)$) circle (3pt);
            \fill[qubitcolor] ($ (\c) +(120:0.85cm)$) circle (3pt);
            \fill[qubitcolor] ($ (\c) +(180:0.85cm)$) circle (3pt);
            \fill[qubitcolor] ($ (\c) +(240:0.85cm)$) circle (3pt);
            \fill[qubitcolor] ($ (\c) +(300:0.85cm)$) circle (3pt);
        }
        
        \node at (0:0.53cm) {$\scriptstyle X$};
        \node at (60:0.53cm) {$\scriptstyle X$};
        \node at (120:0.53cm) {$\scriptstyle X$};
        \node at (180:0.53cm) {$\scriptstyle X$};
        \node at (240:0.53cm) {$\scriptstyle X$};
        \node at (300:0.53cm) {$\scriptstyle X$};
        \node at ($ (c3) +(0:0.53cm)$) {$\scriptstyle X$};
        \node at ($ (c3) +(60:0.53cm)$) {$\scriptstyle X$};
        \node at ($ (c3) +(120:0.53cm)$) {$\scriptstyle X$};
        \node at ($ (c3) +(180:0.53cm)$) {$\scriptstyle X$};
        \node at ($ (c3) +(240:0.53cm)$) {$\scriptstyle X$};
        \node at ($ (c3) +(300:0.53cm)$) {$\scriptstyle X$};
        \node at ($ (c2) +(0:0.53cm)$) {$\scriptstyle X$};
        \node at ($ (c2) +(60:0.53cm)$) {$\scriptstyle X$};
        \node at ($ (c2) +(120:0.53cm)$) {$\scriptstyle X$};
        \node at ($ (c2) +(180:0.53cm)$) {$\scriptstyle X$};
        \node at ($ (c2) +(240:0.53cm)$) {$\scriptstyle X$};
        \node at ($ (c2) +(300:0.53cm)$) {$\scriptstyle X$};

        \end{scope}
    \end{pgfonlayer}

    \begin{pgfonlayer}{3dtoric}
        \begin{scope}[xshift=0cm, yshift=-5cm]

        \draw[line width=3pt, plqfgcolor] (-1,0,0)--(1,0,0);
        \draw[line width=3pt,plqfgcolor] (0,-1,0)--(0,1,0);
        \draw[line width=3pt,plqfgcolor] (0,0,-1)--(0,0,1);
        \node at (-1.3,0,0) {$X$};
        \node at (1.3,0,0) {$X$};
        \node at (0,1.3,0) {$X$};
        \node at (0,-1.3,0) {$X$};
        \node at (0,0,-1.6) {$X$};
        \node at (0,0,1.6) {$X$};

        \fill[strfgcolor] (2,0.3,0) rectangle ++(1.8,-1.8,0);
        
        \fill[strfgcolor] (2.1,1.1,0) -- ++(0,0,-1.8) -- ++(1.8,0,0) -- ++(0,0,1.8) -- cycle;
        
        \fill[strfgcolor] (4.4,0.4,0) -- ++(0,0,-1.8) -- ++(0,-1.8,0) -- ++(0,0,1.8) -- cycle;
        
        \node[circle, fill=white, inner sep=0.6pt] at ($(2,0.3,0) + (0.9,0,0)$) {$Z$};
        \node[circle, fill=white, inner sep=0.6pt] at ($(2,0.3,0) + (0.9,-1.8,0)$) {$Z$};
        \node[circle, fill=white, inner sep=0.6pt] at ($(2,0.3,0) + (0,-0.9,0)$) {$Z$};
        \node[circle, fill=white, inner sep=0.6pt] at ($(2,0.3,0) + (1.8,-0.9,0)$) {$Z$};

        \node[circle, fill=white, inner sep=0.6pt] at ($(2.1,1.1,0) + (0.9,0,0)$) {$Z$};
        \node[circle, fill=white, inner sep=0.6pt] at ($(2.1,1.1,0) + (0.9,0,-1.8)$) {$Z$};
        \node[circle, fill=white, inner sep=0.6pt] at ($(2.1,1.1,0) + (0,0,-0.9)$) {$Z$};
        \node[circle, fill=white, inner sep=0.6pt] at ($(2.1,1.1,0) + (1.8,0,-0.9)$) {$Z$};
        
        \node[circle, fill=white, inner sep=0.6pt] at ($(4.4,0.4,0) + (0,-0.9,0)$) {$Z$};
        \node[circle, fill=white, inner sep=0.6pt] at ($(4.4,0.4,0) + (0,-0.9,-1.8)$) {$Z$};
        \node[circle, fill=white, inner sep=0.6pt] at ($(4.4,0.4,0) + (0,0,-0.9)$) {$Z$};
        \node[circle, fill=white, inner sep=0.6pt] at ($(4.4,0.4,0) + (0,-1.8,-0.9)$) {$Z$};

        \end{scope}
    \end{pgfonlayer}

    \begin{pgfonlayer}{tesselation}
        \begin{scope}[xshift=10cm, yshift=-5cm]

        \coordinate (hexstr) at (-2.5, 0);
        \draw[line width=3pt, plqfgcolor] (hexstr)-- ++(60:1);
        \draw[line width=3pt, plqfgcolor] (hexstr)-- ++(180:1);
        \draw[line width=3pt, plqfgcolor] (hexstr)-- ++(300:1);
        \node[circle, fill=white, inner sep=0.6pt] at ($(hexstr) + (60:1)$) {$X$};
        \node[circle, fill=white, inner sep=0.6pt] at ($(hexstr) + (180:1)$) {$X$};
        \node[circle, fill=white, inner sep=0.6pt] at ($(hexstr) + (300:1)$) {$X$};

        \node[regular polygon, regular polygon sides=6, minimum width=2cm, fill=strfgcolor] at (0,0,0) {};
        \node[circle, fill=white, inner sep=0.6pt] at (30:1.2cm) {$Z$};
        \node[circle, fill=white, inner sep=0.6pt] at (90:1.2cm) {$Z$};
        \node[circle, fill=white, inner sep=0.6pt] at (150:1.2cm) {$Z$};
        \node[circle, fill=white, inner sep=0.6pt] at (210:1.2cm) {$Z$};
        \node[circle, fill=white, inner sep=0.6pt] at (270:1.2cm) {$Z$};
        \node[circle, fill=white, inner sep=0.6pt] at (330:1.2cm) {$Z$};

        \end{scope}
    \end{pgfonlayer}

\end{tikzpicture}

%% file: tikz_surface_code_boundary.tex
\begin{tikzpicture}[scale=0.7]

    \pgfmathsetmacro{\sizex}{4}
    \pgfmathsetmacro{\sizey}{6}
    \pgfmathsetmacro{\poscellx}{3}
    \pgfmathsetmacro{\poscelly}{-2.75}
    
    \def\region{33,34,35,36,44,45,46,55,56}


    
    \definecolor{plqfgcolor}{HTML}{00A6FF}
    \definecolor{strfgcolor}{HTML}{FF9000}
    \definecolor{qubitcolor}{HTML}{000000}

    \colorlet{plqbgcolor}{plqfgcolor!30!white}
    \colorlet{strbgcolor}{strfgcolor!30!white}

    \pgfdeclarelayer{background}
    \pgfdeclarelayer{foreground}
    \pgfdeclarelayer{nodes}
    \pgfdeclarelayer{nodes_foreground}
    \pgfdeclarelayer{hexagon}
    \pgfsetlayers{background, foreground, nodes, nodes_foreground, hexagon}
    \begin{pgfonlayer}{nodes}
        \foreach \row in {0,...,\sizex} {
            \foreach \column in {0,..., \sizey}{
                \fill[qubitcolor] (\column+1,\row) circle (3pt);
                \node[coordinate] (\column\row) at (\column+1,\row) {};
            }
        }
        \foreach \row in {0,...,\sizex} {
            \draw (\sizey+2,\row) node[cross out, draw=red!80!black, minimum size=6pt, inner sep=0pt, outer sep=0pt, line width=2pt] {};
        }
        \foreach \column in {0,..., \numexpr\sizey+1} {
            \draw (\column+1,-1) node[cross out, draw=red!80!black, minimum size=6pt, inner sep=0pt, outer sep=0pt, line width=2pt] {};
        }

        \fill[qubitcolor] (\poscelly+0.5,\poscellx+0.5) circle (3pt);
        \fill[qubitcolor] (\poscelly-0.5,\poscellx+0.5) circle (3pt);
        \fill[qubitcolor] (\poscelly+0.5,\poscellx-0.5) circle (3pt);
        \fill[qubitcolor] (\poscelly-0.5,\poscellx-0.5) circle (3pt);
    \end{pgfonlayer}
    \begin{pgfonlayer}{background}
        
        \foreach \row in {0, ..., \numexpr\sizex+1} {
            \foreach \column in {1, 3,..., \numexpr\sizey+1} {
                \fill[strbgcolor] ({\column + mod(\row,2)-1}, \row) rectangle +(1,-1);
                \fill[plqbgcolor] ({\column - mod(\row,2)}, \row) rectangle +(1,-1);
            }
        }
        \foreach \row in {1, ..., \sizex} {
            \foreach \column in {1, 3,..., \sizey} {
                \fill[plqfgcolor] ({\column + mod(\row,2)}, \row) rectangle +(1,-1);
                \fill[strfgcolor] ({\column - mod(\row,2)+1}, \row) rectangle +(1,-1);
            }
        }
        \foreach \x in {1, 3,..., \sizex} {
            \fill[plqfgcolor] (1,\x+1) -- (1,\x) arc[start angle=90, delta angle=180, radius=0.5] -- cycle;
            \fill[plqfgcolor] (\sizey+1,\x+1) -- (\sizey+1,\x) arc[start angle=270, delta angle=180, radius=0.5] -- cycle;
        }
        \foreach \y in {1, 3,..., \sizey} {
            \fill[strfgcolor] (\y+2,0) -- (\y+1,0) arc[start angle=180, delta angle=180, radius=0.5] -- cycle;
            \fill[strfgcolor] (\y+2,\sizex) -- (\y+1,\sizex) arc[start angle=0, delta angle=180, radius=0.5] -- cycle;
        }

        \fill[plqfgcolor] (\poscelly-0.5, \poscellx-0.5) rectangle +(1,1);
        \fill[plqfgcolor] (\poscelly-1.5, \poscellx+0.5) rectangle +(1,1);
        \fill[strfgcolor] (\poscelly-0.5, \poscellx+0.5) rectangle +(1,1);
        \fill[strfgcolor] (\poscelly-1.5, \poscellx-0.5) rectangle +(1,1);
    \end{pgfonlayer}

    \begin{pgfonlayer}{foreground}
        \draw[line width=1.5pt,dashed,xshift=0.5cm,yshift=-1.5cm] (0,0) grid[step=2] (\sizey+2,\sizex+2);
        \draw[line width=1.8pt,xshift=0.5cm,yshift=-1.5cm] (2,2) grid[step=2] (\sizey,\sizex);
        
        \draw[line width=1.8pt] (\poscelly-1,\poscellx-1) rectangle +(2,2);

    \end{pgfonlayer}
    \begin{pgfonlayer}{hexagon}
        \begin{scope}[xshift=-2.75cm, yshift=-0.6cm]

        \coordinate (c1) at (0, 0);
        \coordinate (c2) at (0, 1.47);
        \coordinate (c3) at (1.275, 0.735);
        \coordinate (c4) at (1.275, 2.21);
        \coordinate (c0) at (-1.275, 0.735);
        
        \foreach \c in {c1,c3,c0}{
            \node[draw, regular polygon, regular polygon sides=6, minimum width=1.21cm, fill=strbgcolor, line width=1pt] at (\c) {};
        }
        \node[draw, regular polygon, regular polygon sides=6, minimum width=1.21cm, fill=strfgcolor, line width=1pt] at (c2) {};
        \fill[qubitcolor] ($ (c1) +(60:0.85cm)$) circle (3.5pt);
        \fill[qubitcolor] ($ (c1) +(120:0.85cm)$) circle (3.5pt);
        \draw[line width=1.8pt] (c0)--(c1)--(c3)--(c2)--cycle;

        \end{scope}

    \end{pgfonlayer}

\end{tikzpicture}

%% file: tikz-approxtensor.tex
\begin{tikzpicture}[scale=1]
\begin{scope}[shift={(-4,0)}, scale=0.8]
    \draw[orange, fill=orange, fill opacity=0.1, line width=1.5pt] (-1, -1) rectangle (8, 5);
    \fill[white] (0, 0) rectangle (7, 4);

    \draw[blue, fill=blue, fill opacity=0.1, line width=1.5pt] (0, 0) rectangle (3.5, 4);

    \draw[red, fill=red, fill opacity=0.1, line width=1.5pt] (2.5, 0) rectangle (7, 4);
    
    \node at (-0.5, 4.5) {$C$};
    \node at (1.25, 3.3) {$U$};
    \node at (3, 3.3) {$V$};
    \node at (5.25, 3.3) {$W$};

    \draw[decoration={brace,mirror,raise=5pt},decorate] (3.5,0) -- node[right=5pt] {$ \leq 2 \dist(U,W)^2$} (3.5,4);

    \draw[decoration={brace,mirror,raise=3pt},decorate] (7,4) -- node[above=4pt] {$ \leq 2 \dist(U,W)^2$} (3.5,4);

    \draw[decoration={brace,mirror,raise=6pt},decorate] (2.5,0) -- node[below=7pt] {$\dist(U,W)$} (3.5,0);
    \draw[decoration={brace,mirror,raise=3pt,aspect=0.7},decorate] (0,0) -- node[pos=0.7, below=4pt] {$l_1$} (7,0);
    
\end{scope}

\begin{scope}[shift={(4,0)}, scale=0.8]
    \fill[orange, fill opacity=0.1] (-1, -1) rectangle (6, 5);
    \draw[orange, line width=1.5pt] (-1, -1) -- (6, -1);
    \draw[orange, line width=1.5pt] (-1, 5) -- (6, 5);
    \fill[white] (-1, 0) rectangle (6, 4);

    \fill[blue, fill opacity=0.1] (-1, 0) rectangle (3, 4);
    \draw[blue, line width=1.5pt] (-1, 0) -- (3, 0) -- (3,4) -- (-1,4);
    \fill[blue, fill opacity=0.1] (6, 0) rectangle (4.5, 4);
    \draw[blue, line width=1.5pt] (6, 0) -- (4.5, 0) -- (4.5,4) -- (6,4);

    \fill[red, fill opacity=0.1] (2, 0) rectangle (6, 4);
    \draw[red, line width=1.5pt] (6,4) -- (2,4) -- (2, 0) -- (6, 0);
    \fill[red, fill opacity=0.1] (-1, 0) rectangle (0.5, 4);
    \draw[red, line width=1.5pt] (-1, 0) -- (0.5, 0) -- (0.5,4) -- (-1,4);

    \draw[decoration={brace,mirror,raise=3pt,aspect=0.5},decorate] (-1,0) -- node[pos=0.5, below=4pt] {$l_1$} (6,0);

    \node at (0, 4.5) {$C$};
    \node at (1.25, 2) {$U$};
    \node at (2.5, 2) {$V_1$};
    \node at (5, 2) {$V_2$};
    \node at (0, 2) {$V_2$};
    \node at (3.75, 2) {$W$};

    \fill[white] (-0.9,-1.05) rectangle (-0.8,5.05);
    \fill[white] (-0.7,-1.05) rectangle (-0.6,5.05);
    \fill[white] (-0.5,-1.05) rectangle (-0.4,5.05);

    \fill[white] (5.9,-1.05) rectangle (5.8,5.05);
    \fill[white] (5.7,-1.05) rectangle (5.6,5.05);
    \fill[white] (5.5,-1.05) rectangle (5.4,5.05);

\end{scope}
\end{tikzpicture}

%% file: tikz_dobrushin_shlosman.tex
\begin{tikzpicture}[scale=1]
    \draw[orange, fill=orange, fill opacity=0.1, line width=1.5pt] (0, 0) rectangle (5,4);
    \fill[white] (0.8, 0.8) rectangle (4.2,3.2);

    \draw[blue, fill=blue, fill opacity=0.1, line width=1.5pt] (0.8, 0.8) rectangle (4.2,3.2);

    \draw[red, fill=red, fill opacity=0.1, line width=1.5pt] (1.7, 1.5) rectangle (2.8, 2.5);
    
    \node at (0.35, 3.6) {$\lsym$};
    \node at (1.18, 2.8) {$R'$};
    \node at (2.23, 2) {$R$};
    \node at (4.4,2.5)[circle, fill=black, inner sep=1.3pt] {};
    \node at (4.7,2.5) {$v$};
\end{tikzpicture}

%% file: tikz_ising_partition.tex
\begin{tikzpicture}[scale=0.6]

    \pgfmathsetmacro{\sizex}{7}
    \pgfmathsetmacro{\sizey}{7}
    \pgfmathsetmacro{\shifty}{1}
    \def\region{33,34,35,43,44,45,53,55}
    \pgfmathsetmacro{\scalefactor}{0.6}
    \pgfmathsetmacro{\intboxscale}{0.2}
    
    \definecolor{strbgcolor}{HTML}{CF5219}
    \definecolor{plqbgcolor}{HTML}{29A6FF} 
    \definecolor{qubitcolor}{HTML}{000000}

    \pgfdeclarelayer{background}
    \pgfdeclarelayer{foreground}
    \pgfdeclarelayer{nodes}
    \pgfdeclarelayer{nodes_foreground}
    \pgfsetlayers{background, foreground, nodes, nodes_foreground}

    \begin{pgfonlayer}{nodes}
        \foreach \row in {1,...,\sizex} {
            \foreach \column in {1,..., \sizey}{
                \fill[qubitcolor] (\column,\row) circle (3pt);
                \node[coordinate] (\column\row) at (\column,\row) {};
            }
        }
        \foreach \n in \region {
            \node[coordinate] (z\n) at ($(\n) + (\sizey,0) + (\shifty,0) - \scalefactor*(44) + \scalefactor*(\n)$) {};
            \fill[qubitcolor] (z\n) circle (3pt);
        }
    \end{pgfonlayer}
    \begin{pgfonlayer}{background}
        \draw (1,1) grid (\sizey, \sizex);
        \foreach \row in {1, ..., \sizex} {
            \draw[dashed] (0,\row) -- (1, \row);
            \draw[dashed] (\sizey,\row) -- (\sizey+1, \row);
            
        }
        \foreach \column in {1, ..., \sizey} {
            \draw[dashed] (\column,0) -- (\column, 1);
            \draw[dashed] (\column,\sizex) -- (\column, \sizex+1);
        }
        
    \end{pgfonlayer}

    \begin{pgfonlayer}{foreground}
        \foreach \n in \region{
            \draw[blue, line width=2pt] (\n) -- +(0,1);
            \draw[blue, line width=2pt] (\n) -- +(0,-1);
            \draw[blue, line width=2pt] (\n) -- +(1,0);
            \draw[blue, line width=2pt] (\n) -- +(-1,0);
            \draw[blue, line width=2pt] ($(z\n) + \scalefactor*(0,0.5)$) -- +(0,1);
            \draw[blue, line width=2pt] ($(z\n) + \scalefactor*(0,-0.5)$) -- +(0,-1);
            \draw[blue, line width=2pt] ($(z\n) + \scalefactor*(0.5,0)$) -- +(1,0);
            \draw[blue, line width=2pt] ($(z\n) + \scalefactor*(-0.5,0)$) -- +(-1,0);
        }
        
    \end{pgfonlayer}

    \

    \begin{pgfonlayer}{nodes_foreground}

        \draw [rounded corners,fill=red, red, fill opacity=0.15, line width=2.5pt] (2.7,2.7) -- (5.3,2.7) -- (5.3,3.3) -- (4.5,3.5) -- (4.5,4.5) -- (5.3,4.7) -- (5.3,5.3) -- (2.7,5.3) --cycle;
        
        \draw [rounded corners=2pt,fill=orange, orange, fill opacity=0.15, line width=1.5pt] ($(z33) + \scalefactor*(0.5,0) - \intboxscale*\scalefactor*(1.0,1.2) $) rectangle +($\intboxscale*\scalefactor*(2,2.4)+(1,0)$);
        \draw [rounded corners=2pt,fill=orange, orange, fill opacity=0.15, line width=1.5pt] ($(z33) + \scalefactor*(0,0.5) - \intboxscale*\scalefactor*(1.2,1) $) rectangle +($\intboxscale*\scalefactor*(2.4,2)+(0,1)$);

        \draw [rounded corners=2pt,fill=orange, orange, fill opacity=0.15, line width=1.5pt] ($(z35) + \scalefactor*(0.5,0) - \intboxscale*\scalefactor*(1.0,1.2) $) rectangle +($\intboxscale*\scalefactor*(2,2.4)+(1,0)$);
        \draw [rounded corners=2pt,fill=orange, orange, fill opacity=0.15, line width=1.5pt] ($(z35) - \scalefactor*(0,0.5) + \intboxscale*\scalefactor*(1.2,1) $) rectangle +($-\intboxscale*\scalefactor*(2.4,2)-(0,1)$);

        \node[coordinate] (l35) at ($(z35) - \scalefactor*(0.5,0)$) {};
        \node[coordinate] (u35) at ($(z35) + \scalefactor*(0,0.5)$) {};
        \draw [rounded corners=2pt,fill=orange, orange, fill opacity=0.15, line width=1.5pt] ($(l35) + \intboxscale*\scalefactor*(1,-1.2) $) -- ($(l35) + \intboxscale*\scalefactor*(-1,-1.2) -(1,0)$) -- +($\intboxscale*\scalefactor*(0,2.4)$) -- ($(u35) + (0,1) - \intboxscale*\scalefactor*(1.2,-1)$) -- +($+\intboxscale*\scalefactor*(2.4,0)$) -- ($(u35) + \intboxscale*\scalefactor*(1.2,-1)$)  [sharp corners] -- +($\intboxscale*\scalefactor*(-2.4,0)$) -- ($(l35) + \intboxscale*\scalefactor*(1,1.2) $) [rounded corners] -- cycle;

        \node[coordinate] (l33) at ($(z33) - \scalefactor*(0.5,0)$) {};
        \node[coordinate] (d33) at ($(z33) - \scalefactor*(0,0.5)$) {};
        \draw [rounded corners=2pt,fill=orange, orange, fill opacity=0.15, line width=1.5pt] ($(l33) + \intboxscale*\scalefactor*(1,1.2) $) -- ($(l33) + \intboxscale*\scalefactor*(-1,1.2) -(1,0)$) -- +($\intboxscale*\scalefactor*(0,-2.4)$) -- ($(d33) - (0,1) - \intboxscale*\scalefactor*(1.2,1)$) -- +($+\intboxscale*\scalefactor*(2.4,0)$) -- ($(d33) + \intboxscale*\scalefactor*(1.2,1)$) [sharp corners] -- +($\intboxscale*\scalefactor*(-2.4,0)$) -- ($(l33) + \intboxscale*\scalefactor*(1,-1.2) $) [rounded corners] -- cycle;

        \draw [rounded corners=2pt,fill=orange, orange, fill opacity=0.15, line width=1.5pt] ($(z53) - \scalefactor*(0.5,0) + \intboxscale*\scalefactor*(1.0,1.2) $) rectangle +($-\intboxscale*\scalefactor*(2,2.4)-(1,0)$);

        \node[coordinate] (r53) at ($(z53) + \scalefactor*(0.5,0)$) {};
        \node[coordinate] (d53) at ($(z53) - \scalefactor*(0,0.5)$) {};
        \node[coordinate] (u53) at ($(z53) + \scalefactor*(0,0.5)$) {};
        \draw [rounded corners=2pt,fill=orange, orange, fill opacity=0.15, line width=1.5pt] ($(u53) + \intboxscale*\scalefactor*(-1.2,-1) $) -- ($(u53) + \intboxscale*\scalefactor*(-1.2,1) + (0,1)$) -- +($\intboxscale*\scalefactor*(2.4,0)$) -- ($(r53) + \intboxscale*\scalefactor*(1,1.2) + (1,0)$) -- +($\intboxscale*\scalefactor*(0,-2.4)$) -- ($(d53)- (0,1) - \intboxscale*\scalefactor*(-1.2,1)$) -- +($\intboxscale*\scalefactor*(-2.4,0)$) -- ($(d53) + \intboxscale*\scalefactor*(-1.2,1)$) -- +($\intboxscale*\scalefactor*(2.4,0)$) --  ($(u53) + \intboxscale*\scalefactor*(1.2,-1) $) -- cycle;

        \draw [rounded corners=2pt,fill=orange, orange, fill opacity=0.15, line width=1.5pt] ($(z55) - \scalefactor*(0.5,0) + \intboxscale*\scalefactor*(1.0,1.2) $) rectangle +($-\intboxscale*\scalefactor*(2,2.4)-(1,0)$);

        \node[coordinate] (r55) at ($(z55) + \scalefactor*(0.5,0)$) {};
        \node[coordinate] (d55) at ($(z55) - \scalefactor*(0,0.5)$) {};
        \node[coordinate] (u55) at ($(z55) + \scalefactor*(0,0.5)$) {};

        \draw [rounded corners=2pt,fill=orange, orange, fill opacity=0.15, line width=1.5pt] ($(u55) + \intboxscale*\scalefactor*(-1.2,-1) $) -- ($(u55) + \intboxscale*\scalefactor*(-1.2,1) + (0,1)$) -- +($\intboxscale*\scalefactor*(2.4,0)$) -- ($(r55) + \intboxscale*\scalefactor*(1,1.2) + (1,0)$) -- +($\intboxscale*\scalefactor*(0,-2.4)$) -- ($(d55)- (0,1) - \intboxscale*\scalefactor*(-1.2,1)$) -- +($\intboxscale*\scalefactor*(-2.4,0)$) -- ($(d55) + \intboxscale*\scalefactor*(-1.2,1)$) -- +($\intboxscale*\scalefactor*(2.4,0)$) --  ($(u55) + \intboxscale*\scalefactor*(1.2,-1) $) -- cycle;
        
        \draw [rounded corners=2pt,fill=orange, orange, fill opacity=0.15, line width=1.5pt] ($(z34) - \scalefactor*(0.5,0) + \intboxscale*\scalefactor*(1.0,1.2) $) rectangle +($-\intboxscale*\scalefactor*(2,2.4)-(1,0)$);
        \draw [rounded corners=2pt,fill=orange, orange, fill opacity=0.15, line width=1.5pt] ($(z45) + \scalefactor*(0,0.5) - \intboxscale*\scalefactor*(1.2,1) $) rectangle +($\intboxscale*\scalefactor*(2.4,2)+(0,1)$);
        \draw [rounded corners=2pt,fill=orange, orange, fill opacity=0.15, line width=1.5pt] ($(z43) - \scalefactor*(0,0.5) + \intboxscale*\scalefactor*(1.2,1) $) rectangle +($-\intboxscale*\scalefactor*(2.4,2)-(0,1)$);
        
        \draw [rounded corners=2pt,fill=orange, orange, fill opacity=0.15, line width=1.5pt] ($(z44) + \scalefactor*(0.5,0) - \intboxscale*\scalefactor*(1.0,1.2) $) rectangle +($\intboxscale*\scalefactor*(2,2.4)+(1,0)$);
        \draw [rounded corners=2pt,fill=orange, orange, fill opacity=0.15, line width=1.5pt] ($(z44) - \scalefactor*(0.5,0) + \intboxscale*\scalefactor*(1.0,1.2) $) rectangle +($-\intboxscale*\scalefactor*(2,2.4)-(1,0)$);
        \draw [rounded corners=2pt,fill=orange, orange, fill opacity=0.15, line width=1.5pt] ($(z44) + \scalefactor*(0,0.5) - \intboxscale*\scalefactor*(1.2,1) $) rectangle +($\intboxscale*\scalefactor*(2.4,2)+(0,1)$);
        \draw [rounded corners=2pt,fill=orange, orange, fill opacity=0.15, line width=1.5pt] ($(z44) - \scalefactor*(0,0.5) + \intboxscale*\scalefactor*(1.2,1) $) rectangle +($-\intboxscale*\scalefactor*(2.4,2)-(0,1)$);

    \end{pgfonlayer}

\end{tikzpicture}

%% file: tikz_partition_overlap.tex
\begin{tikzpicture}[scale=0.8]

    \pgfmathsetmacro{\sizex}{4}
    \pgfmathsetmacro{\sizey}{9}
    \pgfmathsetmacro{\sizera}{7.5}
    \pgfmathsetmacro{\sizerb}{7.5}
    \pgfmathsetmacro{\tilingoffset}{0.1}
    
    \definecolor{strbgcolor}{HTML}{86abbf}
    \definecolor{plqbgcolor}{HTML}{ab8d67}
    \definecolor{strfgcolor}{HTML}{00A6FF}
    \definecolor{plqfgcolor}{HTML}{FF9000}
    \definecolor{qubitcolor}{HTML}{000000}

    \pgfdeclarelayer{background}
    \pgfdeclarelayer{foreground}
    \pgfdeclarelayer{nodes}
    \pgfdeclarelayer{nodes_foreground}
    \pgfsetlayers{background, foreground, nodes, nodes_foreground}

    \begin{pgfonlayer}{nodes}
        \foreach \row in {1,...,\sizex} {
            \foreach \column in {1,..., \sizey}{
                \node[coordinate] (\column\row) at (\column,\row) {};
            }
        }
    \end{pgfonlayer}
    \begin{pgfonlayer}{background}
        \draw[line width=1pt, black!30] (0.2,-0.8) grid[step=1] +(\sizey+1.6,\sizex+2.6);
    \end{pgfonlayer}

    \begin{pgfonlayer}{foreground}
        \foreach \row in {-2,...,\sizex} {
            \foreach \column in {-1,..., \sizey}{
                \filldraw [black!40!white] ($(\column, \row) + (0.35,0.3) + (1,1)$) circle (1.5pt);
                \filldraw [black!40!white] ($(\column, \row) + (0.7,0.4) + (1,1)$) circle (1.5pt);
                \filldraw [black!40!white] ($(\column, \row) + (0.25,0.7) + (1,1)$) circle (1.5pt);
            }
        }

        \draw[red, dash pattern={on 8pt off 6pt}, dash phase=4pt, line width=1.5pt] (1,1) rectangle +(\sizera-2.5,\sizex);
        \draw[fill=red, red, fill opacity=0.05, line width=1.5pt] (1,1) rectangle +(\sizera,\sizex);
        \draw[fill=blue, blue, fill opacity=0.05, line width=1.5pt] (\sizey+1,1) rectangle +(-\sizerb,\sizex);
        \draw [rounded corners=3pt,fill=orange, orange, fill opacity=0.4, line width=2pt] (5+\tilingoffset,1+\tilingoffset) rectangle +(2-2*\tilingoffset, 2-2*\tilingoffset);
        \draw [rounded corners=0pt, dash pattern={on 3pt off 4pt}, green!60!black, line width=3pt] (4,0) rectangle +(4, 4);

        \draw [decorate,decoration={brace,amplitude=5pt,mirror,raise=6ex}](1,1) -- +(\sizera,0) node[midway,yshift=-6.5ex-1em, xshift=-0.3em,  rectangle, fill=white, inner sep=1.5pt]{$R_1$};
        \draw [decorate,decoration={brace,amplitude=5pt,mirror,raise=0.5ex}](1,1) -- +(\sizera-2.5,0) node[midway,yshift=-1.1ex-1em,  rectangle, fill=white, inner sep=0pt]{$R_1^{-}$};
        \draw [decorate,decoration={brace,amplitude=5pt,mirror,raise=6ex}](\sizera+1,1) -- (\sizey+1,1) node[midway,yshift=-6.5ex-1em, xshift=-0.2em, rectangle, fill=white, inner sep=1pt]{$R_2\setminus R_1$};

        \node [] at (\sizera-1.1,2.5) {$r$};

        \draw [decorate,decoration={brace,amplitude=5pt,mirror,raise=6ex}](\sizey+1,\sizex+1) -- +(-\sizerb,0) node[midway,yshift=6.5ex+1em]{$R_2$};

    \end{pgfonlayer}

\end{tikzpicture}

%% file: tikz-martingale.tex
\begin{tikzpicture}[scale=1]

    \draw[orange, fill=orange, fill opacity=0.1, line width=1.5pt] (-1, -1) rectangle (8, 5);
    \fill[white] (0, 0) rectangle (7, 4);
    \node at (-0.5, 4.5) {$C$};

    \draw[blue, fill=blue, fill opacity=0.1, line width=1.5pt] (0, 3.1) rectangle (2, 4);
    \node at (0.5, 3.6) {$U_1$};
    \node at (1.5, 3.6) {$V_1$};
    \draw[decoration={brace,mirror,raise=5pt},decorate] (2,4) -- node[above=6pt] {$\lceil \sqrt{L}\rceil$} (1,4);
    \draw[red, fill=red, fill opacity=0.1, line width=1.5pt] (1, 3.1) rectangle (7, 4);
    \node at (4.4, 3.6) {$W_1$};

    \draw[blue, fill=blue, fill opacity=0.1, line width=1.5pt] (0, 1.7) rectangle (3, 2.9);
    \draw[red, fill=red, fill opacity=0.1, line width=1.5pt] (2, 1.7) rectangle (7, 2.9);

    \node at (1, 2.3) {$U_2$};
    \node at (2.5, 2.3) {$V_2$};
    \node at (4.9, 2.3) {$W_2$};

    \draw[blue, fill=blue, fill opacity=0.1, line width=1.5pt] (0, 0) rectangle (5, 1);
    \draw[red, fill=red, fill opacity=0.1, line width=1.5pt] (4, 0) rectangle (7, 1);

    \node at (2.1, 0.4) {$U_{\lfloor \sqrt[3]{L} \rfloor}$};
    \node at (4.5, 0.4) {$V_{\lfloor \sqrt[3]{L} \rfloor}$};
    \draw[decoration={brace,mirror,raise=5pt},decorate] (4,0) -- node[below=6pt] {$\lceil \sqrt{L}\rceil$} (5,0);
    \node at (6, 0.4) {$W_{\lfloor \sqrt[3]{L} \rfloor}$};

    \fill[white, draw=white, line width=1.5pt] (-0.01,3.2) decorate[decoration={snake,amplitude=1.5,segment length=110, post length=0}]{ --  (7.01,3.2) }  -- (7.01,2.8) decorate[decoration={snake,amplitude=1.5,segment length=99, mirror}]{ -- (-0.01,2.8)} --cycle;

    \fill[white, draw=white, line width=1.5pt] (-0.01,1.8) decorate[decoration={snake,amplitude=1.5,segment length=108, post length=0}]{ --  (7.01,1.8) }  -- (7.01,0.9) decorate[decoration={snake,amplitude=1.5,segment length=95}]{ -- (-0.01,0.9)} --cycle;

    \draw[thick, dotted] (1, 1.65)-- (2.1, 1.05);
    \draw[thick, dotted] (4.9, 1.65)-- (6, 1.05);
\end{tikzpicture}